\documentclass[examplefnt,biber]{nowfnt2} %

\usepackage[utf8]{inputenc}

\usepackage{amsmath}
\usepackage{amsfonts}
\usepackage{amssymb}
\usepackage{enumitem}

\newcommand{\D}{\displaystyle}

\newcommand{\supn}{^{(n)}}
\newcommand{\tod}{\Rightarrow}
\usepackage[ruled,vlined]{algorithm2e}
\usepackage{tablefootnote}

\usepackage{algpseudocode}

\newcommand{\algostate}{\State\hspace{-.3cm}}
\newcommand{\algoindent}{\hspace{.3cm}}
\newcommand{\supth}{^{\textrm{th}}}
\usepackage{url}

\usepackage{float}
\usepackage{booktabs}
\providecommand{\tabularnewline}{\\}
\floatstyle{ruled}
\newfloat{algorithm}{tbp}{loa}
\providecommand{\algorithmname}{Algorithm}
\floatname{algorithm}{\protect\algorithmname}
\extrafloats{100}

\usepackage{%
	pgf,%
	pgfplots,%
	tikz,%
}
\usepgfplotslibrary{patchplots}
\newcommand{\EQN}[1]{\begin{equation}#1\end{equation}}
\newcommand{\E}[1]{\mathbb{E}\left [#1\right]}
\newcommand{\set}[1]{\left\{#1\right\}}

\title{Modeling and Optimization of Latency in Erasure-coded Storage Systems}

\maintitleauthorlist{
	Vaneet Aggarwal \\
	Purdue University \\
vaneet@purdue.edu
	\and
	Tian Lan\\
	George Washington University\\
	tlan@gwu.edu
}

\issuesetup
{%
	copyrightowner={A.~Heezemans and M.~Casey},
	volume        = xx,
	issue         = xx,
	pubyear       = 2020,
	isbn          = xxx-x-xxxxx-xxx-x,
	eisbn         = xxx-x-xxxxx-xxx-x,
	doi           = 10.1561/XXXXXXXXX,
	firstpage     = 1, %
	lastpage      = 138
}

\usepackage{mwe}

\author[1]{Aggarwal, Vaneet}
\author[2]{Lan, Tian}

\affil[1]{Purdue University}
\affil[2]{George Washington University}

\begin{document}

	\makeabstracttitle
	
	\begin{abstract}
		As consumers are increasingly engaged in social networking and E-commerce activities, businesses grow to rely on Big Data analytics for intelligence, and traditional IT infrastructures continue to migrate to the cloud and edge, these trends cause distributed data storage demand to rise at an unprecedented speed. Erasure coding has seen itself quickly emerged as a promising technique to reduce storage cost while providing similar reliability as replicated systems, widely adopted by companies like Facebook, Microsoft and Google. However, it also brings new challenges in characterizing and optimizing the access latency when erasure codes are used in distributed storage. The aim of this monograph is to provide a review of recent progress (both theoretical and practical) on systems that employ erasure codes for distributed storage.
	
	In this monograph, we will first identify the key challenges and taxonomy of the research problems and then give an overview of different approaches that have been developed to quantify and model latency of erasure-coded storage. This includes recent work leveraging MDS-Reservation, Fork-Join, Probabilistic, and Delayed-Relaunch scheduling policies, as well as their applications to characterize access latency (e.g., mean, tail, asymptotic latency) of erasure-coded distributed storage systems. We will also extend the problem to the case when users are streaming videos from erasure-coded distributed storage systems. Next, we bridge the gap between theory and practice, and discuss lessons learned from prototype implementation. In particular, we will discuss exemplary implementations of erasure-coded storage, illuminate key design degrees of freedom and tradeoffs, and summarize remaining challenges in real-world storage systems such as in content delivery and caching. Open problems for future research are discussed at the end of each chapter.

\end{abstract}

\chapter{Introduction}\label{chp:1}

In this Chapter, we will introduce the problem in Section \ref{model}. This will be followed by the key challenges in the problem in Section \ref{sec:chall}. Section \ref{sec:taxo} explains the different approaches for the problem considerd in this monograph. Section \ref{mono_out} gives the outline for the remaining chapters, and Chapter \ref{sec:intro_notes} provides additional notes. 

\section{Erasure Coding in Distributed Storage}\label{model}

Distributed systems such as Hadoop, AT$\&$T Cloud Storage, Google File System and Windows Azure have evolved to support different types of erasure codes, in order to achieve the benefits of improved storage efficiency while providing the same reliability as replication-based schemes \citep{balaji2018erasure}. Various erasure code plug-ins and libraries have been developed in storage systems like Ceph \citep{Ceph,Yu-TON16}, Tahoe \citep{Yu_TON},  Quantcast (QFS) \citep{ovsiannikov2013quantcast}, and Hadoop (HDFS) \citep{B182}. 

We consider a data center consisting of $m$ heterogeneous servers\footnote{We will use storage nodes and storage servers interchangeably throughout this monograph.}, denoted by $\mathcal{M}=\{1,2,\ldots,m\}$, called storage nodes. To distributively store a set of $r$ files, indexed by $i=1,\ldots,r$, we partition each file $i$ into $k_i$ fixed-size chunks\footnote{While we make the assumption of fixed chunk size here to simplify the problem formulation, the results can be easily extended to variable chunk sizes. Nevertheless, fixed chunk sizes are indeed used by many existing storage systems \citep{DPR04,AJX05,LC02}.} and then encode it using an $(n_i,k_i)$ MDS erasure code to generate $n_i$ distinct chunks of the same size for file $i$. The encoded chunks are assigned to and stored on $n_i$ distinct storage nodes %
  to store file $i$. Therefore, each chunk is placed on a different node to provide high reliability in the event of node or network failures. %

The use of $(n_i,k_i)$ MDS erasure code allows the file to be reconstructed from any subset of $k_i$-out-of-$n_i$ chunks, whereas it also introduces a redundancy factor of $n_i/k_i$. %
For known erasure coding and chunk placement, we shall now describe a queueing model of the distributed storage system. We assume that the arrival of client requests for each file $i$ form an independent Poisson process with a known rate $\lambda_i$. We consider chunk service time $\mathbb{X}_{j}$ of node $j$ with {\em arbitrary distributions}, whose statistics can be obtained inferred from existing work on network delay \citep{AY11,WK} and file-size distribution \citep{D11,PT12}. We note that even though exponential service time distribution is common, realistic implementation in storage systems show that this is not a practical assumption \citep{chen2014queueing,Yu_TON}, where Amazon S3 and Tahoe storage systems are considered. Both these works points towards shifted exponential service times being a better approximation (an example service time distribution from realistic system is depicted in Fig. \ref{fig:service_dist}), while other distributions may be used for better approximation.  Let $\mathbb{Z}_{j}(\tau)=\mathbb{E}\left[e^{\tau\mathbb{X}_{j}}\right]$
be the moment generating function of $\mathbb{X}_{j}$. Under MDS codes, each file $i$ can be retrieved from any $k_i$ distinct nodes that store the file chunks. We model this by treating each file request as a {\em batch} of $k_i$ chunk requests, so that a file request is served when all $k_i$ chunk requests in the batch are processed by distinct storage nodes. All requests are buffered in a common queue of infinite capacity.

We now introduced the definition of MDS queues according to the system model. 
\begin{definition}
	An MDS queue is associated to four sets of parameters $\{m,r\}$, $\{(n_i, k_i): \  i=1,2,\ldots,r\}$, $\{ \lambda_i: \ i=1,2,\ldots,r\}$, and $\{ \mu_j: \  j=1,2,\ldots,m\}$ satisfying i) There are $m$ servers and $r$ files; ii) File-$i$ requests arrive in batches of $k_i$ chunk requests each; iii) Each batch of $k_i$ chunk requests can be processed by any subset of $k_i$ out of $n_i$ distinct servers; iv) These batches arrive as a Poisson process with a rate of $\lambda_i$; and v) The service time for a chunk request at any server is random and follows some known distribution with mean $\mu_j$, and is independent of the arrival and service times of all other requests. %
\end{definition}

\section{Key Challenges in Latency Characterization}\label{sec:chall}

An exact analysis of this MDS queue model is known to be an open problem. The main challenge comes from the fact that since each file request needs to be served by $k$ distinct servers, a Markov-Chain representation of the MDS queue must encapsulate not only the number of file and chunk requests waiting in the shared buffer, but also the processing history of each active file requests to meet such requirement in future schedules. Let $b$ be the number of current file requests in the system and can take values in $\{0, 1,2, \ldots \}$. The Markov-Chain representation could have $\Omega{b^k}$ states, which becomes infinity in at least $k$ dimensions \citep{MDS_queue}. This is extremely difficult to analyze as the transitions along different dimensions are tightly coupled in MDS queues.

The challenge can be illustrated by an abstracted example shown in Fig.~\ref{fig:sysmodel}. We consider two files, each partitioned into $k=2$ blocks of equal size and encoded using maximum distance separable (MDS) codes. Under an $(n,k)$ MDS code, a file is encoded and stored in $n$ storage nodes such that the chunks stored in any $k$ of these $n$ nodes suffice to recover the entire file. There is a centralized scheduler that buffers and schedules all incoming requests. For instance, a request to retrieve file $A$ can be completed after it is successfully processed by 2 distinct nodes chosen from $\{1,2,3,4\}$ where desired chunks of $A$ are available. Due to shared storage nodes and joint request scheduling, delay performances of the files are highly correlated and are collectively determined by control variables of both files over three dimensions: (i) the scheduling policy that decides what request in the buffer to process when a node becomes available, (ii) the placement of file chunks over distributed storage nodes, and (iii) erasure coding parameters that decide how many chunks are created. The latency performances of different files are tightly entangled. While increasing erasure code length of file B allows it to be placed on more storage nodes, potentially leading to smaller latency (because of improved load-balancing) at the price of higher storage cost, it inevitably affects service latency of file A due to resulting contention and interference on more shared nodes.

\begin{figure}[!thbp]
\begin{center}
\scalebox{0.59}{\includegraphics[draft=false]{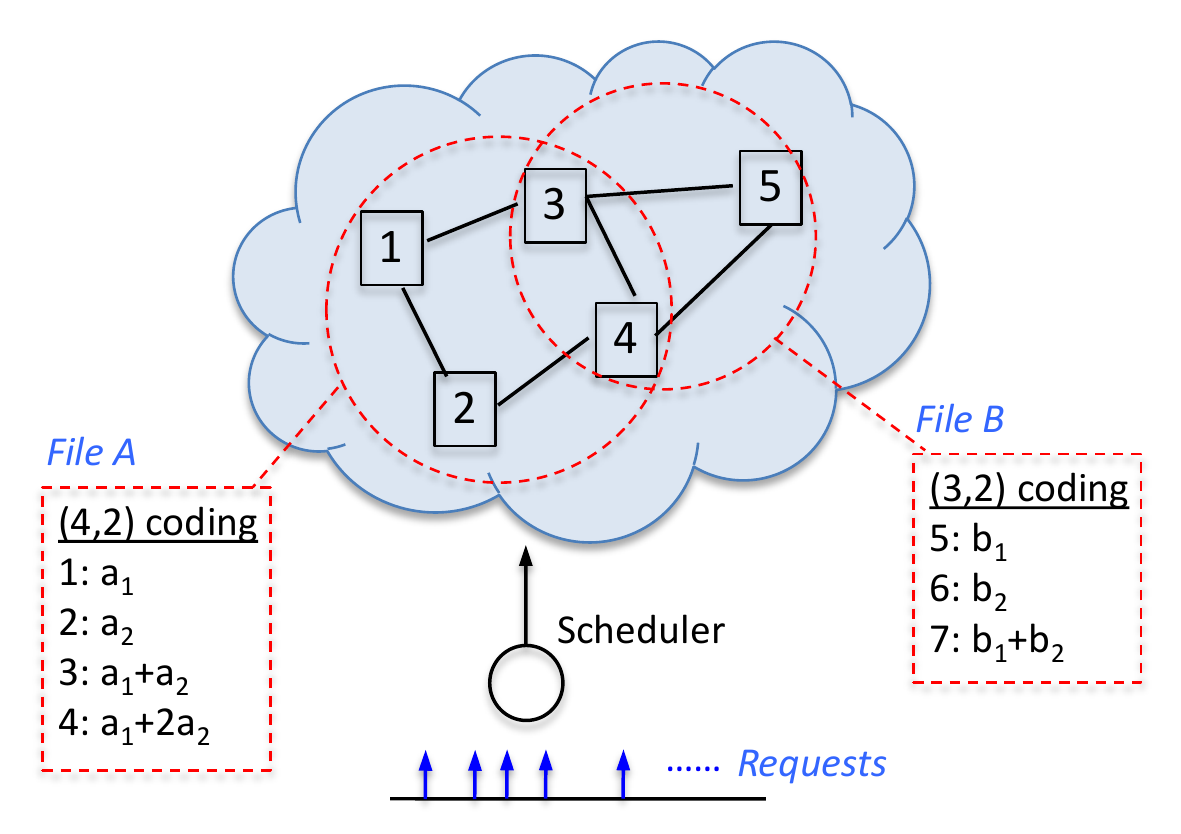}}
\vspace{-3mm}
\caption{An erasure-coded storage of 2 files, which partitioned into 2 blocks and encoded using $(4,2)$ and $(3,2)$ MDS codes, respectively. Resulting file chunks are spread over 5 storage nodes. Any file request must be processed by 2 distinct nodes that have the desired chunks. Nodes $3,4$ are shared and can process requests for both files.}
\label{fig:sysmodel}
\end{center}
\vspace{-.35in}
\end{figure}

In Figure \ref{fig:sysmodel}, files $A$ and $B$ are encoded using $(4,2)$ and $(3,2)$ MDS codes, respectively, file $A$ will have chunks as $A_1$, $A_2$, $A_3$ and $A_4$, and file $B$ will have chunks $B_1$, $B_2$ and $B_3$. As depicted in Fig.~\ref{fig:sysmodelintro}, each file request comes in as a batch of $k_i=2$ chunk requests, e.g., $(R_1^{A,1},R_1^{A,2})$, $(R_2^{A,1},R_2^{A,2})$, and $(R_1^{B,1},R_1^{B,2})$, where $R_{i}^{A,j}$, denotes the $i$th request of file $A$, $j=1, 2$ denotes the first or second chunk request of this file request. Denote the five nodes (from left to right) as servers 1, 2, 3, 4, and 5, and we initialize 4 file requests for file $A$ and 3 file requests for file $B$, i.e., requests for the different files have different arrival rates.  The two chunks of one file request can be any two different chunks from $A_1$, $A_2$, $A_3$ and $A_4$ for file $A$ and  $B_1$, $B_2$ and $B_3$ for file $B$. Due to chunk placement in the example, any 2 chunk requests in file A's batch must be processed by 2 distinct nodes from $\{1,2,3,4\}$, while 2 chunk requests in file B's batch must be served by 2 distinct nodes from $\{3,4,5\}$. Suppose that the system is now in a state depicted by Fig.~\ref{fig:sysmodelintro}, wherein the chunk requests $R_1^{A,1}$, $R_2^{A,1}$, $R_1^{A,2}$, $R_1^{B,1}$, and $R_2^{B,2}$ are served by the 5 storage nodes, and there are 9 more chunk requests buffered in the queue. Suppose that node 2 completes serving chunk request $R_2^{A,1}$ and is now free to serve another request waiting in the queue. Since node 2 has already served a chunk request of batch $(R_2^{A,1},R_2^{A,2})$ and node 2 does not host any chunk for file B, it is not allowed to serve either $R_2^{A,2}$ or $R_2^{B,j},R_3^{B,j}$ where $j=1,2$ in the queue. One of the valid requests, $R_3^{A,j}$ and $R_4^{A,j}$, will be selected by a scheduling algorithm and assigned to node 2. We denote the scheduling policy that minimizes average expected latency in such a queuing model as {\em optimal scheduling}.

\vspace{-3mm}
\begin{figure}[!thbp]
	\begin{center}
		\scalebox{0.32}{\includegraphics[trim=0in .3in 5in 0in, clip, draft=false]{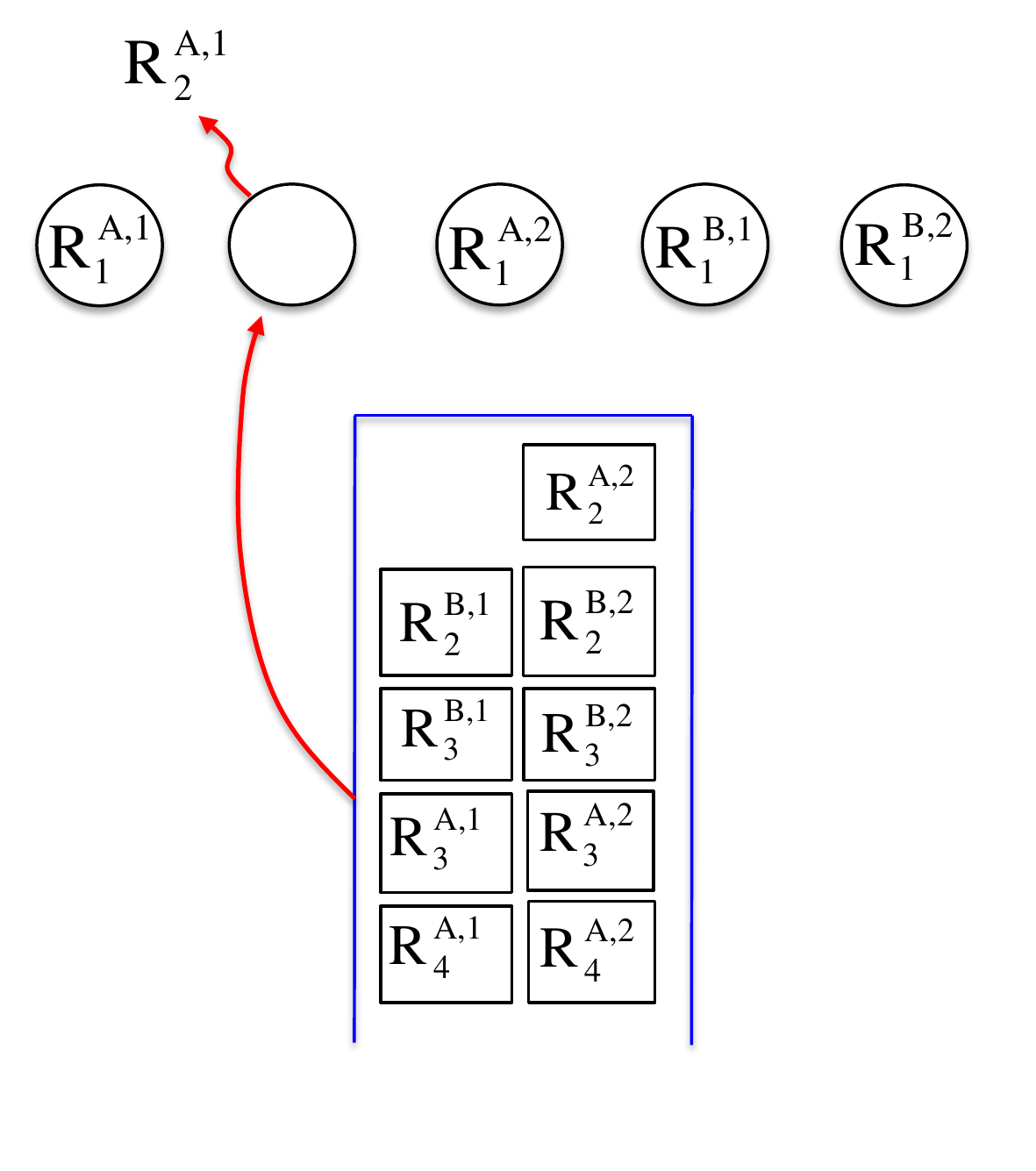}}
		\caption{Functioning of an optimal scheduling policy.}
		\label{fig:sysmodelintro}
	\end{center}
	\vspace{-.2in}
\end{figure}

\begin{definition} {\em (Optimal scheduling)} An optimal scheduling policy (i) buffers all requests in a queue of infinite capacity; (ii) assigns at most 1 chunk request from a batch to each appropriate node, and (iii) schedules requests to minimize average latency if multiple choices are available.
\end{definition}

An exact analysis of optimal scheduling is extremely difficult. Even for given erasure codes and chunk placement, it is unclear what scheduling policy leads to minimum average latency of multiple heterogeneous files. For example, when a shared storage node becomes free, one could schedule either the earliest valid request in the queue or the request with scarcest availability, leading to different implications on average latency.

\section{Problem Taxonomy}\label{sec:taxo}

Given that the optimal scheduling is hard to evaluate, many scheduling strategies that aim to provide latency bounds have been explored. This monograph aims to provide a detailed description of these strategies, and the latency analysis for these strategies. The different scheduling strategies considered in this monograph are:
\begin{enumerate}
	\item MDS-Reservation Scheduling: Consider a single file request. In this approach, the first $t$ requests, the request is added in buffer. Each server on finishing the current task, goes through the buffer in order to find a task it can process (batch from which it has not processed a task). The file requests after $t+1$ can move ahead in the buffer only when $k$ of its chunk requests can move forward together, and the request is assigned to those servers where it can move forward together, for a given parameter $t$. 
	\item Fork-Join Scheduling: In this approach, the file request for file $i$  is sent to all $n_i$ servers, and the request is complete when $k_i$ coded chunks are received. The remaining $n_i-k_i$ requests are then cancelled. 
	\item Probabilistic Scheduling: In this approach, each file request for file $i$ is randomly dispatched to $k_i$ out of $n_i$ storage nodes that store the coded chunks of the file. 
	\item Delayed Relaunch Scheduling: In this approach, the request for file $i$ is first sent to $n_{i,0}$ servers using probabilistic scheduling, and when the request is completed from $\ell_{i,0}$ servers, the request is sent to the remaining $n_i-n_{i,0}$ servers. After $k_i$ chunks are received, the request is cancelled from remaining $n_i-k_i$ servers. 
\end{enumerate}

Even though the implementation of these scheduling strategies in simulators and practical systems may be straightforward, the analysis and optimization of latency are quite the opposite. In many cases, only bounds and asymptotic guarantees can be provided using these queue models. Further, a number of assumptions are commonly considered in the literature to ensure tractability of the analysis. These assumptions include homogeneous files (i.e., all files are of equal size and encoded using the same $(n,k)$ MDS code); homogeneous placement (i.e., there are $n$ servers and each file has exactly one chunk placed on each server); homogeneous servers (i.e., all servers have i.i.d. service time with mean $\mu$); exponential service time distribution (i.e., all server has exponentially distributed service times). We summarize the assumption made by different strategies in each chapter in~Table~\ref{tbl_intro0}.

\begin{table}[htbp]
\begin{center}
	\begin{tabular}{|c|c|c|c|c|}
		\hline
		& MDS- & Fork- & Probabilitic  & Delayed   \\
		& Reservation & Join &  Scheduling &  Relaunch  \\
		\hline
		{\small	Homogenous}	& Yes & No & No & Yes  \\
		{\small	Files}	&  &  &  &   \\
		\hline
		{\small Homogenous}		&  Yes & Yes & No & Yes \\
		{\small Placement}		&  &  &  &  \\
		\hline
		{\small Homogeneous}		& Yes & Yes & No & Yes  \\
		{\small Servers}		&  &  &  &   \\
		\hline
		{\small Exponential}		&  Yes & No & No & No\tablefootnote{Queueing analysis is not applicable to delayed relaunch.} \\
		{\small Service Time}		&  &  &  &  \\
		\hline
	\end{tabular}
	\end{center}
	\caption{Assumptions considered in the analysis of different scheduling strategies.}\label{tbl_intro0}
\end{table}

We compare the different strategies for different service rates at high arrival rates. We consider single file request, $r=1$, and thus index $i$ for files is suppressed. Also, we assume that all $m=n$ servers are homogenous. If the service times are deterministic, fork-join scheduling sends request to all $n$ servers and finish at the same time. Thus, the strategy wastes unncessary time at the $n-k$ servers leading to non-optimal stability region. In contrast, probabilistic scheduling can use the probabilities of selection of different servers as uniform, and can be shown to achieve optimal stability region. Further, delayed relaunch scheduling has probabilistic scheduling as special case with $n_{0} = \ell_{0}=k$, and thus can achieve optimal stability region. For MDS-Reservation scheduling, unless $n$ is a multiple of $k$, there would be wasted time at some servers and thus will not have optimal stability region. For exponential service times, MDS-Reservation Scheduling would hold on request scheduling at certain servers, and thus not achieve optimal stability region. The other strategies will have optimal stability region. Thus, the probabilistic scheduling and the delayed relaunch scheduling are optimal in terms of stability region for both service distributions, and can indeed be shown to achieve optimal stability region for general service distribution. 

We further note that the delayed relaunch scheduling has fork-join scheduling as a special case when $n_{i,0}=n_i$ and $\ell_{i,0}=k_i$, and has probabilistic scheduling as a special case when $n_{i,0} = \ell_{i,0}=k_i$, and thus give a more tunable approach than the two scheduling approaches.

\begin{table}[htbp]
	\resizebox{\columnwidth}{!}{%
	\begin{tabular}{|c|c|c|c|c|}
		\hline
		& MDS- & Fork- & Probabilitic  & Delayed   \\
		& Reservation & Join &  Scheduling &  Relaunch  \\
		\hline
		{\small	Optimal Homogenous}	& No & Exponential & General & General  \\
		{\small	Stability Region}	&  &  &  &   \\
		\hline
		{\small Queuing}		&  Yes & Yes & Yes & No \\
		{\small Analysis}		&  &  &  &  \\
		\hline
		{\small Analysis for }		& No & Yes & Yes & No  \\
		{\small general distribution}		&  &  &  &   \\
		\hline
		{\small Closed Form}		&  No & Yes & Yes & N/A\tablefootnote{Queueing analysis is not applicable to delayed relaunch.} \\
		{\small Expressions }		&  &  &  &  \\
		\hline
		{\small Asymptotic}		&  No & No & Yes & Yes \\
		{\small Optimality }		&  &  &  &  \\
		\hline
			{\small Tail}		&  No & No & Yes & No \\
		{\small Characterization }		&  &  &  &  \\
		\hline
	\end{tabular}
}
	\caption{The different regimes for the known results of the different scheduling algorithms}\label{tbl_intro}
\end{table}

In Table \ref{tbl_intro}, we will describe the different cases where the analysis of these algorithms have been studied. The first line is for the a single file and homogenous servers. As mentioned earlier, the MDS-Reservation scheduling does not achieve optimal stability region for both the scenarios of determistic and exponential service times. The Fork-Join scheduling has been shown to achieve the optimal stability region only for exponential service times, while uniform probabilistic scheduling achieves optimal stability region for general service distributions. The second line indicates that the queueing analysis to find upper bound on latency using the proposed algorithms have been studied for the first three schemes, while for delayed relaunch, no non-trivial queueing analysis exists, while has been studied for single file request in absence of queue. 
The next line considers whether there is latency analysis (upper bounds) for general service-time distribution, which are not available for the case of MDS-Reservation and Delayed Relaunch scheduling. The fourth line indicates whether closed-form expressions for the latency bounds exist in the queueing analysis, which is not true for MDS-Reservation scheduling. Since there is no queueing analysis for delayed relaunch, N/A is marked. In the next line, we note that there are asymptotic guarantees for probabilistic scheduling with exponential service times in two regimes. The first is in the case of homogenous servers with $m=n$ and single file, where $n$ goes large. The second is where the file sizes are heavy-tailed. Since the delayed relaunch scheduling is a generalization of probabilistic scheduling, it inherits the guarantees. The last line indicates whether the analysis exist for tail latency, which is the probability that the latency is greater than the threshold which exist for probabilistic scheduling.

As an example, we consider a shifted exponential distribution for service times, ${\rm Sexp}(\beta,\alpha)$ as defined in \eqref{eq:Sexp}. The parameters are different for different servers, which for $m=12$ servers are given in Table \ref{tab:Storage-Nodes-Parameters1}. We consider homogeneous files with $k_i=7$ and $n_i=12$. In order to do the simulations, requests arrive for $10^4$ seconds, and their service time is used to calculate the latency based on the different scheduling strategies.  For varying arrival rate, we compare the three strategies - MDS-Reservation(1000), Fork-Join scheduling, and Probabilistic scheduling. Since optimized Delayed Relaunching includes the Fork-Join and Probabilistic scheduling approaches, we do not compare this.  Since the Probabilistic scheduling have the probabilities of choosing the servers as random, we run ten choices of the probability terms using a uniform random variable between zero and 1 and normalization and choose the best one among these. %
Note that even though Fork-Join queues have not been analyzed for heterogeneous servers, the results indicate the simulated performance.  The simulation results are provided in Fig. \ref{fig:comp_intro}. We note that the MDS-Reservation and the fork-join queue does not achieve the optimal stable throughput region and it can be seen that the mean latency starts diverging faster. Further, we note that the probabilistic scheduling performs better than Fork-Join scheduling for all arrival rates in this system. %

\begin{figure}
\begin{center}
	\includegraphics[trim=1.1in .9in 1.1in .2in, clip, width=.5\textwidth, angle =90]{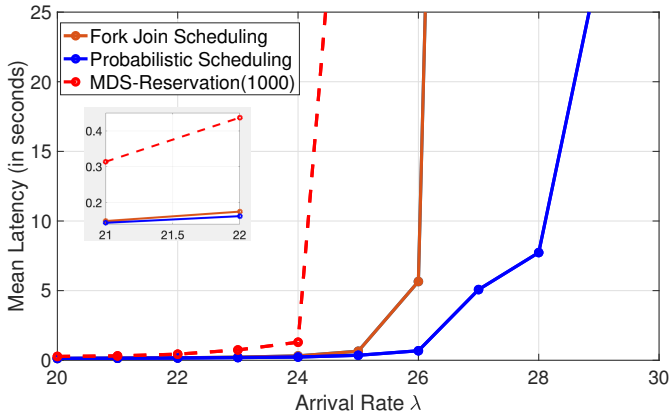}
\end{center}
\caption{Comparison of the different strategies in simulation. We note that probabilistic scheduling outperforms the other strategies with the considered parameters. }	\label{fig:comp_intro}
\end{figure}

\section{Outline of the Monograph}\label{mono_out}
In the remainder of the monograph, we will describe the four considered scheduling strategies in Chapters 2, 3, 4, and 5, respectively. The corresponding results for the latency characetrization, approximations, generalizations, and asymptotic optimality guarantees will be provided. In Chapter 6, we will demonstrate an extension of using the approach of erasure-coded file storage to erasure-coded video storage, where the metric of importance for the video is the stall duration rather than the latency. The analysis of upper bound of stall duration using probabilistic scheduling is provided. Finally, Chapter 7 demonstrates the insights of the approaches on prototype evaluation in realistic systems. We also discuss open problems at the end of each chapter to illuminate directions for future work.
\section{Notes}\label{sec:intro_notes}
In this monograph, we consider using an maximal distance separable (MDS) code for erasure-coded storage system. The problem to evaluate scheduling algorithms, and the latency analysis will be analyzed in the following chapters, with discussion on related papers. 

In addtion to latency, a key importance of the use of erasure codes is to deal with node failures. With node failures, one important aspect is the reliability of the storage system, which is measured through the mean time to data loss. Different models of failure and repair have been considered in \citep{angus1988computing,chen1994raid}. 

Another aspect of node failure is to minimize the amount of data transferred to repair a failed node. In order to consider this problem, a novel approach of regenerating codes was proposed in \citep{dimakis2010network}. Functional repair has been proposed which aims to repair a failed node with a content that satisfies similar properties as the original code. The problem has connections with network coding, and such frameworks have been used to provide efficient code designs. The regenerating codes allow for a tradeoff between the storage capacity at the node and the amount of data transferred from other $d\ge k$ nodes in an $(n,k)$ erasure-coded system. In many realistic systems, exact repair, where the failed node is repaired with the exact same content, is needed. Construction of erasure codes with exact repair guarantee has been widely studied \citep{suh2011exact,MDS_12,MDS_13,tian2015layered,MDS_14}. Regenerating codes have also been used to evaluate the mean time to data loss of the storage systems \citep{aggarwal2014distributed}. The regenerating codes includes a minimum storage regenerating (MSR) point, where $1/k$ of the file is placed at each node, and thus is best for the latency. However, any of the point of the code can be used for the analysis in this work with corresponding increased amount of data from each of the $k$ nodes based in the increased storage.

\chapter{MDS-Reservation Scheduling Approach}
\label{chp:MDS}
In this Chapter, we introduce the model of MDS-Reservation($t$) queues in Section~\ref{sec:mds_sch}, which was first proposed in \citep{MDS_queue}. It allows us to develop an upper bound on the latency of MDS queues and characterize the stability region in Section~\ref{sec:MDS:upper}. We will also develop a lower bound for MDS-Reservation($t$) queues using $M/M/n$ queues in Section~\ref{sec:MDS:lower} and investigate the impact of redundant requests in Section~\ref{sec:MDS:redundant}. Sections~\ref{sec:MDS:sim} and \ref{sec:MDS:note} contain simulation results and notes on future directions.

\section{MDS-Reservation Queue}
\label{sec:mds_sch}

For this chapter, we consider a homogeneous erasure-coded storage system, where the files have identical sizes and are encoded using the same MDS code with parameters $(n,k)$. There are $n$ identical storage serves, each storing exactly one chunk of each file. Incoming file requests follow a Poisson process and are independent of the state of the system, and the chunk service times on storage nodes have an {\em i.i.d.} exponential distribution. Under an MDS code, a file can be retrieved by downloading chunks from any $k$ of the $n$ servers. Thus, a file request is considered as served when all $k$ of its chunk request have been scheduled and completed service.

\begin{algorithm}[h]
	\caption{MDS-Reservation(t) Scheduling Policy}
	\begin{algorithmic}
		\algostate \textbf{On} arrival of a batch
		\algostate \algoindent \textbf{If}  {buffer has strictly fewer than t batches}
		\algostate \algoindent \algoindent  Assign jobs of new batch to idle servers
		\algostate \algoindent Append remaining jobs of batch to end of buffer
		\algostate \textbf{On} departure of job from a server (say, server $s$)
		\algostate \algoindent {\medmuskip=0\medmuskip \thickmuskip=0\thickmuskip Find $\hat{i} =\min\{i\geq 1:$ $s$ has not served job of $i\supth$ waiting batch$\}$}
		\algostate \algoindent Let $b_{t+1}$ be the $(t+1)\supth$ waiting batch (if any)
		\algostate \algoindent \textbf{If} {$\hat{i}$ exists \& $\hat{i} \leq t$}
		\algostate \algoindent \algoindent  Assign a job of $\hat{i}\supth$ waiting batch to $s$
		\algostate \algoindent \algoindent \textbf{If} {$\hat{i}=1$ \& the first waiting batch had only one job in the buffer \& $b_{t+1}$ exists}
		\algostate \algoindent \algoindent \algoindent To every remaining idle server, assign a job from batch $b_{t+1}$
	\end{algorithmic}
	\label{alg:mdst}
\end{algorithm}

To derive a latency upper bound, a class of MDS-Reservation($t$) scheduling policies are proposed in \citep{MDS_queue}. This class of scheduling policies are indexed by an integer parameter $t$. When a file is requested, a set of $k$ tasks is created. The requests are scheduled as many of the servers are idle (up to $k$). Further, the remaining tasks are kept as a batch in the buffer. On departure of any task from a server, the buffer is searched in order for a batch from which any job has not been serverd by this server, and a task from that batch is served. For the parameter $t$, an additional restriction is imposed: any file request $i\ge t+1$ (i.e., the $i$th batch of chunk requests) can move forward in the buffer only when all $k$ of its chunk requests can move forward together (i.e., when one of the $t$ head-of-line file requests is completed). The basic pseudo-code for the MDS-reservation$(t)$ is described in Algorithm \ref{alg:mdst}. 
 It is easy to see that by blocking any file requests for $i\ge t+1$, such MDS-Reservation($t$) scheduling policies provide an upper bound on the file access latency of MDS queues, with a larger $t$ leading to a tighter bound, yet at the cost of more states to maintain and analyze.

An example of MDS-reservation$(t)$ scheduling policy (for $t=1,2$ respectively) and the corresponding queuing policies are illustrated in Figure~\ref{fig:MDS_sys} for $(n,k)=(5,2)$ codes. Since chunk request $R_{3}^{A,1}$ in batch 3 is already processed by server 2, the second chunk request $R_{3}^{A,2}$ in the same batch cannot be processed by the same server. Under MDS-reservation$(1)$, a batch not at the head-of-line can proceed only if all chunk requests in the batch can move forward together. Since the condition is not satisfied, server 2 must enter an idle state next, leading to resource under utilization and thus higher latency. On the other hand, an MDS-reservation$(2)$ policy allows any chunk requests in the first two batches to move forward individually. Chunk request $R_{4}^{A,1}$ moves into server 2 for processing. It is easy to see that as $t$ grows, MDS-reservation$(t)$ scheduling policy becomes closer to the optimal scheduling policy.

\begin{figure}[!thbp]
\begin{center}
\scalebox{1}{\includegraphics[draft=false,width=0.59\textwidth]{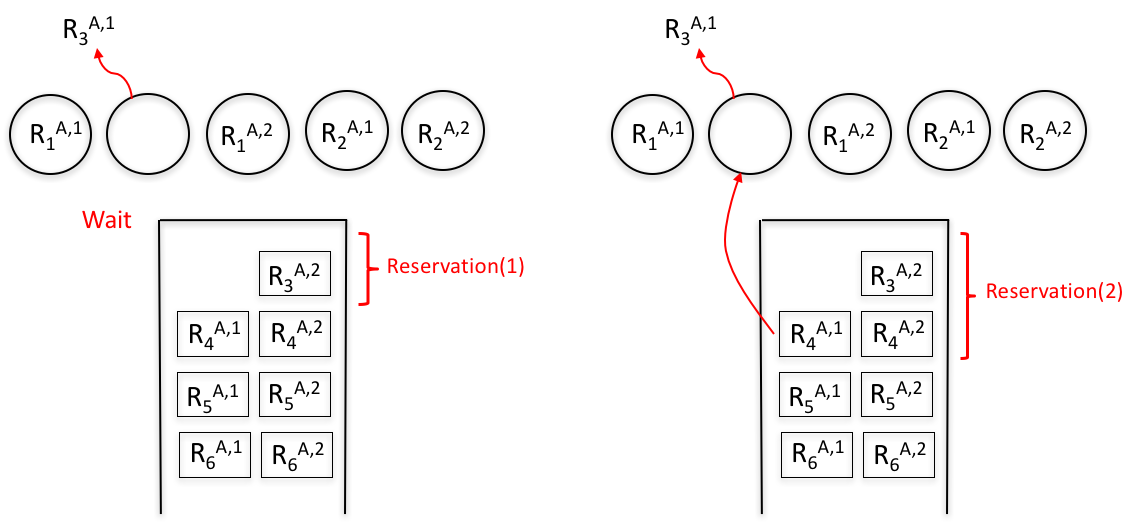}}
\vspace{-3mm}
\caption{An illustration of MDS-reservation$(1)$ (Left) and MDS-reservation$(2)$ (Right) scheduling.}
\label{fig:MDS_sys}
\end{center}
\vspace{-.35in}
\end{figure}

\section{Characterization of Latency Upper Bound via MDS-Reservation Scheduling }
\label{sec:MDS:upper}

The key idea in analyzing these MDS queues is to show that the corresponding Markov chains belong to a class of processes known as Quasi-Birth-Death (QBD) processes \citep{MDS_queue}. Thus, the steady-state distribution can be obtained by exploiting the properties of QBD processes. More precisely, a birth-death process is defined as a continuous-time Markov process on discrete states $\{0, 1,2, \ldots \}$, with transition rate $\lambda$ from state $i$ to $i+1$, transition rate $\mu$ from state $i+1$ to $i$, and rates $\mu_0,\lambda_0$ to and from the boundary state $i=0$, respectively. A QBD process is a generalization of such birth-death processes whose states $i$ are each replaced by a set of states, known as a level. Thus a QBD process could have transitions both within a level and between adjacent levels. It has a diagonal transition probability matrix:
\begin{eqnarray}
\left[ \begin{array}{ccccc} 
B_1 & B_2 & 0 & 0 & \ldots \\
B_0 & A_1 & A_2 & 0 & \ldots \\
0 & A_0 & A_1 & A_2 & \ldots \\
0 & 0 & A_0 & A_1 & \ldots \\
0 & 0 & 0 & A_0 & \ldots \\
\vdots & \vdots & \vdots & \vdots & \ddots
 \end{array} \right]
\end{eqnarray}
where matrices $B_1,B_2,B_0$ are transition probabilities within, from, and to the boundary states, and $A_1,A_2,A_0$ are transition probabilities within each level, entering the next level, and entering the previous level, respectively.

\begin{theorem}[\citep{MDS_queue}]
The Markovian representation of the MDS-Reservation($t$) queue has a state space $\{0,1,\ldots,k\}^t \times \{0,1,\ldots,\infty \}$. It is a QBD process with boundary states $\{0,1,\ldots,k\}^t \times \{0,1,\ldots,n-k+tk \}$ and levels $\{0,1,\ldots,k\}^t \times \{0,1,\ldots,n+jk \}$ for $j=\{t,t+1,\ldots,\infty\}$. 
\end{theorem}

\begin{proof} We briefly summarize the proof in \citep{MDS_queue}. For any state of the system $(w_1,w_2,\ldots,w_t, m) \in \{0,1,\ldots,k\}^t \times \{0,1,\ldots,\infty \}$, define 
\begin{eqnarray}
q=\left\{ \begin{array}{cc} 
0 &  {\rm if} \ w_1=0 \\
t & {\rm else \ if} \ w_t\neq 0 \\
{\rm arg \ max} \{ \tau: w_\tau\neq 0, 1\le \tau \le t \} & {\rm otherwise.}
\end{array} \right. \label{eq:MDS_1}
\end{eqnarray}
Then we can find the number of waiting request batches ($b$), the number of idle servers in the system ($z$), the number of jobs of $i$th waiting batch in the servers ($s_i$), and the number of jobs of $i$th waiting batch in the buffer ($w_i$) as follows:
\begin{eqnarray}
b=\left\{ \begin{array}{cc} 
0 &  {\rm if} \ q=0 \\
q & {\rm else \ if} \ 0<1<t\\
t+\left\lfloor \frac{m-\sum_j w_j-n}{k} \right\rfloor  & {\rm otherwise.}
\end{array} \right.
\end{eqnarray}
\begin{eqnarray}
z=n-\left( {m-\sum_j w_j-(b-t)^+k} \right),
\end{eqnarray}
\begin{eqnarray}
s_i=\left\{ \begin{array}{cc} 
w_{i+1}-w_i &  {\rm if} \ i\in\{1,\ldots ,q-1\} \\
k-z-w_i & {\rm if} \ i=q\\
0  & {\rm if} \ i\in\{q+1,\ldots ,b\}
\end{array} \right. \ \ \ {\rm for} \ i\in\{1,\ldots , b\}
\end{eqnarray}
\begin{eqnarray}
w_i=k, \ \ {\rm for} \ i\in\{t+1,\ldots , b\}. \label{eq:QBD}
\end{eqnarray}

These equations characterize the full transition. It is easy to verify that the MDS-Reservation($t$) queue has the following two key properties: i) Any transitions
change the value of $m$ by at most $k$; and ii) For $m\ge n-k+1+tk$, the transition from any state $(w_1,m)$ to any other states $(w_1^{'}, m^{'}\ge n-k+1+tk)$ depends on $m \ mod \ k$ and not on the actual value of $m$. It is then straightforward to show that this satisfies the boundary and level conditions of a QBD process, with boundary and level transitions specified in the theorem.
\end{proof}

The proof of Theorem provides a procedure to obtain the configuration of the entire queuing system under the MDS-Reservation($t$) scheduling policies. Further, as $t$ goes to infinity, the system approaches an MDS queue, thus resulting in a tighter upper bound at the cost of more complicated queuing analysis. This is because the MDS-Reservation($t$) scheduling policy follows the MDS scheduling policy when the number of file requests in the buffer is less than or equal to $t$. Thus, it is identical to MDS queue when $t$ goes to infinity.

\begin{theorem}
The MDS-Reservation($t$) queue, when $t =\infty$, is precisely the MDS queue for homogenous files, homogenous servers, exponential service times, and $n=m$.
\end{theorem}

These results allow us to employ any standard solver to obtain the steady-state distribution of the QBD process, enabling latency analysis under MDS-Reservation($t$) scheduling policies. In particular, for $t=0$, the MDS-Reservation($0$) policy is rather simple, as the file request (consisting of a batch of chunk requests) at the head of the line may move forward and enter service only if there are at least $k$ idle servers. When $n=k$, this becomes identical to a split-merge queue \citep{MDS_22}. For $t=1$, the MDS-Reservation($1$) policy is identical to the block-one scheduling policy proposed in \citep{MG1:12}.

Using this queue model, we can also find the stability region of the MDS-Reservation($t$) scheduling policy. While an exact characterization is non-tractable in general, bounds on the maximum stability region, defined as the maximum possible number of requests that can be served by the system per unit time (without resulting in infinite queue length) is find in \citep{MDS_queue}. 

\begin{theorem}[\citep{MDS_queue}]
For any given $(n,k)$ and $t>1$, the maximum throughput $\lambda^{*}_{\rm Resv(t)}$ in the stability region satisfies the following inequalities when $k$ is treated as a constant:
\begin{eqnarray}
(1-O(n^{-2}))\frac{n}{k}\mu \le \lambda^{*}_{\rm Resv(t)} \le \frac{n}{k}\mu.
\end{eqnarray}
\end{theorem}

\begin{proof}
First, we note that for $t\ge 2$, latency of each of the MDS-Reservation($t$) queues is upper bounded by that of MDS-Reservation($1$), since less batches of chunk requests are blocked and not allowed to move forward into the servers, as $t$ increases. 

Next, we evaluate the maximum throughput in the stability region of MDS-Reservation($1$) by exploiting properties of QBD systems. We follow the proof in~{\citep{MDS_queue}}. Using the QBD process representations in Equation~(\ref{eq:QBD}), the maximum throughput $\lambda^{*}_{\rm Resv(t)}$ of any QBD system is the value of $\lambda$ such that: $\exists v$ satisfying $v^T(A_0+A_1+A_2)=0$ and $v^TA_0{\bf 1}=v^TA_1{\bf 1}$, where ${\bf 1}$ is an all-one vector. For fixed values of $\mu$ and $k$, it is easy to verify that the matrices $A_0$, $A_1$, and $A_2$ are affine transformations of arrival rate $\lambda$. Plugging in the values of $A_0$, $A_1$ and $A_2$ in the QBD representation of MDS-Reservation($1$) queues, we can show that such $v$ vector exist if $\lambda^{*}_{\rm Resv(1)}\ge (1-O(n^{-2}))\frac{n}{k}\mu$.

It then follows that $\lambda^{*}_{\rm Resv(t)}\ge \lambda^{*}_{\rm Resv(1)}\ge (1-O(n^{-2}))\frac{n}{k}\mu$ for any $t\ge 2$. The upper bound on $\lambda^{*}_{\rm Resv(t)}$ is straightforward since each batch consists of $k$ chunk requests, the rate at which batches exit the system (for all $n$ servers combined) is at most $n\mu / k$.
\end{proof}

\section{Characterization of Latency Lower Bound}
\label{sec:MDS:lower}

To derive a lower bound on service latency for MDS queues, we leverage a class of $M^k/M/n(t)$ scheduling policies proposed in \citep{MDS_queue}, which relax the requirement that $k$ chunk requests belonging to the same file request must be processed by distinct servers after the first $t$ requests.  It applies the MDS scheduling policy whenever there are $t$ or fewer file requests (i.e., $t$ or fewer batches of chunk requests) in the system, while ignoring the requirement of distinct servers when there are more than $t$ file requests.

\begin{theorem}[\citep{MDS_queue}]
The Markovian representation of the $M^k/M/n(t)$ queue has a state space $\{0,1,\ldots,k\}^t \times \{0,1,\ldots,\infty \}$. It is a QBD process with boundary states $\{0,1,\ldots,k\}^t \times \{0,1,\ldots,n+tk \}$ and levels $\{0,1,\ldots,k\}^t \times \{n-k+1+jk,\ldots,n+jk \}$ for $j=\{t+1,\ldots,\infty\}$. 
\end{theorem}

\begin{proof}
We again define $q$ for any system state $(w_1,w_2,\ldots,w_t, m) \in \{0,1,\ldots,k\}^t \times \{0,1,\ldots,\infty \}$ as in Equation (\ref{eq:MDS_1}). The values of $b$, $z$, $s_i$, and $w_i$ can be derived accordingly and are identical to those in Chapter \ref{sec:MDS:upper}. These equations capture the entire state transitions. It is then easy to see that the $M^k/M/n(t)$ queue satisfy the following two properties: i) Any transitions change the value of $m$ by at most $k$; and ii) For $m\ge n-k+1+tk$, the transition from any state $(w_1,m)$ to any other states $(w_1^{'}, m^{'}\ge n-k+1+tk)$ depends on $m \ mod \ k$ and not on the actual value of $m$. This results in a QBD process with boundary states and levels described in the theorem.
\end{proof}

Similar to the case of MDS-Reservation($t$) scheduling policies, as $t$ goes to infinity in $M^k/M/n(t)$ scheduling policies, the resulting system approaches an MDS queue, thus providing a tighter lower bound at the cost of more complicated queuing analysis. 

\begin{theorem}
The $M^k/M/n(t)$ queue, when $t =\infty$, is precisely the MDS queue for homogenous files, homogenous servers, exponential service times, and $n=m$.
\end{theorem}

Again, these results allow us to obtain the steady-state distribution of the QBD process, which enables latency analysis under $M^k/M/n(t)$ scheduling policies.

\section{Extension to Redundant Requests}
\label{sec:MDS:redundant}

In erasure coded storage systems, access latency can be further reduced by sending redundant requests to storage servers. Consider a scheme under $(n,k)$ MDS codes, which sends (redundantly) each file request to $v> k$ servers. Clearly, upon completion of any $k$-out-of-$v$ chunk requests, the file request is considered to be served, and the remaining $v-k$ active chunk requests could be canceled and removed from the system.  It is easy to see that redundant requests allow the reduction of individual requests at the expense of an increase in overall queuing delay due to the use of additional resources on $v-k$ straggler requests. We note that when $k=1$, the redundant-request policy reduces to a replication-based scheme.

Formally, an MDS queue with redundant requests is associated to five parameters $(n, k)$, $[\lambda,\mu]$, and the redundant level $v\ge k$, satisfying the following modified assumptions: i) File requests arrive in batches of $v$ chunk requests each; ii) Each of the $v$ chunk requests in a batch can be served by an arbitrary set of $v$ distinct servers; iii) Each batch of $v$ chunk requests is served when any $k$ of the $v$ requests are served.

While empirical results in \citep{MDS_9,MDS_18,MDS_24} demonstrated that the use of redundant requests can lead to smaller latency under various settings, a theoretical analysis of the latency - and thus a quantification of the benefits - is still an open problem. Nevertheless, structural results have been obtained in \citep{MDS_9} using MDS-queue models, e.g., to show that request flooding can indeed reduce latency in certain special cases.

\begin{theorem}
Consider a homogeneous MDS$(n,k)$ queue with Poisson arrivals, exponential service time, and identical service rates. If the system is stable in the absence of redundant requests, a system with the maximum number $v=n$ of redundant requests achieves a strictly smaller average latency than any other redundant request policies, including no redundancy $v=k$ and time-varying redundancy. 
\end{theorem}

\begin{proof}
Consider two systems, system $S_1$ with redundant level $v<n$ and system $S_2$ with redundant level $n$. We need to prove that under the same sequence of events (i.e., the arrivals and server completions), the number of batches remaining in system $S_1$ is at least as much as that in $S_2$ at any given time. To this end, we use the notion ``time $z$" to denote the time immediately following the $z$th arrival/departure event. Assume both systems, $S_1$ and $S_2$, are empty at time $Z=0$. Let $b_i(z)$ denote the number of batches remaining in system $S_i$ at time $z$. We use induction to complete this proof. 

At any time $z$, we consider induction hypothesis: i) $b_1(z)\ge b_2(z)$, and ii) if there are no further arrivals from time $z$ onwards, then at any time $z^{'}> z$, $b_1(z^{'})\ge b_2(z^{'})$. It is easy to see that both conditions hold at $z=0$. Next, we show that they also hold for $z+1$.

First suppose that the event at $z+q$ is the completion of a chunk request at one of the $n$ servers. Then the hypothesis to time $z$ implies the satisfaction of all the hypotheses at time $z+1$, due to the second hypothesis, as well as the fact of no further arrivals from $z$ to $z+1$.

Now suppose the event at $z+1$ is the arrival of a new batch. Let $a_1(z^{'})$ and $a_2(z^{'})$ be the number of batches remaining in the two systems at time $z^{'}$ if the new batch had not arrived. From the second hypothesis, we have $a_1(z^{'})\ge a_2(z^{'})$. Since the MDS scheduling policy processes all batches sequentially, under any sequence of departures, we should have $b_1(z^{'})=a_1(z^{'})$ if the the new batch has been served in $S_i$, and $b_1(z^{'})=a_1(z^{'})+1$ otherwise. When $b_1(z^{'})=a_1(z^{'})+1$, it is easy to see that $b_1(z^{'})=a_1(z^{'})+1\ge a_2(z^{'})+1\ge b_2(z^{'})$. Thus, we only need to consider the case $b_1(z^{'})=a_1(z^{'})$, which implies that the new batch has been served in $S_1$ at or before time $z^{'}$.

Let $z_1, \ldots, z_k$ be the events when the $k$ chunk requests of the new batch have been served in $S_1$. At these times, the corresponding servers must have been idle in $S_1$ if the new batch had not arrived, implying $a_1(z^{'}) = c_1(z^{'})$, where $c_1(z^{'})$ is the number of batches remaining at time $z^{'}$  excluding the events at $z_1, \ldots, z_k$. From the second hypothesis, we also have $c_1(z^{'})\ge c_2(z^{'})$. Then, it is sufficient to prove $b_2(z^{'})=c_2(z^{'})$ next. 

Note that if at time $z$, all $b_2(z)$ batches present in the system had all its $k$ chunk requests remaining to be served, and there were no further arrivals, then the
total number of batches served between times $z$ and $z^{'}>z$ can be counted by assigning $n$ independent exponential timers to the servers. Since
the events $z_1, \ldots, z_k$ must correspond to the completion of exponential timers at $k$ distinct servers, it must be that $b_2(z^{'})=c_2(z^{'})$. Putting the pieces together, it means $b_1(z^{'})\ge b_2(z^{'})$, which completes the induction.

Finally, when when $S_1$ employs redundant level $v<n$, we need to show that for a fraction of time that is strictly bounded away from zero, the number of batches remaining in $S_2$ is strictly smaller than that in $S_1$. We start with the state where both systems $S_1$ and $S_2$ are empty. Consider the arrival of a single batch followed by a sequence of completion of the exponential timers at different servers. It is easy to see that the systems can be taken to states $b_1(z^{'})=1$ and $b_2(z^{'})=0$ respectively, with a probability strictly bounded away from zero, when $v<n$ (since the batch in $S_2$ is served when any $k$ out of $n$ timer expire, rather than $k$ out of $v$ in $S_1$). Thus, the fraction of time in which $b_1(z^{'})-b_2(z^{'})=1$ is strictly bounded away from zero. This completes the proof.
\end{proof}

\section{Simulations}
\label{sec:MDS:sim}

\begin{figure}[!thbp]
\begin{center}
\scalebox{1}{\includegraphics[draft=false,width=0.59\textwidth]{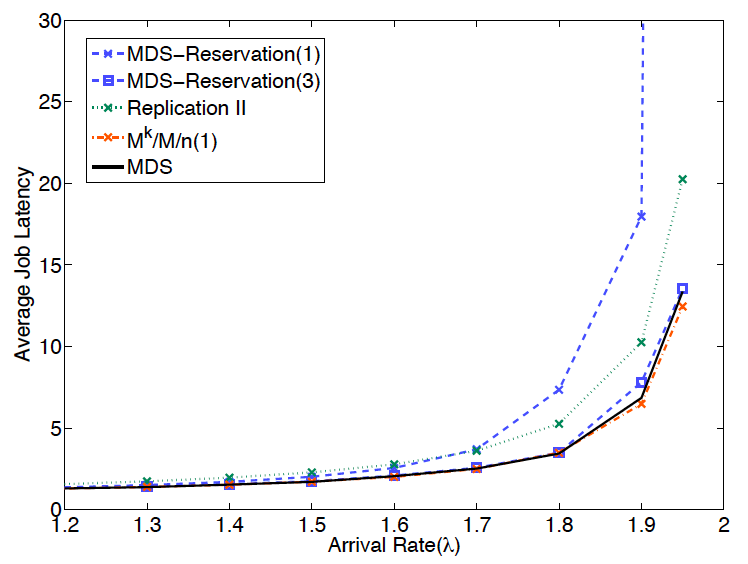}}
\vspace{-3mm}
\caption{A comparison of average file access latency and its upper/lower bounds through MDS-Reservation($t$) and $M^k/M/n(t)$ scheduling policies.}
\label{fig:MDS_1}
\end{center}
\vspace{-.35in}
\end{figure}

The proposed latency bounds - using MDS-Reservation($t$) and $M^k/M/n(t)$ scheduling policies respectively - have been compared in \citep{MDS_queue} through numerical examples. For an MDS system with $(n,k)=(10,5)$ and $\mu=1$, Figure~\ref{fig:MDS_1} plots the average file access latency for various scheduling policies. Here, average latency bounds under MDS-Reservation($t$) and $M^k/M/n(t)$ scheduling policies are computed using Little's Law to the stationary distribution. A Monte-Carlo simulation is employed to numerically find the exact latency of MDS queues. The results are also compared to a simple policy, ``Replication-II", in which the $n$ servers are partitioned into a $k$ sets of $n/k$ servers each, and each of the $k$ chunk requests is served by a separate set of servers. We note that tail latency performance cannot be analyzed through MDS-Reservation($t$) and $M^k/M/n(t)$ scheduling policies.

\begin{figure}[!thbp]
\begin{center}
\scalebox{1}{\includegraphics[draft=false,width=0.59\textwidth]{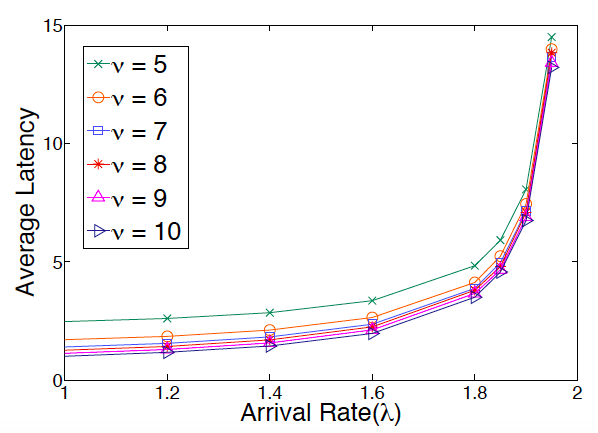}}
\vspace{-3mm}
\caption{Simulation results showing the reduction of average latency with an increase in the redundant level $v$ for an MDS(10,5) queue.}
\label{fig:MDS_2}
\end{center}
\vspace{-.35in}
\end{figure}

For a homogeneous MDS$(n,k)$ queue with redundant requests, Figure~\ref{fig:MDS_2} shows the simulation of file access latency for varying redundant levels $v$. It corroborates the analysis that when the service times are $i.i.d.$ exponential, the
average latency is minimized by $v=n$ requests. Further, it also seems that average latency strictly reduces with an increase in redundant level $v$. But it is unclear if such property carry over to general service time distributions.

\section{Notes and Open Problems}
\label{sec:MDS:note}

The study of latency using MDS queues was initiated by \citep{MG1:12}, which consider a special case of ``block-one-scheduling" policy to obtain a an upper bound on service latency. For arbitrary service time distributions, an analysis of the blocking probability was presented in \citep{MDS_10} in the absence of a shared request buffer. Later, these results were extended in \citep{MDS_queue} to general MDS-Reservation(t) queues and a tighter upper bound on request latency was provided. %
There are a number of open problems that can be considered in future work.
\begin{enumerate}
	\item {\bf General service time}: The proposed analysis on MDS queues assumes an exponential service time distribution. More accurate modeling of service time distribution based on the underlying properties of the storage devices and corresponding latency analysis are open problems.
	\item {\bf Heterogeneous queues}: The analysis in this chapter is limited to servers with equal service time and files of equal chunk sizes, which are not practical. Further, when files are stored on different sets of $n$ servers, it is unclear if the analysis in this chapter using MDS queues can be extended.
	\item {\bf Redundant requests}: The quantification of service-time distributions under redundant requests is an open problem. Even in the case of exponentially distributed service times, the precise amount of latency improvement due to allowing redundant requests is unknown.
\end{enumerate}

\chapter{Fork-Join Scheduling Approach}

In this Chapter, we introduce the model of Fork-Join Scheduling in Section~\ref{sec:fj_sch}, which was first proposed in \citep{Joshi:13}. We will first consider homogenous files and exponential service times to derive upper and lower bounds to the latency in Section \ref{sec:lat_fj}. Further, approximate characterization of latency expressions will also be provided. Section \ref{sec:lat_fj_st} extends the upper and lower bounds on latency to general service time distributions. Section \ref{sec:lat_fj_hs} extends the results in Section \ref{sec:lat_fj} for heterogenous files, where the parameter $k$ can be different for different files, while each file is placed on all the $n$ servers.  Sections~\ref{sec:lat_fj_sims} and \ref{sec:fj_notes} contain simulation results and notes on future directions, respectively.

\section{Fork-Join Scheduling}\label{sec:fj_sch}

We consider a data center consisting of $n$ homogeneous servers, denoted by $\mathcal{M}=\{1,2,\ldots,n\}$, called storage nodes. We consider one file request, where the file arrival is Poisson distributed with rate $\lambda$. We partition the file into $k$ fixed-size chunks, and then encode it using an $(n,k)$ MDS erasure code to generate $n$ distinct chunks of the same size. The encoded chunks are assigned to and stored on $n$ distinct storage nodes. Therefore, each chunk is placed on a different node to provide high reliability in the event of node or network failures. We assume that service time distribution for each server is exponentially distributed with rate $\mu$. Even though we assume one file, multiple homogeneous files (with the same size) can be easily incorporated.

We first introduce the Fork-Join system in the following definition. 

\begin{definition}
	An $(n, k)$ fork-join system consists of $n$ nodes. Every
	arriving job is divided into $n$ tasks which enter first-come first-serve queues at each of the $n$
	nodes. The job departs the system when any $k$ out of $n$ tasks are served by their respective	nodes. The remaining $n - k$ tasks abandon their queues and exit the system before completion
	of service.
	\end{definition}

The $(n, n)$ fork-join system, known in literature as fork-join queue, has been extensively
studied in, e.g., \citep{kim1989analysis,nelson1988approximate,varki2008m}. The $(n,k)$ generalization was first studied in \citep{Joshi:13}, and has been followed by multiple works in the distributed storage literature \citep{gardner2015reducing,fidler2016non,kumar2014latency,parag2017latency,badita2019latency}. 

It can be shown that for the  $(n, k)$
fork-join system to be stable, the rate of Poisson arrivals $k \lambda$ should be less than $n\mu$. Thus, $\lambda<n\mu/k$ is the stable region. We will now use the Fork-Join system as a scheduling strategy, where the $k$ tasks are encoded to $n$ tasks, and the scheduling starts the job on all $n$ servers. The job departs the system when any $k$ out of $n$ tasks are served by their respective servers. The remaining $n - k$ tasks abandon their queues and exit the system before completion of service. In the following section, we will provide bounds on the latency with Fork-Join scheduling.

An example of fork-join queue for $(n,k)=(5,2)$ is illustrated in Figure~\ref{fig:fj_sys}. Each batch of chunk requests are mapped to all 5 servers. They depart the system together as soon as any 2 chunk requests are served and the remaining 3 requests abandon processing. No other batches in the queue can move forward before all servers become available. This leads to underutilized of server capacity and thus provides an upper bound on the optimal scheduling policy.

\begin{figure}[!thbp]
\begin{center}
\scalebox{1}{\includegraphics[draft=false,width=0.59\textwidth]{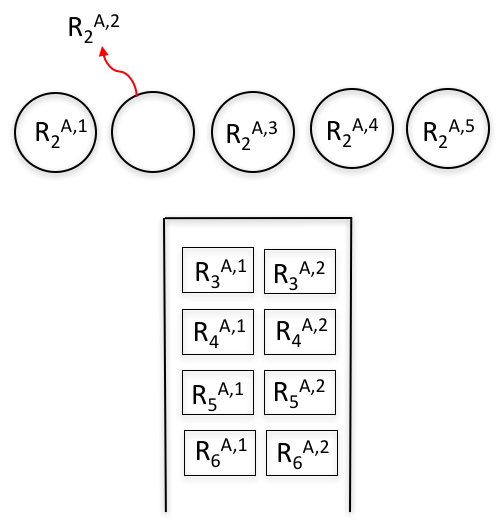}}
\vspace{-3mm}
\caption{An illustration of MDS-reservation$(1)$ (Left) and MDS-reservation$(2)$ (Right) scheduling.}
\label{fig:fj_sys}
\end{center}
\vspace{-.35in}
\end{figure}

\section{Characterization of  Latency}\label{sec:lat_fj}
We now provide bounds on the expected file latency, which is the mean response time $T_{(n,k)}$ of the $(n,k)$ fork-join system. It is the expected time that a job spends in the system, from its arrival until $k$ out of $n$ of its tasks are served by their respective nodes.

Since the $n$ tasks are served by independent M/M/1 queues, intuition suggests that $T_{(n,k)}$ is the $k^{th}$ order statistic of $n$ exponential service times. However this is not true, which makes the analysis of $T_{(n,k)}$ challenging. The reason why the order statistics approach does not work is the cancellation of jobs in the queue, where their abandonment has to be taken into account. 

Let $H_{x,y}^z$ be a generalized harmonic number of order $z$ defined by
\begin{equation}
H_{x,y}^z = \sum\limits_{j=x+1}^y \frac{1}{j^z},
\label{eq:harmMean}
\end{equation}
for some positive integers $x, y$ and $z$. The following result provides an upper bound on the expected download time.

\begin{theorem}[\citep{Joshi:13}]
\label{thm:upper_bnd}
The expected file latency, $T_{(n,k)}$,  satisfies
\begin{equation}
\label{eqn:upper_bnd}
T_{(n,k)} \leq   \, \frac{H_{n-k,n}^1}{\mu}\,  +  \frac{\lambda\bigl [H_{n-k,n}^2  + (H_{n-k,n}^1)^2\bigr]}{2 \mu^2 \bigl[ 1- \rho H_{n-k,n}^1\bigr]},
\end{equation}
where $\lambda$ is the request arrival rate, $\mu$ is the service rate at each queue, $\rho = \lambda/\mu$ is the load factor. We note that the bound is valid only when $\rho H_{n-k,n}^1 < 1$.
\end{theorem}

\begin{proof}
To find this upper bound, we use a model called the split-merge system, which is similar but easier to analyze than the fork-join system. In the $(n,k)$ fork-join queueing model, after a node serves a task, it can start serving the next task in its queue. On the contrary, in the split-merge model, the $n$ nodes are blocked until $k$ of them finish service. Thus, the job departs all the queues at the same time. Due to this blocking of nodes, the mean response time of the $(n,k)$ split-merge model is an upper bound on (and a pessimistic estimate of) $T_{(n,k)}$ for the $(n,k)$ fork-join system. 

The $(n,k)$ split-merge system is equivalent to an M/G/1 queue where arrivals are Poisson with rate $\lambda$ and service time is a random variable $S$ distributed according to the $k^{th}$ order statistic of the exponential distribution.

The mean and variance of $S$ are given as
\begin{equation}
{\mathbb{E}}[S] = \frac{H_{n-k.n}^1}{\mu}  ~~ \text{and} ~~  \text{Var}[S] = \frac{H_{n-k,n}^2}{\mu^2}.
\label{eqn:mvk_exp}
\end{equation}
The Pollaczek-Khinchin formula \citep{zwart2000sojourn1}  gives the mean response time $T$ of an M/G/1 queue in terms of the mean and variance of $S$ as,
\begin{equation}
T = {\mathbb{E}}[S] + \frac{ \lambda ( {\mathbb{E}}[S]^2  + \text{Var}[S] } {2( 1-\lambda {\mathbb{E}}[S])}. \label{eqn:pollac_khin}
\end{equation}
Substituting the values of ${\mathbb{E}}[S]$ and $\text{Var}[S]$ given by (\ref{eqn:mvk_exp}), we get the upper bound (\ref{eqn:upper_bnd}). Note that the Pollaczek-Khinchin formula is valid only when $ \frac{1}{\lambda} > {\mathbb{E}}[S]$, the stability condition of the M/G/1 queue. Since ${\mathbb{E}}[S]$ increases with $k$, there exists a $k_0$ such that the M/G/1 queue is unstable for all $k \geq k_0$. The inequality $\frac{1}{\lambda} > {\mathbb{E}}[S]$ can be simplified to $\rho H_{n-k,n}^1 <1$ which is the condition for validity of the upper bound given in Theorem~\ref{thm:upper_bnd}.
\end{proof}

We also note that the stability condition for the upper bound is $\rho H_{n-k,n}^1 < 1$ which is not the same as the stability condition of the Fork-Join queue $\lambda<n\mu/k$. This shows that the upper bound technique is loose, and does not result in an efficient bound in the region close to $\lambda=n\mu/k$. We now find the lower bound on the latency in the following theorem. 

\begin{theorem}[\citep{Joshi:13}]
	\label{thm:lower_bnd}
	The expected file latency, $T_{(n,k)}$,  satisfies
	\begin{equation}
	T_{(n,k)} \!\! \geq \! \sum_{j=0}^{k-1} \frac{1}{(n-j)\mu - \lambda}, \label{eqn:lower_bnd}
	\end{equation}
	where $\lambda$ is the request arrival rate and $\mu$ is the service rate at each queue.
\end{theorem}

\begin{proof}
	The lower bound in (\ref{eqn:lower_bnd}) is a generalization of the bound for the $(n,n)$ fork-join system derived in \citep{varki2008m}. The bound for the $(n,n)$ system is derived by considering that a job goes through $n$ stages of processing. A job is said to be in the $j^{th}$ stage if $j$ out of $n$ tasks have been served by their respective nodes for $0 \leq j \leq n-1$. The job waits for the remaining $n-j$ tasks to be served, after which it departs the system. For the $(n,k)$ fork-join system, since we only need $k$ tasks to finish service, each job now goes through $k$ stages of processing. In the $j^{th}$ stage, where $0 \leq j \leq k-1$, $j$ tasks have been served and the job will depart when $k-j$ more tasks to finish service.
	
	We now show that the service rate of a job in the $j^{th}$ stage of processing is \emph{at most} $(n-j) \mu$. Consider two jobs $B_1$ and $B_2$ in the $i^{th}$ and $j^{th}$ stages of processing respectively. Let $i>j$, that is, $B_1$ has completed more tasks than $B_2$. Job $B_2$ moves to the $(j+1)^{th}$ stage when one of its $n-j$ remaining tasks complete. If all these tasks are at the heads of their respective queues, the service rate for job $B_2$ is exactly $(n-j) \mu$. However since $i > j$, $B_1$'s task could be ahead of $B_2$'s in one of the $n-j$ pending queues, due to which that task of $B_2$ cannot be immediately served. Hence, we have shown that the service rate of in the $j^{th}$ stage of processing is at most $(n-j)\mu$.

	Thus, the time for a job to move from the $j^{th}$ to $(j+1)^{th}$ stage is lower bounded by $1/( (n-j)\mu - \lambda)$, the mean response time of an M/M/1 queue with arrival rate $\lambda$ and service rate $(n-j) \mu$. The total mean response time is the sum of the mean response times of each of the $k$ stages of processing and is bounded below as in the statement of the theorem.
\end{proof}
We note that  the lower bound does not achieve the optimal stability region, giving the threshold as $\lambda<(n-k+1)\mu$. 

An approximate characterization of latency has also been studied \citep{badita2019latency}. The approach follows the structure of the lower bound mentioned above, which goes in stages. A job is said to be in the $j^{th}$ stage if $j$ out of $n$ tasks have been served by their respective nodes for $0 \leq j \leq k-1$. Since the job goes from stage 0 to stage 1, all the way to stage $k-1$ and then get served when $k$ chunks have been serviced, the procedure is akin to a tandem queue where a service from stage $j$ leads to stage $j+1$. Thus, we consider $k-1$ tandem queues for the approximation, which are assumed to be uncoupled, labeled as queue $j\in \{0, \cdots, k-1\}$. The arrival rate at the tandem queue $0$ is the external arrivals which is Poisson at rate $\lambda$. Since it is tandem queue and service time is assumed to be exponential, the arrival rate at each queue will be $\lambda$ \citep{Ross2019}. In the case of the lower bound, the service rate for tandem queue $j$ was taken as $(n-j) \mu$, while this is where a better approximation will be used. We let $\gamma_j$ be the approximate service rate of queue $j$ and $\pi_j(r)$ be the probability that the queue length of tandem queue $j$ is $r$.

The service rate of queue $k-1$ is $\gamma_{k-1} = (n-k+1)\mu$ as in the lower bound. For the other queues,  service rate includes $\mu$ and the additional resources from the later queues, which for the lower bound became $(n-j) \mu$. However, the later queues are not always empty and the resources cannot be used to serve the earlier queues. In the approximation, we let the resources of the later queues help the earlier queues only when they are empty.  Using additional resources of tandem queue $k-1$ to serve requests at queue $k-2$ when queue $k-1$ is empty gives $\gamma_{k-2} = \mu + \gamma_{k-1} \pi_{k-1}(0)$. Proceeding back with the same method, we have the recursion on $\gamma_i$ as:
\begin{equation}
\gamma_j = \begin{cases}
(n-k+1)\mu, & i=k-1\\
\mu + \gamma_{j+1}\pi_{j+1}(0) & j\in \{0, 1, \cdots, k-2\}
\end{cases}. 
\end{equation}

Since each tandem queue is M/M/1, $\pi_j(0) = 1-\frac{\lambda}{\gamma_j}$ \citep{Ross2019}. Thus, we can compute $\gamma_j$ as 
\begin{equation}
\gamma_j = (n-j)\mu - (k-j-1)\lambda, j\in \{0, 1, \cdots, k-1\}. 
\end{equation} 
The overall approximate latency is given as

	\begin{equation}
T_{(n,k)} \!\!\approx \! \sum_{j=0}^{k-1} \frac{1}{\gamma_j - \lambda} = \sum_{j=0}^{k-1} \frac{1}{(n-j)\mu - (k-j)\lambda}.
\end{equation}

This is summarized in the following lemma. 

\begin{lemma}[\citep{badita2019latency}]
	The expected file latency, $T_{(n,k)}$ can be approximated as 
		\begin{equation}
	T_{(n,k)} \!\!\approx \! \sum_{j=0}^{k-1} \frac{1}{(n-j)\mu - (k-j)\lambda} \label{eqn:approx_bnd},
	\end{equation}
	where $\lambda$ is the request arrival rate and $\mu$ is the service rate at each queue.
\end{lemma}

We further note that the stability condition for the approximation is $k\lambda<n\mu$, and thus it satisfies the optimal stability condition. 
\section{Extension to General Service Time Distributions}\label{sec:lat_fj_st}
The above results for exponential service distribution can be extended to the case where each homogeneous server has a general service time distribution.  Let $X_1, X_2, \ldots, X_n$ be the i.i.d random variables representing the service times of the $n$ nodes, with expectation ${\mathbb{E}}[X_i] = \frac{1}{\mu}$ and variance $\text{Var}[X_i] = \sigma^2$ for all $i$. 

\begin{theorem}[\citep{Joshi:13}]
	The mean response time $T_{(n,k)}$ of an $(n, k)$ fork-join system with general service time $X$ such that ${\mathbb{E}}[X] = \frac{1}{\mu}$ and $\text{Var}[X] = \sigma^2$ satisfies
	\begin{align}
	T_{(n,k)} \leq \frac{1}{\mu} &+ \sigma \sqrt{\frac{k-1}{n-k+1}} \nonumber \\
	&+ \frac{\lambda \left[ \left ( \frac{1}{\mu} + \sigma \sqrt{\frac{k-1}{n-k+1}} \right)^2 + \sigma^2 C(n,k) \right]}{2 \left[ 1- \lambda \left( \frac{1}{\mu} + \sigma \sqrt{\frac{k-1}{n-k+1}} \right ) \right ]}, \label{eqn:upper_bnd_gen}
	\end{align}
	where $C(n,k)$ is a constant depending on $n$ and $k$, and is defined as \citep{Papadatos}
	\begin{equation}
	C(n,k) = \sup_{0<x<1}\{\frac{I_x(k,n+1-k)(1-I_x(k,n+1-k))}{x(1-x)}\}, 
	\end{equation}
	where $I_x(a,b)$ is the incomplete beta function. The values of constant $C(n,k)$ for certain $n$ and $k$ can be found in the table in \citep{Papadatos}, respectively.
\end{theorem}
\begin{proof}
	The proof follows from Theorem~\ref{thm:upper_bnd} where the upper bound can be calculated using $(n, k)$ split-merge system and Pollaczek-Khinchin formula (\ref{eqn:pollac_khin}). Unlike the exponential distribution, we do not have an exact expression for $S$, i.e., the $k^{th}$ order statistic of the service times $X_1, \,\, X_2, \,\, \cdots X_n$. Instead, we use the following upper bounds on the expectation and variance of $S$ derived in \citep{ArnoldGroeneveld} and \citep{Papadatos}.
	\begin{align}
	{\mathbb{E}}[S] &\leq \frac{1}{\mu} + \sigma \sqrt{\frac{k-1}{n-k+1}}, \label{eqn:upperbound_exp} \\
	\text{Var}[S] &\leq C(n,k) \sigma^2. \label{eqn:upperbound_var}
	\end{align}
	
	The proof of (\ref{eqn:upperbound_exp}) involves Jensen's inequality and Cauchy-Schwarz inequality. For details please refer to \citep{ArnoldGroeneveld}.    The proof of (\ref{eqn:upperbound_var}) can be found in \citep{Papadatos}.
	
	Note that (\ref{eqn:pollac_khin}) strictly increases as either ${\mathbb{E}}[S]$ or $\text{Var}[S]$ increases. Thus, we can substitute the upper bounds in it to obtain the upper bound on mean response time (\ref{eqn:upper_bnd_gen}).
\end{proof}

We note that the stability condition for the upper bound of latency is $ \lambda \left( \frac{1}{\mu} + \sigma \sqrt{\frac{k-1}{n-k+1}} \right ) <1$. For deterministic service times, $\sigma=0$, and the condition becomes $\lambda<\mu$. However, this is not optimal stability condition for the best scheduling approach. For deterministic service times, Fork-Join scheduling spends an additional time for the $n-k$ tasks which could have been saved and thus leads to non-optimality of the stability region of Fork-Join queues.

Regarding the lower bound, we note that our proof in Theorem~\ref{thm:lower_bnd} cannot be extended to this general service time setting. The proof requires memoryless property of the service time, which does not necessary hold in the general service time case. However, the proof can be extended directly to the shifted exponential service time distribution easily. We let the shifted exponential distribution be ${\rm Sexp}(\beta,\alpha)$, where $\beta$ is the shift and there is an exponential with rate $\alpha$ after the shift. The probability density function of ${\rm Sexp}(\beta,\alpha)$ is given as
\begin{equation}
f_{X}(x) = \begin{cases}
\alpha e^{-\alpha(x-\beta)} & {\text{ for }} x\ge \beta\\
0 & {\text{ for }} x<\beta
\end{cases}\label{eq:Sexp}
\end{equation}

 The lower bound on latency is then given as

\begin{theorem}[\citep{joshi2017efficient}]
	\label{thm:lower_bnd_sexp}
	The expected file latency, $T_{(n,k)}$,  satisfies
	\begin{equation}
	T_{(n,k)} \!\! \geq \! \beta+\frac{1}{n\alpha} +\frac{\lambda\left(\left(\beta+ \frac{1}{n\alpha}\right)^2 + \left( \frac{1}{n\alpha}\right)^2\right)}{2\left( 1-\lambda \left( \beta+ \frac{1}{n\alpha}\right) \right)} +  \sum_{j=1}^{k-1} \frac{\left( \beta+ \frac{1}{n\alpha}\right)}{(n-j) - \lambda\left( \beta+ \frac{1}{n\alpha}\right)}, \label{eqn:lower_bnd_sexp}
	\end{equation}
	where $\lambda$ is the request arrival rate and the the service distribution at each queue is ${\rm Sexp}(\beta,\alpha)$.
\end{theorem}
\begin{proof}
	The proof is an adaptation of Theorem \ref{thm:lower_bnd}, where the latency for the first stage is found using Pollaczek-Khinchin formula with ${\rm Sexp}(\beta,n\alpha)$ (as if all tasks are at head, this will be the distribution of first job finishing). Using this as the completion latency of the first task, the remaining task completions are similar as they still follow exponential distributions. 
\end{proof}

\section{Extension to Heterogeneous Systems}\label{sec:lat_fj_hs}
We now extend the setup where there are $r$ files, where each file $i$ is encoded using $(n,k_i)$ MDS code. We assume that file $i$ is of   size $l_i$. The arrival process for file $i$ is assumed to be Poisson with rate $\lambda_i$. The service time at each server is assumed to follow an exponential distribution with service rate $\mu$ (per unit file size). The effective service rate at any server for file $i$ is $\mu_i = \frac{k_i \mu}{l_i}$ since each server stores $1/k_i$ fraction of data. Let $\rho_i = \frac{\lambda_i}{\mu_i}$ be the server utilization factor for file $i$. The following result describes the conditions for the queues to be stable using Fork-Join queueing. 

\begin{lemma}[\citep{kumar2014latency}]
For the system to be stable using Fork-Join queueing system, we require
\begin{align}
\left(\sum\limits_{i=1}^r{k_i \lambda_i}\right)\left(\sum\limits_{i=1}^r{\frac{\lambda_i l_i}{k_i}}\right) < n \mu \sum\limits_{i=1}^r{\lambda_i}.
\end{align}
\end{lemma}
\begin{proof}
	Jobs of file $i$ enter the queue with rate $\lambda_i$. Each  file $i$ is serviced by the system when $k_i$ sub-tasks of that job are completed. The remaining $n-k_r$ sub-tasks are then cleared from the system. Thus for each request of file $i$, $\frac{(n-k_i)}{n}$ fraction of the sub-tasks are deleted and hence the effective arrival rate of file $i$ at any server is $\lambda_i\left(1-\frac{n-k_i}{n}\right) = \frac{k_i \lambda_i}{n}$. Thus the overall arrival rate at any server, $\lambda_{\textnormal{eff}}$, is 
	\begin{align}
	\lambda_{\textnormal{eff}} = \sum\limits_{i=1}^r{\frac{k_i \lambda_i}{n}}.
	\label{lambda_eff}
	\end{align} 
	Let $S$ denote the service distribution for a single-server FCFS system serving $r$ files, with $p_i$ being the fraction of jobs of class $i$. Then, the mean service time at a server is 
	\begin{align}
	{\mathbb{E}}[S] = \sum\limits_{i=1}^r{p_i 	{\mathbb{E}}[S_i]} = \sum\limits_{i=1}^r{\frac{\lambda_i}{\mu_i \sum\limits_{i=1}^{r}\lambda_i}},	 
	\label{meanServiceTime}	
	\end{align}	
	where \eqref{meanServiceTime} follows the assumption that the service time for file $i$ is exponential with rate $\mu_i$. To ensure stability, the net arrival rate should be less than the average service rate at each server. Thus from \eqref{lambda_eff} and \eqref{meanServiceTime} the stability condition of each queue is 
	\begin{align}
	\sum\limits_{i=1}^r{\frac{k_i \lambda_i}{n}} < {\left(\sum\limits_{i=1}^r \frac{\lambda_i}{\mu_i \sum\limits_{i=1}^{r}\lambda_i}\right)}^{-1}, \nonumber \end{align}
	Since $\mu_i = \frac{k_i \mu}{l_i}$ and the term $\sum\limits_{i=1}^{r}\lambda_i$ is a constant, with simple algebraic manipulations we arrive at the statement of the lemma. 
\end{proof}

Let $\mathcal{S}_j = \sum\limits_{i=1}^j \rho_r H_{n-k_i,n}^1$. %
We will now provide the lower and upper bounds for the mean latency extending the results for homogenous files. The following results provides an upper bound on the latency.

\begin{theorem}[\citep{kumar2014latency}]\label{Theorem1} 
	The average latency for job requests of file $i$ using Fork-Join queueing is upper-bounded as follows:

		\begin{align}
		\label{eq:fcfs_ub}
		T^{i} \leq \underbrace{\frac{H_{n-k_i,n}^1}{\mu_i}}_\text{Service time}+ \underbrace{\frac{\sum\limits_{i=1}^r{\lambda_i[H_{n-k_i,n}^2 + {(H_{n-k_i,n}^1)}^2]/ {\mu_i}^2}}{2\left(1-\mathcal{S}_r\right)}}_\text{Waiting time}.
		\end{align}
		The bound is valid only when $\mathcal{S}_r < 1$.
\end{theorem}
\begin{proof}
	The  system can be modeled as a M/G/1 queuing system with arrival rate $\lambda = \sum\limits_{i=1}^r \lambda_i$ and a general service time distribution $S$. Then the average latency for a job of class $i$  is given  as 
	\begin{align}
	T^i = {\mathbb{E}}[S_i]+ \frac{\sum\limits_{i=1}^r \lambda_i \left[\textnormal{Var}[S_i]+{\mathbb{E}}[S_i]^2\right]}{2\left(1-\sum\limits_{i=1}^r  \lambda_i {\mathbb{E}}[S_i]\right)}. \label{fcfs_ub1}
	\end{align}
	
	To obtain an upper bound on the average latency, we degrade the Fork-Join system in the following manner. For a request of file $i$, the servers that have finished processing a sub-task of that request are blocked and do not accept new jobs until $k_i$ sub-tasks of that request have been completed. Then the sub-tasks at remaining $n-k_i$ servers exit the system immediately. Now this performance-degraded system can be modeled as a M/G/1 system where the distribution of the service process, $S_i$, follows $k_i^{\textnormal{th}}$ ordered statistics. For any file $i$, the service time at each of the $n$ servers is exponential with mean $1/\mu_i$. Hence from \eqref{eqn:mvk_exp}, the mean and variance of $S_i$ are,
	\begin{align}
	{\mathbb{E}}[S_i] = \frac{H_{n-k_i,n}^1}{\mu_i},~\textnormal{V}[S_i] = \frac{H^2_{n-k_i,n}}{\mu_i^2}. \label{eq:meanVarSi}
	\end{align}

	Substituting \eqref{eq:meanVarSi} in \eqref{fcfs_ub1}, we get the following upper bound on average latency as in the statement of the theorem.	
\end{proof}

Without loss of generality, assume the files are relabeled such that $k_1 \le k_2 \le ... \le k_r$.  The next theorem provides the lower bound of the latency of file $i$. 
\begin{theorem}[\citep{kumar2014latency}]\label{Theorem2} 
	The average latency for file $i$  is lower-bounded as follows:
		\begin{equation}
		\label{eq:fcfs_lb}
		T^i \geq \sum\limits_{s=1}^{k_i} \left(\underbrace{\frac{t_{s,i}}{\lambda_i}}_\text{service time}+ \underbrace{\frac{\sum\limits_{j=c_{s,i}+1}^r \frac{t_{s,j}^2}{\lambda_j}}{1-\sum\limits_{j=c_{s,i}+1}^{r}t_{s,j}}}_\text{waiting time}\right),
		\end{equation}
		where $t_{s,i} = \frac{\lambda_i}{(n-s+1)\mu_i}$, and $c_{s,i}$ is given as
		\begin{align}
		c_{s,i} = \begin{cases} 0, &1\leq s \leq k_1\\ 1, &k_1 < s \leq k_2\\\vdots\\i-1, &k_{i-1} < s \leq k_i \end{cases}. \label{csdef}
		\end{align}
\end{theorem}
\begin{proof}
	For the purpose of obtaining a lower bound on the average latency of file $i$, using insights from proof of Theorem \ref{thm:lower_bnd}, we map the parallel processing in the  Fork-Join system to a sequential process consisting of $k_i$ processing stages for $k_i$ sub-tasks of a request of file $i$. The transition from one stage to the next occurs when one of the remaining servers finishes a sub-task of the file $i$. Note that $c_{s,i}$ in the theorem statement denotes the number of classes of file $i$ that are finished before start of stage $s$. 	The processing in each stage $s$ corresponds to a single-server FCFS system with jobs of all but classes $1, 2, \cdots, c_{s,i}$. Then, using Pollaczek-Khinchin formula at stage $s$, the average latency for a sub-task of a job of class $i$ in stage $s$ is given by, 
	\begin{align}
	T_{\textnormal{FCFS},s}^i &= {\mathbb{E}}[S_i^s]+ \frac{\lambda {\mathbb{E}}[{(S^s)}^2]}{2(1-\lambda {\mathbb{E}}[S^s]))}, \label{tmp}
	\end{align} 
	where $S^s$ is a r.v. denoting the service time for any sub-task in stage $s$ and $S_i^s$ denotes the service time for a sub-task of class $i$ in stage $s$, which are given as
	\begin{align}
	{\mathbb{E}}[S^s] = \sum\limits_{i=c_{s,i}+1}^R{p_i {\mathbb{E}}[S_i^s]},~~{\mathbb{E}}[{(S^s)}^2] = \sum\limits_{i=c_{s,i}+1}^R{p_i {\mathbb{E}}[{(S_i^s)}^2]}, 
	\label{meanVarSs}
	\end{align}
	where $p_i = \frac{\lambda_i}{\sum_{i=1}^r \lambda_i}$. 	Substituting \eqref{meanVarSs} in \eqref{tmp}, we get
	\begin{align}
	T_{s,c_{s,i}}^i &={\mathbb{E}}[S_i^s] + \frac{\sum\limits_{j=c_{s,i}+1}^r \lambda_j {\mathbb{E}}[{(S^s_j)}^2]}{2\left(1-\sum\limits_{j=c_{s,i}+1}^r \lambda_j {\mathbb{E}}[S^s_j]\right)}.
	\end{align}
	Now we note that at any stage $s$, the maximum possible service rate for a request of file $j$ that is not finished yet is $(n-s+1)\mu_j$. This happens when all the remaining sub-tasks of request of file $j$ are at the head of their buffers. Thus, we can enhance the latency performance in each stage $s$ by approximating it with a M/G/1 system with service rate $(n-s+1)\mu_j$ for request of file $j$. Then, the average latency for sub-task of request of file $i$ in stage $s$ is lower bounded as,
	\begin{align}
	T_{s,c_{s,i}}^i  &\geq \frac{1}{(n-s+1)\mu_i} + \frac{\sum\limits_{j=c_{s,i}+1}^r{\frac{\lambda_j}{(n-s+1)\mu_j^2}}}{1- \sum\limits_{j=c_{s,i}+1}^r{\frac{\lambda_j}{(n-s+1)\mu_j}}}, \label{latEnhcd}
	\end{align} 
	Finally, the average latency for file $i$ in this enhanced system is simply $\sum\limits_{s=1}^{k_i} T_{s,c_{s,i}}^i$. This gives us the result as in the statement of the theorem. 
\end{proof}
\section{Simulations}\label{sec:lat_fj_sims}
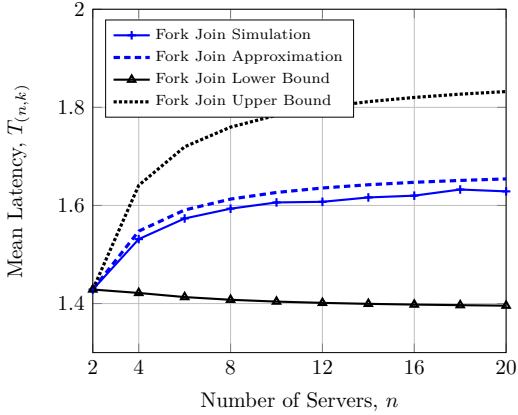
\begin{figure}[hh]
	\centering
	\scalebox{0.8}{\begin{tikzpicture}
\begin{axis}[
  font=\small,
  xlabel={Number of Servers, $n$},
  ylabel={Mean Latency, $T_{(n,k)}$},
  xmajorgrids,
  ymajorgrids,
  ymin=1.3,
  ymax=2,
  xmin=2,
  xmax=20,
  xtick = {2,4,8,12,16,20},
  legend entries={Fork Join Simulation, Fork Join Approximation, Fork Join Lower Bound, Fork Join Upper Bound},
  legend style={legend pos=north west,nodes=right,font=\scriptsize},
  ]

\addplot [color=blue, mark=+, line width=1pt]
coordinates {
(2,1.4286) (4,1.5313566666667098) (6,1.5734715734715716) (8,1.593404926738312) (10,1.6061482728148844) (12,1.6074216074215562) (14,1.6163217170116202) (16,1.6199683800316131) (18,1.6324866666667122) (20,1.6287016287015386)
};

\addplot [color=blue, mark=, densely dashed, line width=1.5pt]
coordinates {
(2,1.4286) (4,1.5476) (6,1.5907) (8,1.6130) (10,1.6265) (12,1.6356) (14,1.6422) (16,1.6472) (18,1.6510) (20,1.6541)
};

\addplot [color=black, mark=triangle,line width=1pt]
coordinates {
(2,1.4286) (4,1.4216 ) (6,1.4132) (8,1.4077) (10,1.4040) (12,1.4014 ) (14,1.3994) (16,1.3979) (18,1.3967) (20,1.3957) 
};

\addplot [color=black, mark=, densely dotted, line width=1.5pt]
coordinates {
(2,1.4286) (4,1.6410) (6,1.7196) (8,1.7597) (10,1.7839) (12,1.8001) (14,1.8116) (16,1.8202) (18,1.8269) (20,1.8322)
};

\end{axis}
\end{tikzpicture}}
	\caption{This graph displays the latency as the number of servers $n$ increases. Throughout,
		the code rate is kept constant at $k/n = 0.5$, the arrival rate is set to
		$\lambda= 0.3$, and the service rate of each server is $\mu=0.5$. The approximate result, upper bound, and lower bound in Chapter \ref{sec:lat_fj} are depicted along with the simulation results. 
	}
	\label{PS:Changen}
\end{figure}

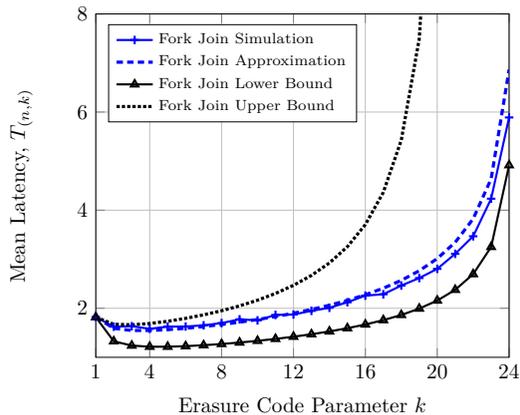
\begin{figure}[hh]
	\centering
	\scalebox{0.8}{\begin{tikzpicture}
\begin{axis}[
  font=\small,
  xlabel={Erasure Code Parameter $k$},
  ylabel={Mean Latency, $T_{(n,k)}$},
  xmajorgrids,
  ymajorgrids,
  ymin=1,
  ymax=8,
  xmin=1,
  xmax=24,
  xtick = {1,4,8,12,16,20,24},
  legend entries={Fork Join Simulation, Fork Join Approximation, Fork Join Lower Bound, Fork Join Upper Bound},
  legend style={legend pos=north west, nodes=right, font=\scriptsize},
  ]

\addplot [color=blue,mark=+,line width=1pt]
coordinates {
	(1, 1.8220) (2,1.6244) (3,1.6343) (4,1.5811) (5,1.6255) (6,1.6279) (7,1.6520) (8,1.6979) (9, 1.7742) (10,1.7534) (11,1.8639) (12,1.8726) (13,1.9474) (14,2.0099) (15,2.1224) (16, 2.2580) (17,2.2853) (18,2.4633) (19, 2.6160) (20, 2.8030) (21, 3.1098) (22, 3.4676) (23, 4.2335) (24, 5.8897)
};

\addplot [color=blue, mark=, densely dashed, line width=1.5pt]
coordinates {
 (1, 1.8182) (2,1.591) (3,1.5472) (4,1.5465) (5,1.5640) (6,1.5919) (7,1.6273) (8,1.6689) (9, 1.7166) (10,1.7704) (11,1.8308) (12,1.8986) (13,1.9750) (14,2.0616) (15,2.1605) (16, 2.2747) (17, 2.4086) (18, 2.5686) (19, 2.7644) (20, 3.0126) (21, 3.3438) (22, 3.8235) (23, 4.6383) (24,6.8654)
};

\addplot [color=black, mark=triangle, line width=1pt]
coordinates {
 (1, 1.8182) (2,1.3270) (3,1.2393) (4,1.2160) (5,1.2159) (6,1.2276) (7,1.2468) (8,1.2717) (9, 1.3015) (10, 1.3361 ) (11,1.3754) (12,1.4200) (13,1.4705) (14,1.5279 ) (15,1.5934) (16, 1.6692) (17,1.7579 ) (18,1.8637) (19, 1.9932) (20, 2.1571  ) (21, 2.3761 ) (22, 2.6955) (23, 3.2503) (24, 4.9162)
};

\addplot [color=black, mark=, densely dotted, line width=1.5pt]
coordinates {
(1, 1.8182) (2,1.67404) (3,1.6631) (4,1.6895) (5,1.73489) (6,1.794) (7,1.8654) (8,1.9497) (9, 2.0484) (10,2.1644) (11,2.301959) (12,2.4672486) (13, 2.6697) (14,2.9242475) (15,3.2552) (16, 3.7063) (17, 4.36409) (18, 5.4287) (19, 7.495565) (20, 13.50718)
};

\end{axis}
\end{tikzpicture}}
	\caption{This graph displays the latency as  $k$ increases. We let $n=24$, $\lambda = 0.45$, and $\mu = k/n$. The approximate result, upper bound, and lower bound in Chapter \ref{sec:lat_fj} are depicted along with the simulation results. 
	}
	\label{PS:Changek}
\end{figure}

We  evaluate the bounds in Chapter \ref{sec:lat_fj} for exponential service times, and compare them with the simulation results. In Figures \ref{PS:Changen} and \ref{PS:Changek}, we consider different parameter regimes to compare the bounds and the simulation results. We see that the approximation is close to the simulation results. The upper and lower bounds capture efficient bounds for the problem, while are still far from the actual latency. In Figure \ref{PS:Changen}, we change the value of the number of servers $n$, and change $k$ as $k=n/2$. The rest of the parameters are $\lambda= 0.3$ and $\mu=0.5$. In Figure \ref{PS:Changek}, we let $n=24$, $\lambda = 0.45$, and $\mu = k/n$. On increasing $k$, the latency increases as is illustrated in the figure with the bounds.

\section{Notes and Open Problems}\label{sec:fj_notes}

The $(n, k)$ fork-join system was first proposed in \citep{Joshi:13} to analyze content download latency from erasure coded distributed storage for exponential service times. They consider that a content file coded into $n$ chunks can be recovered by accessing any $k$ out of the $n$
chunks, where the service time of each chunk is exponential. Even with the exponential assumption
analyzing the $(n, k)$ fork-join system is a hard problem. It is a generalization of the $(n, n)$ fork-join
system, which was actively studied in queueing literature \citep{flatto1984two,nelson1988approximate,varki2008m}. The results were extended to general service time distributions in \citep{joshi2017efficient}. For exponential service times, approximate latency was studied in \citep{badita2019latency}. These works assume homogeneous files, in the sense that each file has the same arrival  distributions,  have the same erasure-coding parameters, and run on the servers with same service distribution. For exponential service times, the authors of \citep{kumar2014latency} studied the case when the different files have different arrival rates and erasure-code parameters. However, all these works assume that the number of servers is the same as the $n$, which is the erasure coding parameter representing the number of encoded chunks. Further, the servers are assumed to be homogeneous,  with the same service rate. Thus, the following problems are of interest. 

\begin{enumerate}
	\item {\bf Tighter Upper Bounds}: We note that even for the exponential service times, the studied upper bounds do not meet the optimal stability conditions. Thus, an efficient upper bound that satisfies the stability conditions is open.  
	\item {\bf General File Placement}: In the prior work, the number of servers are the same as the erasure-coded encoded chunks. However, in general, the number of servers may be large and each file $i$ may be placed on a subset $n_i$ of the servers. The latency results for a general parameter system has not been considered. A related concept to the placement of the files is that it is not essential,  in general, to have only one chunk per node, some nodes may have more chunks. In this case, the $n$ requests are not in parallel and the same analysis cannot be easily extended. 
	\item {\bf Heterogeneous Servers}: In the prior works, the servers serving the files are assumed to be homogeneous with the same service distribution. However, this is not the case in general, especially with fog computing. Obtaining efficient analysis for such  heterogeneous server system is an open problem. 
	\item {\bf Approximation and its Guarantees}: While an approximation of latency has been proposed for exponential service times \citep{badita2019latency}, such characterization for heterogenous files and general service times is open. Further, the guarantees on approximation, in some asymptotic regime or bounding the gap between the two by a constant (or within a multiplicative factor) has not yet been considered.
\end{enumerate}

\chapter{Probabilistic Scheduling Approach}

In this Chapter, we introduce the model of Probabilistic Scheduling in Section~\ref{sec:prob_sch}, which was first proposed in \citep{Yu_TON,Xiang:2014:Sigmetrics:2014}. We will evaluate an upper bound on the latency for heteogeneous files, heterogenous servers, general service times, and general placemeny of files in Section \ref{mean_ps}. A notion of tail latency is provided and an upper bound on tail latency is characterized in Section \ref{tail_ps}. For homogenous files, homogenous servers, Section \ref{ps_asymp} shows that the latency of uniform probabilistic scheduling is upper bounded by assuming independence across the servers. Further, asymptotic latency is considered as the number of servers increase. Another asymptotic metric for heavy tailed file sizes is provided and analyzed in Section \ref{ps_tail_asym}. Sections~\ref{sec:ps_sims} and \ref{sec:ps_notes} contain simulation results and notes on future directions, respectively.

\section{Probabilistic Scheduling}\label{sec:prob_sch}

We assume the model given in Section \ref{model}. Under $(n_i,k_i)$ MDS codes, each file $i$ can be retrieved by processing a batch of $k_i$ chunk requests at distinct nodes that store the file chunks. Recall that each encoded file $i$ is spread over $n_i$ nodes, denoted by a set $\mathcal{S}_i$. Upon the arrival of a file $i$ request, in probabilistic scheduling we randomly dispatch the batch of $k_i$ chunk requests to $k_i$ out of $n_i$ storage nodes in $\mathcal{S}_i$, denoted by a subset $\mathcal{A}_i\subseteq \mathcal{S}_i$ (satisfying  $|\mathcal{A}_i|=k_i$) with predetermined probabilities. Then, each storage node manages its local queue independently and continues processing requests in order. A file request is completed if all its chunk requests exit the system. An example of probabilistic scheduling is depicted in Fig.~\ref{fig:sysmodel2}  for the setup in Section \ref{sec:chall}, wherein 5 chunk requests are currently served by the 5 storage nodes, and there are 9 more chunk requests that are randomly dispatched to and are buffered in 5 local queues according to chunk placement, e.g., requests $B_2,B_3$ are only distributed to nodes $\{3,4,5\}$. Suppose that node 2 completes serving chunk request $A_2$. The next request in the node's local queue will move forward.

\vspace{-3mm}
\begin{figure}[!thbp]
	\begin{center}
		\scalebox{0.32}{\includegraphics[trim=4.5in .3in 0in 0in, clip, draft=false]{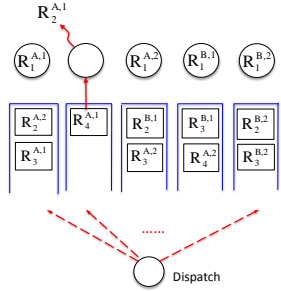}}
		\vspace{-3mm}
		\caption{Functioning of  a probabilistic scheduling policy.}
		\label{fig:sysmodel2}
	\end{center}
	\vspace{-.2in}
\end{figure}

\begin{definition} {\em (Probabilistic scheduling)} A Probabilistic scheduling policy (i) dispatches each batch of chunk requests to appropriate nodes with predetermined probabilities; (ii) each node buffers requests in a local queue and processes in order.
\end{definition}

It is easy to verify that such probabilistic scheduling ensures that at most 1 chunk request from a batch to each appropriate node. It provides an upper bound on average service latency for the optimal scheduling since rebalancing and scheduling of local queues are not permitted. Let $\mathbb{P}(\mathcal{A}_i)$ for all $\mathcal{A}_i\subseteq \mathcal{S}_i$ be the probability of selecting a set of nodes $\mathcal{A}_i$ to process the $|\mathcal{A}_i|=k_i$ distinct chunk requests\footnote{It is easy to see that $\mathbb{P}(\mathcal{A}_i)=0$ for all $\mathcal{A}_i\nsubseteq \mathcal{S}_i$ and $|\mathcal{A}_i|=k_i$ because such node selections do not recover $k_i$ distinct chunks and thus are inadequate for successful decode.}.

\begin{lemma} \label{th:lemma_bound}
For given erasure codes and chunk placement, average service latency of probabilistic scheduling with feasible probabilities $\{\mathbb{P}(\mathcal{A}_i): \ \forall i,\mathcal{A}_i\}$ upper bounds the latency of optimal scheduling.
\end{lemma}

Clearly, the tightest upper bound can be obtained by minimizing average latency of probabilistic scheduling over all feasible probabilities $\mathbb{P}(\mathcal{A}_i)$ $\forall \mathcal{A}_i\subseteq \mathcal{S}_i$ and $\forall i$, which involves $\sum_i {n_i\choose k_i}$ decision variables. We refer to this optimization as a {\em scheduling subproblem}. While it appears prohibitive computationally, we will demonstrate next that the optimization can be transformed into an equivalent form, which only requires $\sum_i n_i$ variables. The key idea is to show that it is sufficient to consider the conditional probability (denoted by $\pi_{i,j}$) of selecting a node $j$, given that a batch of $k_i$ chunk requests of file $i$ are dispatched. It is easy to see that for given $\mathbb{P}(\mathcal{A}_i)$, we can derive $\pi_{i,j}$ by
\vspace{-2mm}
\begin{eqnarray}
\pi_{i,j} = \sum_{\mathcal{A}_i:\mathcal{A}_i\subseteq \mathcal{S}_i} \mathbb{P}(\mathcal{A}_i) \cdot {\bf 1}_{\{ j\in\mathcal{A}_i\}}, \ \forall i \label{eq:pi}
\vspace{-2mm}
\end{eqnarray}
where ${\bf 1}_{\{ j\in\mathcal{A}_i\}}$ is an indicator function which equals to 1 if node $j$ is selected by $\mathcal{A}_i$ and 0 otherwise.

We first provide Farkas-Minkowski Theorem \citep{Angell:02} that will be used in this transformation. 
\begin{lemma} Farkas-Minkowski Theorem \citep{Angell:02}. Let ${\bf A}$ be an $m\times n$ matrix with real entries, and ${\bf x}\in \mathbb{R}^{n}$ and ${\bf b}\in \mathbb{R}^{m}$ be 2 vectors. A necessary and sufficient condition that ${\bf A}\cdot{\bf x}={\bf b}, \ {\bf x}\ge 0$ has a solution is that, for all ${\bf y}\in \mathbb{R}^m$ with the property that ${\bf A}^T \cdot {\bf y}\ge 0$, we have  $\langle{\bf y}, {\bf b}\rangle \ge 0$.
\end{lemma}
The next result formally shows that the optimization can be transformed into an equivalent form, which only requires $\sum_i n_i$ variables. 
\begin{theorem} \label{th:thm_prob}
A probabilistic scheduling policy with feasible probabilities $\{\mathbb{P}(\mathcal{A}_i): \ \forall i,\mathcal{A}_i\}$ exists if and only if there exists conditional probabilities $\{\pi_{i,j}\in [0,1], \forall i,j\}$ satisfying
\vspace{-2mm}
\begin{eqnarray}
\sum_{j=1}^m \pi_{i,j} = k_i \ \forall i \ {\rm and} \ \pi_{i,j}=0 \ {\rm if} \ j\notin \mathcal{S}_i. \label{eq:thm_prob}
\vspace{-2mm}
\end{eqnarray}

\end{theorem}
\begin{proof}
	We first prove that the conditions $\sum_{j=1}^m \pi_{i,j} = k_i$ $\forall i$ and $\pi_{i,j}\in [0,1]$ are necessary. $\pi_{i,j}\in [0,1]$ for all $i,j$ is obvious due to its definition. Then, it is easy to show that
	\begin{eqnarray}
	& \D \sum_{j=1}^m \pi_{i,j} & = \sum_{j=1}^m \sum_{\mathcal{A}_i\subseteq \mathcal{S}_i} {\bf 1}_{\{ j\in\mathcal{A}_i\}} \mathbb{P}(\mathcal{A}_i) \nonumber \\
	&  & =\sum_{\mathcal{A}_i\subseteq \mathcal{S}_i} \sum_{j\in\mathcal{A}_i}  \mathbb{P}(\mathcal{A}_i) \nonumber \\
	& & =\sum_{\mathcal{A}_i\subseteq \mathcal{S}_i}  k_i \mathbb{P}(\mathcal{A}_i)  = k_i \label{eq:k_1}
	\end{eqnarray}
	where the first step is due to  (\ref{eq:pi}), ${\bf 1}_{\{ j\in\mathcal{A}_i\}}$ is an indicator function, which is 1 if $j\in\mathcal{A}_i$, and 0 otherwise. The second step changes the order of summation, the last step uses the fact that each set $\mathcal{A}_i$ contain exactly $k_i$ nodes and that $\sum_{\mathcal{A}_i\subseteq \mathcal{S}_i} \mathbb{P}(\mathcal{A}_i)=1$.
	
	Next, we prove that for any set of $\pi_{i,1},\ldots,\pi_{i,m}$ (i.e., node selection probabilities of file $i$) satisfying $\sum_{j=1}^m \pi_{i,j} = k_i$ and $\pi_{i,j}\in [0,1]$, there exists a probabilistic scheduling scheme with feasible load balancing probabilities $\mathbb{P}(\mathcal{A}_i)$ $\forall \mathcal{A}_i \subseteq \mathcal{S}_i$ to achieve the same node selection probabilities. We start by constructing $\mathcal{S}_i=\{j: \pi_{i,j} >0\}$, which is a set containing at least $k_i$ nodes, because there must be at least $k_i$ positive probabilities $\pi_{i,j}$ to satisfy $\sum_{j=1}^m \pi_{i,j} = k_i$. Then, we choose erasure code length $n_i=|\mathcal{S}_i|$ and place chunks on nodes in  $\mathcal{S}_i$. From (\ref{eq:pi}), we only need to show that when $\sum_{j\in\mathcal{S}_i} \pi_{i,j} = k_i$ and $\pi_{i,j}\in [0,1]$, the following system of $n_i$ linear equations have a feasible solution $\mathbb{P}(\mathcal{A}_i)$ $\forall \mathcal{A}_i\subseteq \mathcal{S}_i$:
	\begin{eqnarray}
	\sum_{\mathcal{A}_i\subseteq \mathcal{S}_i} {\bf 1}_{\{ j\in\mathcal{A}_i\}} \cdot \mathbb{P}(\mathcal{A}_i)=\pi_{i,j}, \ \forall j\in \mathcal{S}_i \label{eq:pi_1}
	\end{eqnarray}
	We prove the desired result using mathematical induction. It is easy to show that the statement holds for $n_i=k_i$. In this case, we have a unique solution $\mathcal{A}_i = \mathcal{S}_i$ and $\mathbb{P}(\mathcal{A}_i)=\pi_{i,j}=1$ for the system of linear equations (\ref{eq:pi_1}), because all chunks must be selected to recover file $i$. Now assume that the system of linear equations (\ref{eq:pi_1}) has a feasible solution for some $n_i\ge k_i$. Consider the case with arbitrary  $|\mathcal{S}_i+\{h\}|=n_i+1$ and $\pi_{i,h}+\sum_{j\in\mathcal{S}_i} \pi_{i,j}= k_i$. We have a system of linear equations:
	\begin{eqnarray}
	\sum_{\mathcal{A}_i\subseteq \mathcal{S}_i+\{h\}} {\bf 1}_{\{ j\in\mathcal{A}_i\}} \cdot \mathbb{P}(\mathcal{A}_i)=\pi_{i,j}, \ \forall j\in \mathcal{S}_i+\{h\} \label{eq:pi_2}
	\end{eqnarray}
	Using the Farkas-Minkowski Theorem \citep{Angell:02}, a sufficient and necessary condition that (\ref{eq:pi_2}) has a non-negative solution is that, for any $y_1,\ldots, y_m$ and $\sum_{j} y_j \pi_{i,j}<0$, we have
	\begin{eqnarray}
	\sum_{j\in \mathcal{S}_i+\{h\}} y_j {\bf 1}_{\{ j\in\mathcal{A}_i\}}  < 0 \ {\rm for \ some} \ \mathcal{A}_i\subseteq \mathcal{S}_i+\{h\}.\label{eq:y_1}
	\end{eqnarray}
	
	Toward this end, we construct $\hat{\pi}_{i,j} =\pi_{i,j} + [u-\pi_{i,j}]^{+}$ for all $j\in \mathcal{S}_i$. Here $[x]^{+}=\max(x, 0)$ is a truncating function and $u$ is a proper {\em water-filling level} satisfying
	\begin{eqnarray}
	\sum_{j\in \mathcal{S}_i}  [u-\pi_{i,j}]^{+} = \pi_{i,h}.  \label{eq:h_1}
	\end{eqnarray}
	It is easy to show that $\sum_{j\in \mathcal{S}_i} \hat{\pi}_{i,j}= \pi_{i,h} +\sum_{j\in \mathcal{S}_i} \pi_{i,j} =k_i$ and $\hat{\pi}_{i,j}\in[0,1]$, because $\hat{\pi}_{i,j} = \max(u,{\pi}_{i,j} ) \in[0, 1]$. Here we used the fact that $u<1$ since $k_i=\sum_{j\in \mathcal{S}_i} \hat{\pi}_{i,j}\ge\sum_{j\in \mathcal{S}_i} u \ge k_iu  $.  Therefore, the system of linear equations in (\ref{eq:pi_1}) with $\hat{\pi}_{i,j}$ on the right hand side must have a non-negative solution due to our induction assumption for $n_i=|\mathcal{S}_i|$. Furthermore, without loss of generality, we assume that $y_h\ge y_j$ for all $j\in \mathcal{S}_i$ (otherwise a different $h$ can be chosen). It implies that
	\begin{eqnarray}
	& \D \sum_{j\in \mathcal{S}_i} y_j \hat{\pi}_{i,j} & =  \sum_{j\in \mathcal{S}_i} y_j (\pi_{i,j}+[u-\pi_{i,j}]^{+}) \nonumber \\
	& & \stackrel {(a)} {\le} \sum_{j\in \mathcal{S}_i} y_j \pi_{i,j} + \sum_{j\in \mathcal{S}_i} y_h [u-\pi_{i,j}]^{+}  \nonumber \\
	& & \stackrel {(b)} {=}\sum_{j\in \mathcal{S}_i} y_j \pi_{i,j} + y_h \sum_{j\in \mathcal{S}_i} [u-\pi_{i,j}]^{+} \nonumber \\
	& & \stackrel {(c)} {=} \sum_{j\in \mathcal{S}_i} y_j \pi_{i,j} + y_h \pi_{i,h} \stackrel {(d)} {\le} 0, \label{eq:h_2}
	\end{eqnarray}
	where \em{(a)} uses $y_h\ge y_j$, \em{(b)} uses that $y_h$ is independent of $j$, \em{(c)} follows from (\ref{eq:h_1}) and the last step uses $\sum_{j} y_j \pi_{i,j}<0$.
	
	Applying the Farkas-Minkowski Theorem to the system of linear equations in (\ref{eq:pi_1}) with $\hat{\pi}_{i,j}$ on the right hand side, the existence of  a non-negative solution (due to our induction assumption for $n_i$) implies that $\sum_{j\in \mathcal{S}_i} y_j {\bf 1}_{\{ j\in\mathcal{A}_i\}}  < 0$ for some $\hat{\mathcal{A}_i}\subseteq \mathcal{S}_i$. It means that
	\begin{eqnarray}
	\sum_{j\in \mathcal{S}_i+\{h\}} y_j {\bf 1}_{\{ j\in\hat{\mathcal{A}}_i\}}  = y_h {\bf 1}_{\{ h\in\hat{\mathcal{A}}_i\}}+ \sum_{j\in \mathcal{S}_i} y_j {\bf 1}_{\{ j\in\hat{\mathcal{A}}_i\}} <0.
	\end{eqnarray}
	The last step uses ${\bf 1}_{\{ h\in\hat{\mathcal{A}}_i\}}=0$ since $h\notin\mathcal{S}_i$ and $\hat{\mathcal{A}_i}\subseteq \mathcal{S}_i$. This is exactly the desired inequality in (\ref{eq:y_1}). Thus, (\ref{eq:pi_2}) has a non-negative solution due to the Farkas-Minkowski Theorem. The induction statement holds for $n_i+1$. Finally, the solution indeed gives a probability distribution since $\sum_{\mathcal{A}_i\subseteq \mathcal{S}_i+\{h\}}  \mathbb{P}(\mathcal{A}_i) =\sum_j \pi_{i,j} /k_i=1$ due to (\ref{eq:k_1}). This completes the proof.
	\end{proof}
The proof of Theorem~\ref{th:thm_prob} relies on Farkas-Minkowski Theorem \citep{Angell:02}. Intuitively, $\sum_{j=1}^m \pi_{i,j} = k_i$ holds because each batch of requests is dispatched to exact $k_i$ distinct nodes. Moreover, a node does not host file $i$ chunks should not be selected, meaning that $ \pi_{i,j}=0$ if $j\notin \mathcal{S}_i$. Using this result, it is sufficient to study probabilistic scheduling via conditional probabilities $\pi_{i,j}$, which greatly simplifies our analysis. In particular, it is easy to verify that under our model, the arrival of chunk requests at node $j$ form a Poisson Process with rate $\Lambda_j=\sum_i \lambda_i\pi_{i,j}$, which is the superposition of $r$ Poisson processes each with rate $\lambda_i\pi_{i,j}$, $\mu_j$ is the service rate of node $j$. The resulting queuing system under probabilistic scheduling is stable if all local queues are stable.

\begin{corollary} \label{th:stability}
	The queuing system can be stabilized by a probabilistic scheduling policy under request arrival rates $\lambda_1,\lambda_2,\ldots,\lambda_r$ if there exists $\{\pi_{i,j}\in [0,1], \forall i,j\}$ satisfying (\ref{eq:thm_prob}) and
	\vspace{-2mm}
	\begin{eqnarray}
	\Lambda_j = \sum_i \lambda_i\pi_{i,j} < \mu_j, \ \forall j. \label{eq:lambda}
	\end{eqnarray}
\end{corollary}

We let uniform probabilistic scheduling be defined as the case when $\pi_{i,j} = k_i/n_i$ for $j\in {\mathcal S}_i$. If all $k_i=k$, $n_i=n$, $m=n$, $r=1$, $\lambda_1=\lambda$, $\mu_j=\mu$, we have the stable region as $\lambda < n\mu/k$, which is the same as for the case of Fork-Join Queues, while for general service times. In contrast, this stability region is only valid for Fork-Join queues for only exponential service times. 
\section{Characterization of Mean Latency}\label{mean_ps}
An exact analysis of the queuing latency of probabilistic scheduling is still hard because local queues at different storage nodes are dependent of each other as each batch of chunk requests are dispatched jointly.

Since local queues at different storage nodes are dependent of each
other as each batch of chunk requests are jointly dispatched,
the exact analysis of the queuing latency of probabilistic scheduling
is not tractable. Thus, we will use probabilistic scheduling to bound
the mean latency and since probabilistic scheduling is a feasible
strategy, the obtained bound is an upper bound to the optimal strategy.
We define ${\bf W}_{i,j}$ as the random waiting time (sojourn
time) in which a chunk request (for file $i$) spends in the queue
of node $j$. Typically, the  latency of file $i$, denoted
as ${\bf Q}_{i}$, request is determined by the maximum latency
that $k_{i}$ chunk requests experience on distinct servers. These
servers are probabilistically scheduled with a prior known probabilities,
i.e., $\pi_{i,j}$. Thus, we have 

\begin{equation}
\mathbb{E}[{\bf Q}_{i}]\triangleq\mathbb{E}_{{\bf W}_{i,j}}\left[\mathbb{E}_{\mathcal{A}_{i}}\left[\underset{j\in\mathcal{A}_{i}}{\text{max}}\,{\bf W}_{i,j}\right]\right]\label{eq:Q_j_max}
\end{equation}
where the first expectation $\mathbb{E}_{{\bf W}_{j}}$ is taken
over system queuing dynamics and the second expectation $\mathbb{E}_{\mathcal{A}_{i}}$
is taken over random dispatch decisions $\mathcal{A}_{i}$. Hence,
we derive an upper-bound on the expected latency of a file $i$, i.e.,
$\mathbb{E}[{\bf Q}_{i}]$, as follows. Using Jensen\textquoteright s inequality
\citep{kuczma2009introduction1}, we have for $t_{i}>0$

\begin{equation}
e^{t_{i}\mathbb{E}\left[{\bf Q}_{i}\right]}\leq\mathbb{E}\left[e^{t_{i}{\bf Q}_{i}}\right]\label{eq:jenIneq}
\end{equation}

We notice from (\ref{eq:jenIneq}) that by bounding the moment generating
function of ${\bf Q}_{i}$, we are bounding the mean latency of
file $i$. Then, 

\begin{eqnarray}
\mathbb{E}\left[e^{t_{i}{\bf Q}_{i}}\right] & \overset{(a)}{=} & \mathbb{E}_{\mathcal{A}_{i},{\bf W}_{i,j}}\left[\underset{j\in\mathcal{A}_{i}}{\text{max}}\,e^{t_{i}{\bf W}_{i,j}}\right]\\
& = & \mathbb{E}_{\mathcal{A}_{i}}\left[\mathbb{E}_{{\bf W}_{i,j}}\left[\underset{j\in\mathcal{A}_{i}}{\text{max}}\,e^{t_{i}{\bf W}_{i,j}}\left|\mathcal{A}_{i}\right.\right]\right]\\
& \overset{(b)}{\leq} & \mathbb{E}_{\mathcal{A}_{i}}\left[{\displaystyle \sum_{j\in\mathcal{A}_{i}}}{\displaystyle \mathbb{E}_{{\bf W}_{i,j}}}\left[e^{t_{i}{\bf W}_{i,j}}\right]\right]
\end{eqnarray}
\begin{eqnarray}
& = & \mathbb{E}_{\mathcal{A}_{i}}\left[{\displaystyle \sum_{j}}{\displaystyle \mathbb{E}_{{\bf W}_{i,j}}}\left[e^{t_{i}{\bf W}_{i,j}}\right]\text{\ensuremath{\boldsymbol{1}}}_{(j\in\mathcal{A}_{i})}\right]\\
& = & {\displaystyle \sum_{j}}{\displaystyle \mathbb{E}_{{\bf W}_{i,j}}}\left[e^{t_{i}{\bf W}_{i,j}}\right]\mathbb{E}_{\mathcal{A}_{i}}\left[\text{\ensuremath{\boldsymbol{1}}}_{(j\in\mathcal{A}_{i})}\right]\\
& = & {\displaystyle \sum_{j}}{\displaystyle \mathbb{E}_{{\bf W}_{i,j}}}\left[e^{t_{i}{\bf W}_{i,j}}\right]\mathbb{P}(j\in\mathcal{A}_{i})\\
& \overset{(c)}{=} & {\displaystyle \sum_{j}\pi_{i,j}}{\displaystyle \mathbb{E}_{{\bf W}_{i,j}}}\left[e^{t_{i}{\bf W}_{i,j}}\right]\label{eq:MGF_Q_j}
\end{eqnarray}
where $(a)$ follows from (\ref{eq:Q_j_max}) and (\ref{eq:jenIneq}),
$(b)$ follows by replacing the $\underset{j\in\mathcal{A}_{i}}{\text{max}}$
by ${\displaystyle \sum_{j\in\mathcal{A}_{i}}}$ and $(c)$ follows
by probabilistic scheduling. We note that the only inequality here
is for replacing the maximum by the sum. However, since this term
will be inside the logarithm for the mean latency, the gap between
the term and its bound becomes additive rather than multiplicative.
Since the request pattern is Poisson and the service time is general
distributed, the Laplace-Stieltjes Transform of the waiting time ${\bf W}_{i,j}$
can be characterized using Pollaczek-Khinchine formula for M/G/1 queues
\citep{zwart2000sojourn1} as follows

\begin{equation}
{\displaystyle \mathbb{E}}\left[e^{t_{i}{\bf W}_{i,j}}\right]=\frac{\left(1-\rho_{j}\right)t_{i}\mathbb{Z}_{j}(t_{i})}{t_{i}-\Lambda_{j}\left(\mathbb{Z}_{j}(t_{i})-1\right)}\label{eq:mg1_waitingTime}
\end{equation}
where $\rho_{j}=\Lambda_{j}{\displaystyle \mathbb{E}}\left[\mathbb{X}_{j}\right]=\Lambda_{j}\left[\frac{d}{dt}\mathbb{Z}_{j}(t_{i})\left|_{t_{i}=0}\right.\right]$
and $\mathbb{Z}_{j}(t_{i})$ is the moment generating function of the chunk service time. Plugging
(\ref{eq:mg1_waitingTime}) in (\ref{eq:MGF_Q_j}) and substituting
in (\ref{eq:jenIneq}), we get the following Theorem. 

\begin{theorem}\label{psmeangen}
	The mean latency for file $i$ is bounded by
	
	\begin{equation}
	\mathbb{E}[{\bf Q}_{i}]\leq\frac{1}{t_{i}}\mbox{log}\left(\sum_{j=1}^{m}\pi_{i,j}\frac{(1-\rho_{j})t_{i}\mathbb{Z}_{j}(t_{i})}{t_{i}-\Lambda_{j}\left(\mathbb{Z}_{j}(t_{i})-1\right)}\right)\label{eq:upper_bound_T-1}
	\end{equation}
	
	for any $t_{i}>0$, $\rho_{j}=\Lambda_{j}\left[\frac{d}{dt}\mathbb{Z}_{j}(t_{i})\left|_{t_{i}=0}\right.\right]$, $\rho_j<1$,
	and $\Lambda_{j}\left(\mathbb{Z}_{j}(t_{i})-1\right)<t_{i}$. 
\end{theorem}
Note that the above Theorem holds only in the range of $t_{i}$ when
$t_{i}-\Lambda_{j}\left(\mathbb{Z}_{j}(t)-1\right)>0$.
Further, the server utilization $\rho_{j}$ must be less than $1$
for stability of the system.

We now specialize this result for shifted exponential distribution. Let the service time distribution from server $j$, $\mathbb{X}_{j}$, has probability density function $f_{\mathbb{X}_{j}}(x)$, given as

\begin{equation}
f_{\mathbb{X}_{j}}(x) = \begin{cases}
\alpha_j e^{-\alpha_j(x-\beta_j)} & {\text{ for }} x\ge \beta_j\\
0 & {\text{ for }} x<\beta_j
\end{cases}\label{dist_sexp}
\end{equation}
Exponential distribution is a special case for $\beta_j=0$. The moment generating function, $\mathbb{Z}_{j}(t)$ is given as
\begin{equation}
\mathbb{Z}_{j}(t) = \frac{\alpha_j}{\alpha_j-t} e^{\beta_j t} \text{ for } t<\alpha_j. \label{sexpmgf}
\end{equation}
Using \eqref{sexpmgf} in Theorem \ref{psmeangen}, we have

\begin{corollary}
	The mean latency for file $i$ for Shifted Exponential Service time at each server is bounded by
	
	\begin{equation}
	\mathbb{E}[{\bf Q}_{i}]\leq\frac{1}{t_{i}}\mbox{log}\left(\sum_{j=1}^{m}\pi_{i,j}\frac{(1-\rho_{j})t_{i}\mathbb{Z}_{j}(t_{i})}{t_{i}-\Lambda_{j}\left(\mathbb{Z}_{j}(t_{i})-1\right)}\right)
	\end{equation}
	
	for any $t_{i}>0$, $\rho_{j}=\Lambda_{j}\left(\frac{1}{\alpha_j}+\beta_j\right)$, $\rho_j<1$,
	 $t_i(t_i-\alpha_j+\Lambda_j) +\Lambda_j \alpha_j(e^{\beta_j t_i}-1)<0$, and $\mathbb{Z}_{j}(t) = \frac{\alpha_j}{\alpha_j-t} e^{\beta_j t}$. 
\end{corollary}

Further, the exponential distribution has $\beta_j=0$, and the result for exponential follows as a special case. We note that the bound presented here has been shown to outperform that in \citep{Yu_TON,Xiang:2014:Sigmetrics:2014} in \citep{Abubakr_meantail}. The comparison between the two is further illustrated in Figure \ref{lat_vs_arrRate}.  The bound in \citep{Yu_TON,Xiang:2014:Sigmetrics:2014}  is further shown to be better than that in \citep{Joshi:13}. Moreover, replication coding follows as a special case when $k_{i}=1$
and thus the proposed upper bound for file download can be used to
bound the latency of replication based systems by setting $k_{i}=1$. 

For completeness, we will also present the  latency bound provided in \citep{Yu_TON,Xiang:2014:Sigmetrics:2014}.

\begin{theorem}
	The mean latency of file $i$ is bounded by 
	\begin{equation}
		\mathbb{E}[{\bf Q}_{i}] \le z_i + \frac{1}{2}\underset{j\in\mathcal{A}_{i}}{\sum}\, \pi_{i,j} \left[{\mathbb E}[{\bf W}_{i,j}]-z_i + \sqrt{({\mathbb E}[{\bf W}_{i,j}]-z_i)^2 + {\text Var}[{\bf W}_{i,j}]}\right],
	\end{equation}
where ${\mathbb E}[{\bf W}_{i,j}]$ and ${\text Var}[{\bf W}_{i,j}]$ are given from calculating moments of the moment generating function in  \eqref{eq:mg1_waitingTime}. \label{old_upper_bound}
\end{theorem}
\begin{proof}
	\begin{eqnarray}
	\mathbb{E}[{\bf Q}_{i}]&=&\mathbb{E}_{{\bf W}_{i,j}}\left[\mathbb{E}_{\mathcal{A}_{i}}\left[\underset{j\in\mathcal{A}_{i}}{\text{max}}\,{\bf W}_{i,j}\right]\right]\nonumber\\
	&\le & \mathbb{E}_{{\bf W}_{i,j}}\left[\mathbb{E}_{\mathcal{A}_{i}}\left[ z_i + \left[\underset{j\in\mathcal{A}_{i}}{\text{max}}\,{\bf W}_{i,j}-z_i\right]^+\right]\right]\nonumber\\
	&=& \mathbb{E}_{{\bf W}_{i,j}}\left[\mathbb{E}_{\mathcal{A}_{i}}\left[ z_i + \underset{j\in\mathcal{A}_{i}}{\text{max}}\,\left[{\bf W}_{i,j}-z_i\right]^+\right]\right]\nonumber\\
	&\le& \mathbb{E}_{{\bf W}_{i,j}}\left[\mathbb{E}_{\mathcal{A}_{i}}\left[ z_i + \underset{j\in\mathcal{A}_{i}}{\sum}\,\left[{\bf W}_{i,j}-z_i\right]^+\right]\right]\nonumber\\
	&=& \mathbb{E}_{{\bf W}_{i,j}}\left[\mathbb{E}_{\mathcal{A}_{i}}\left[ z_i + \frac{1}{2}\underset{j\in\mathcal{A}_{i}}{\sum}\,\left[{\bf W}_{i,j}-z_i + |{\bf W}_{i,j}-z_i|\right]\right]\right]\nonumber\\
		&=& \mathbb{E}_{{\bf W}_{i,j}}\left[ z_i + \frac{1}{2}\underset{j\in\mathcal{A}_{i}}{\sum}\, \pi_{i,j} \left[{\bf W}_{i,j}-z_i + |{\bf W}_{i,j}-z_i|\right]\right]\nonumber
	\end{eqnarray}
We note that ${\mathbb E}[{\bf W}_{i,j}]$ can be found from \eqref{eq:mg1_waitingTime}. Further, ${\mathbb E}[ |{\bf W}_{i,j}-z_i|]$ can be upper bounded as 
\begin{equation}
{\mathbb E}[ |{\bf W}_{i,j}-z_i|]\le \sqrt{({\mathbb E}[{\bf W}_{i,j}]-z_i)^2 + {\text Var}[{\bf W}_{i,j}]},
\end{equation}
where both ${\mathbb E}[{\bf W}_{i,j}]$ and ${\text Var}[{\bf W}_{i,j}]$ can be found using \eqref{eq:mg1_waitingTime}. 
\end{proof}

\section{Characterization of Tail Latency}\label{tail_ps}
Latency tail probability of file $i$ is defined as the probability
that the latency tail is greater than (or equal) to a given number
$\sigma$, i.e., Pr$\left({\bf Q}_{i}\geq\sigma\right)$. Since
evaluating Pr$\left({\bf Q}_{i}\geq\sigma\right)$ in closed-form
is hard \citep{MDS_queue,Yu_TON,Xiang:2014:Sigmetrics:2014,CS14,MG1:12},
we derive a tight upper bound on the latency tail probability using
Probabilistic Scheduling as follows \citep{Jingxian,al2019ttloc}.

\begin{eqnarray}
\text{Pr}\left({\bf Q}_{i}\geq\sigma\right) & \overset{(d)}{=} & \Pr\left(\max_{j\in\mathcal{A}_{i}}{\bf W}_{i,j}\ge\sigma\right)\\
 & = & \mathbb{E}_{\mathcal{A}_{i}}\left[\mathbb{E}_{{\bf W}_{i,j}}\left[\max_{j\in\mathcal{A}_{i}}\,{\bf W}_{i,j}\ge\sigma\left|\mathcal{A}_{i}\right.\right]\right]\\
 & = & \mathbb{E}_{\mathcal{A}_{i},{\bf W}_{i,j}}\left[\max_{j\in\mathcal{A}_{i}}\boldsymbol{1}_{({\bf W}_{i,j}\ge\sigma)}\right]\\
 & \overset{(e)}{\le} & \mathbb{E}_{\mathcal{A}_{i},{\bf W}_{i,j}}\left[\sum_{j\in\mathcal{A}_{i}}\left[1_{({\bf W}_{i,j}\ge\sigma)}\right]\right]\\
 & = & \mathbb{E}_{\mathcal{A}_{i}}\left[\sum_{j\in\mathcal{A}_{i}}\left[\Pr({\bf W}_{i,j}\ge\sigma)\right]\right]\\
 & \overset{(f)}{=} & \sum_{j}\pi_{i,j}\left[\Pr({\bf W}_{i,j}\ge\sigma)\right]\label{eq:P_Qi_sigma}
\end{eqnarray}
where $(d)$ follows from (\ref{eq:Q_j_max})\footnote{As the time to reconstruct the file $i$ is the maximum of the time
of reconstructing all the chunks from the set $\mathcal{A}_{i}.$}, $(e)$ follows by bounding $\max_{j\in\mathcal{A}_{i}}$ by $\sum_{j\in\mathcal{A}_{i}}$
and $(f)$ follows from probabilistic scheduling. To evaluate $\Pr({\bf W}_{i,j}\ge\sigma)$,
we use Markov Lemma, i.e., 

\begin{align}
\ensuremath{\Pr({\bf W}_{i,j}\ge\sigma)\le} & \frac{\mathbb{E}[e^{t_{i,j}{\bf W}_{i,j}}]}{e^{t_{i,j}\sigma}}\nonumber \\
\overset{(g)}{=} & \frac{1}{e^{t_{i,j}\sigma}}\frac{\left(1-\rho_{j}\right)t_{i,j}\mathbb{Z}_{j}(t_{i,j})}{t_{i,j}-\Lambda_{j}\left(\mathbb{Z}_{j}(t_{i,j})-1\right)}\label{eq:MarkovLemma}
\end{align}
where $(g)$ follows from (\ref{eq:mg1_waitingTime}). Plugging (\ref{eq:MarkovLemma})
in (\ref{eq:P_Qi_sigma}), we have the following Lemma. 
\begin{theorem}
Under probabilistic scheduling, the latency tail probability for file
$i$, i.e., $\text{Pr}\left({\bf Q}_{i}\geq\sigma\right)$ is bounded
by

\begin{align}
\text{Pr}\left({\bf Q}_{i}\geq\sigma\right)\le\sum_{j}\frac{\pi_{i,j}}{e^{t_{i,j}\sigma}}\frac{(1-\rho_{j})t_{i,j}\mathbb{Z}_{j}(t_{i,j})}{t_{i,j}-\Lambda_{j}(\mathbb{Z}_{j}(t_{i,j})-1)}\label{eq:bound_tail_1}
\end{align}
for any $t_{i,j}>0$, $\rho_{j}=\Lambda_{j}\left[\frac{d}{dt}\mathbb{Z}_{j}(t_{i,j})\left|_{t_{i,j}=0}\right.\right]$, $\rho_j<1$, 
and $\Lambda_{j}(\mathbb{Z}_{j}(t_{i,j})-1)<t_{i,j}$.
\end{theorem}

We now specialize the result to the case where the service times of the servers are given in \eqref{dist_sexp} in the following corollary.

\begin{corollary}
	Under probabilistic scheduling and shifted exponential service times, the latency tail probability for file
	$i$, i.e., $\text{Pr}\left({\bf Q}_{i}\geq\sigma\right)$ is bounded
	by
	
	\begin{align}
	\text{Pr}\left({\bf Q}_{i}\geq\sigma\right)\le\sum_{j}\frac{\pi_{i,j}}{e^{t_{i,j}\sigma}}\frac{(1-\rho_{j})t_{i,j}\mathbb{Z}_{j}(t_{i,j})}{t_{i,j}-\Lambda_{j}(\mathbb{Z}_{j}(t_{i,j})-1)}\label{eq:bound_tail_1}
	\end{align}
	for any $t_{i,j}>0$,  $\rho_{j}=\Lambda_{j}\left(\frac{1}{\alpha_j}+\beta_j\right)$, $\rho_j<1$,
	$t_{i,j}(t_{i,j}-\alpha_j+\Lambda_j) +\Lambda_j \alpha_j(e^{\beta_j t_{i,j}}-1)<0$, and $\mathbb{Z}_{j}(t) = \frac{\alpha_j}{\alpha_j-t} e^{\beta_j t}$. 
\end{corollary}

\section{Characterization of Asymptotic Latency}\label{ps_asymp}
In this section, we consider homogeneous servers, and all files having same size and erasure code $(n,k)$. Further, we assume that the number of servers $m=n$. In order to understand the asymptotic delay characteristics, we also assume that $\pi_{ij}= k/n$, which chooses the $k$ servers uniformly at random. Jobs arrive over time according to a Poisson process with rate $\Lambda\supn$, and each job (file request) consists of $k\supn$ tasks with $k\supn\le n$.  Upon arrival, each job picks $k\supn$ {distinct} servers uniformly at random from the $n$ servers (uniform probabilistic scheduling) and sends one task to each server. We assume that $\Lambda\supn=n\lambda/k\supn$ for a {constant} $\lambda$, where the constant $\lambda$ is the task arrival rate to each individual queue. Since different jobs choose servers independently, the task arrival process to each queue is also a Poisson process, and the rate is $\lambda$. The service times of tasks are i.i.d.\ following a c.d.f. $G$ with expectation $1/\mu$ and a finite second moment. We think of the service time of each task as being generated upon arrival: each task brings a required service time with it, but the length of the required service time is revealed to the system only when the task is completed.  The load of each queue, $\rho=\lambda/\mu$, is then a constant and we assume that $\rho<1$.

As mentioned earlier,  each queue is an M/G/1 queue. Let $W\supn_i(t)$ denote the {workload} of server $i$'s queue at time $t$, i.e., the total remaining service time of all the tasks in the queue, including the partially served task in service.  So the workload of a queue is the waiting time of an incoming task to the queue before the server starts serving it.  Let ${\bf W}\supn(t)=\bigl(W\supn_1(t),W\supn_2(t),\dots,W\supn_n(t)\bigr)$.  Then the workload process, $({\bf W}\supn(t),t\ge 0)$, is Markovian and ergodic.  The ergodicity can be proven using the rather standard Foster-Lyapunov criteria \citep{meyn1993stability}, so we omit it here. Therefore, the workload process has a unique stationary distribution and ${\bf W}\supn(t)\tod {\bf W}\supn(\infty)$ as $t\to\infty$.

Let a random variable $T\supn$ represent this steady-state job delay.  Specifically, the distribution of $T\supn$ is determined by the workload ${\bf W}\supn(\infty)$ in the following way.  When a job comes into the system, its tasks are sent to $k\supn$ queues and experience the delays in these queues.  Since the queueing processes are symmetric over the indices of queues, without loss of generality, we can assume that the tasks are sent to the first $k\supn$ queues for the purpose of computing the distribution of $T\supn$.  The delay of a task is the sum of its waiting time and service time.  So the task delay in queue $i$, denoted by $T\supn_i$, can be written as $T\supn_i=W\supn_i(\infty)+X_i$ with $X_i$ being the service time.  Recall that the $X_i$'s are i.i.d.$\sim G$ and independent of everything else.  Since the job is completed only when all its tasks are completed,
\begin{equation}\label{eq:T}
T\supn = \max\left\{T\supn_1,T\supn_2,\dots,T\supn_{k\supn}\right\}.
\end{equation}

Let $\hat{T}\supn$ be defined as the job delay given by \emph{independent} task delays.  Specifically, $\hat{T}\supn$ can be expressed as:
\begin{equation}\label{eq:That}
\hat{T}\supn=\max\left\{\hat{T}\supn_1,\hat{T}\supn_2,\dots,\hat{T}\supn_{k\supn}\right\},
\end{equation}
where $\hat{T}\supn_1,\hat{T}\supn_2,\dots,\hat{T}\supn_{k\supn}$ are i.i.d.\ and each $\hat{T}\supn_i$ has the same distribution as $T\supn_i$. 
Again, due to symmetry, all the $T\supn_i$'s have the same distribution. Let $F$ denote the c.d.f. of $T\supn_i$, whose form is known from the queueing theory literature.  Then, we have the following explicit form for $\hat{T}\supn$:
\begin{equation}
\Pr\left(\hat{T}\supn\le \tau\right)=\left(F(\tau)\right)^{k\supn},\quad \tau\ge 0.
\end{equation}

\begin{remark}
	We note that even though the authors of	\citep{wang2019delay} related their results to Fork-Join queue, but need $n=k$, while the results naturally hold for uniform probabilistic scheduling rather than Fork-Join queues. 
\end{remark}

We first consider an asymptotic regime where the number of servers, $n$, goes to infinity, and the number of tasks in a job, $k\supn$, is allowed to increase with~$n$.  We establish the asymptotic independence of any $k\supn$ queues under the condition $k\supn = o(n^{1/4})$. This greatly generalizes the asymptotic-independence type of results in the literature where asymptotic independence is shown only for a fixed constant number of queues.  As a consequence of our independence result, the job delay converges to the maximum of independent task delays. More precisely, we show that the distance between the distribution of job delay, $T\supn$, and the distribution of the job delay given by independent task delays, $\hat{T}\supn$, goes to $0$. This result indicates that assuming independence among the delays of a job's tasks gives a good approximation of job delay when the system is large.  Again, due to symmetry, we can focus on the first $k\supn$ queues without loss of generality.

\begin{theorem}[\citep{wang2019delay}]\label{thm:asym-independence}
		Consider an $n$-server system in the uniform probabilistic scheduling with $k\supn=o(n^{1/4})$. Let $\pi^{(n,k\supn)}$ denote the joint distribution of the steady-state workloads $W\supn_1(\infty)$, $W\supn_2(\infty),\dots,W\supn_{k\supn}(\infty)$, and $\hat{\pi}^{(k\supn)}$ denote the product distribution of $k\supn$ i.i.d.\ random variables, each of which follows a distribution that is the same as the distribution of $W\supn_1(\infty)$. Then
		\begin{equation}\label{eq:asym-independence}
		\lim_{n\to\infty}d_{TV}\Bigl(\pi^{(n,k\supn)},\hat{\pi}^{(k\supn)}\Bigr)=0.
		\end{equation}
		Consequently, the steady-state job delay, $T\supn$, and the job delay given by independent task delays as defined in \eqref{eq:That}, $\hat{T}\supn$, satisfy
		\begin{equation}\label{eq:converge-cdf-general}
		\lim_{n\to\infty}\sup_{\tau\ge 0}\left|\Pr\bigl(T\supn\le \tau\bigr)-\Pr\bigl(\hat{T}\supn\le \tau\bigr)\right|=0.
		\end{equation}
\end{theorem}

For the special case where the service times are exponentially distributed, the job delay asymptotics have explicit forms presented in Corollary~\ref{cor:exp} below.
\begin{corollary}[\citep{wang2019delay}]\label{cor:exp}
	Consider an $n$-server system in the uniform probabilistic scheduling model with $k\supn=o(n^{1/4})$, job arrival rate $\Lambda\supn=n\lambda/k\supn$, and exponentially distributed service times with mean $1/\mu$.  Then the steady-state job delay, $T\supn$, converges as:
	\begin{equation}\label{eq:converge-cdf}
	\lim_{n\to\infty}\sup_{\tau\ge 0}\left|\Pr\bigl(T\supn\le \tau\bigr)-\left(1-e^{-(\mu-\lambda)\tau}\right)^{k\supn}\right|=0,
	\end{equation}
	Specifically, if $k\supn\to\infty$ as $n\to\infty$, then
	\begin{equation}\label{eq:converge-in-distr}
	\frac{T\supn}{H_{k\supn}/(\mu-\lambda)}\tod 1,\quad\text{as }n\to\infty,
	\end{equation}
	where $H_{k\supn}$ is the $k\supn$-th harmonic number, and further,
	\begin{equation}\label{eq:converge-expectation}
	\lim_{n\to\infty}\frac{{\mathbb{E}}\bigl[T\supn\bigr]}{H_{k\supn}/(\mu-\lambda)}=1.
	\end{equation}
\end{corollary}

The results above characterize job delay in the asymptotic regime where $n$ goes to infinity.  In Theorem~\ref{thm:independence-upper} below, we study the non-asymptotic regime for any $n$ and any $k\supn$ with $k\supn=k\le n$, and we establish the independence upper bound on job delay.

\begin{theorem}[\citep{wang2019delay}]\label{thm:independence-upper}
	Consider an $n$-server system in the uniform probabilistic scheduling model with $k\supn=k\le n$. Then the steady-state job delay, $T\supn$, is stochastically upper bounded by the job delay given by independent task delays as defined in \eqref{eq:That}, $\hat{T}\supn$, i.e.,
	\begin{equation}\label{eq:ind-upper}
	T\supn\le_{st} \hat{T}\supn,
	\end{equation}
	where ``$\le_{st}$'' denotes stochastic dominance. Specifically, for any $\tau\ge 0$,
	\begin{align}
	\Pr\bigl(T\supn > \tau\bigr)&\le\Pr\bigl(\hat{T}\supn > \tau\bigr) = 1-\left(F(\tau)\right)^{k\supn}.\label{eq:ind-upper-tail}
	\end{align}
\end{theorem}

We omit proofs for the results, while refer the reader to \citep{wang2019delay} for the detailed proofs in this subsection.

\section{Proof of Asymptotic Optimality for Heavy Tailed Service Rates}\label{ps_tail_asym}
 In this section, we quantify the tail index of service latency for arbitrary erasure-coded storage systems
for Pareto-distributed file size and exponential service time. First, we derive the distribution of the waiting time from a server. Next, we show that this time is a heavy-tailed with tail-index $\alpha-1$. Then, we prove that the probabilistic scheduling based algorithms achieve  optimal tail index.

\subsection{Assumptions and Chunk Size Distribution}

We assume that the arrival of client requests for each file $i$ of size $k L_i$ Mb is assumed to form an independent Poisson process with a known
rate $\lambda_i$. Further, the chunk size $\widetilde{C}_i$ Mb is assumed to have a heavy tail and follows a Pareto distribution with parameters $(x_m,\alpha)$ with shape parameter $\alpha>2$  (implying finite mean and variance). Thus, the complementary cumulative distribution function (c.c.d.f.) of the chunk size is given as
\begin{equation}
\Pr(\widetilde{C}_i >x)=\begin{cases}
(x_m/x)^\alpha \quad x\ge x_m\\
0 \quad x<x_m
\end{cases}
\end{equation}
For $\alpha>1$, the mean is $E[\widetilde{C}_i] = \alpha x_m/(\alpha-1)$. The service time per Mb at server $j$, $X_j$ is distributed as an exponential distribution the mean service time  $1/\mu_j$. Service time for a chunk of size $C$ Mb is $X_jC$.

We will focus on the tail index of the waiting time to access each file. In order to understand the tail index, let the waiting time for the files $T_W$ has $\Pr(T_W>x)$ of the order of $x^{-d}$ for large $x$, then the tail index is $d$. More formally, the tail index $d$ is defined as $\lim_{x\to\infty}\frac{-\log \Pr(T_W>x)}{\log x}$. This index gives the slope of the tail in the log-log scale of the complementary CDF.

\subsection{Waiting Time Distribution for a Chunk from a Server}

In this Section, we will characterize the Laplace Stieltjes transform of the waiting time distribution from a server, assuming that the arrival of requests at a server is Poisson distributed with mean arrival rate $\Lambda_j$.  We first note that the service time per chunk on server $j$ is given as $B_j = X_j \widetilde{C}_i$, where $\widetilde{C}_i$ is distributed as Pareto Distribution given above, and $X_j$ is exponential with parameter $\mu_j$. Using this definition, we find that
\begin{eqnarray}
&&\Pr(B_j<y) \nonumber\\&=& \Pr(X_j \widetilde{C}_i  <y)\nonumber\\
&=& \int_{x=x_m}^\infty\Pr(X_j<y/x)\alpha x_m^\alpha\frac{1}{x^{\alpha+1}}  dx\nonumber\\
&=& \int_{x=x_m}^\infty\left(1-\exp(-\mu_jy/x)\right)\alpha x_m^\alpha\frac{1}{x^{\alpha+1}}  dx\nonumber\\
&=& 1 - \int_{x=x_m}^\infty \exp(-\mu_jy/x)\alpha x_m^\alpha\frac{1}{x^{\alpha+1}}  dx
\end{eqnarray}
Substitute $t = \mu_j y/x$, and then $dt = -\mu_j y/x^2 dx$. Thus,
\begin{eqnarray}
&&\Pr(B_j>y) \nonumber\\&=&  \int_{x=x_m}^\infty \exp(-\mu_jy/x)\alpha x_m^\alpha\frac{1}{x^{\alpha+1}}  dx\nonumber\\
&=&  \int_{t=0}^{\mu_j y/x_m} \exp(-t)\alpha x_m^\alpha \frac{t^{\alpha-1}}{(\mu_j y)^\alpha} dt\nonumber\\
&=&  \alpha (x_m/\mu_j)^\alpha\frac{1}{ y^\alpha} \int_{t=0}^{\mu_j y/x_m} \exp(-t) {t^{\alpha-1}}  dt\nonumber\\
&=&  \alpha (x_m/\mu_j)^\alpha \gamma(\alpha,\mu_jy/x_m)/y^\alpha,
\end{eqnarray}
 where $\gamma$ denote lower incomplete gamma function, given as $\gamma(a,x)=\int_0^x u^{a-1}\exp(-u) du$.
 
 Since $\Pr(B_j>y) = V(y)/y^\alpha$, where $V(y) = \alpha (x_m/\mu_j)^\alpha \gamma(\alpha,\mu_jy/x_m)$ is a slowly varying function, the asymptotic of the waiting time in heavy-tailed limit can be calculated using the results in \citep{olvera2011transition} as
 
 \begin{equation}
 \Pr(W>x) \approx \frac{\Lambda}{1-\rho} \frac{x^{1-\alpha}}{\alpha-1} V(x).
 \end{equation}

Thus, we note that the waiting time from a server is heavy-tailed with tail-index $\alpha-1$. Thus, we get the following result. 
\begin{theorem}
	Assume that the arrival rate for requests is Poisson distributed, service time distribution is exponential and the chunk size distribution is Pareto with shape parameter $\alpha$. Then,  the tail index for the waiting time of chunk in the queue of a server  is $\alpha-1$.
\end{theorem}

\subsection{Probabilistic Scheduling Achieves Optimal Tail Index}

Having characterized the tail index of a single server with Poisson arrival process and Pareto distributed file size, we will now give the tail index for a general distributed storage system. The first result is that any distributed storage system has a tail index of at most $\alpha-1$.
For Poisson arrivals, Pareto chunk sizes, and exponential chunk service times, the tail index is at most $\alpha-1$.

\begin{theorem}
	The tail index for distributed storage system is at most $\alpha-1$. 
\end{theorem}
\begin{proof}
In order to show this result,  consider a genie server which is combination of  all the $n$ servers together. The service rate of this server is $\sum_{j=1}^n \mu_i$ per Mb. As a genie, we also assume that only one chunk is enough to be served. In this case, the problem reduces to the single server problem with Poisson arrival process and the result in Section VI shows that the tail index is $\alpha-1$. Since even in the genie-aided case, the tail index is $\alpha-1$, we cannot get any higher tail index. 
\end{proof}

The next result shows that the probabilistic scheduling achieves the optimal tail index.

\begin{theorem}
The optimal tail index of $\alpha-1$  is achieved by probabilistic scheduling. 
\end{theorem}
\begin{proof}
In order to show that probabilistic scheduling achieves this tail index, we consider the simple case where all the $n$-choose-$k$ sets are chosen equally likely for each file. Using this, we note that each server is accessed with equal probability of $\pi_{ij}=k/n$. Thus, the arrival rate at the server is Poisson and the tail index of the waiting time at the server is $\alpha-1$.

The overall latency  of a file chunk is the sum of the queue waiting time and the service time. Since the service time has tail index of $\alpha$, the overall latency for a chunk is $\alpha-1$.  Probability that latency is greater than $x$ is determined by the $k^{th}$ chunk to be received. The probability is upper bounded by the sum of probability over all servers that waiting time at a server is greater than $x$. This is because $\Pr(\max_j(A_j)\ge x)\le \sum_j \Pr(A_j\ge x)$ even when the random variables $A_j$ are correlated.  Finite sum of terms, each with tail index $\alpha-1$ will still give the term with tail index $\alpha-1$ thus proving that the tail index with probabilistic scheduling is $\alpha-1$. 
\end{proof}

We note that even though we assumed a total of $n$ servers, and the erasure code being the same, the above can be extended to the case when there are more than $n$ servers with uniform placement of files and each file using different erasure code. The upper bound argument does not change as long as number of servers are finite. For the achievability with probabilistic scheduling, we require that the chunks that are serviced follow a Pareto distribution with shape parameter $\alpha$. Thus, as long as placed files on each server are placed with the same distribution and the access pattern does not change the nature of distribution of accessed chunks from a server, the result holds in general.
\section{Simulations}\label{sec:ps_sims}
We define $\boldsymbol{q}=(\pi_{i,j}\forall i=1,\cdots,r\text{ and }j=1,\cdots,m)$,
and $\boldsymbol{t}=\left(\widetilde{t}_{1},\widetilde{t}_{2},\ldots,\right.$ $\left.\widetilde{t}_{r}; \overline{t}_{1},\overline{t}_{2},\ldots,\overline{t}_{r}\right)$.
Note that the values of $t_{i}$'s used for mean latency and tail
latency probability may be different and the parameters $\widetilde{t}$
and $\overline{t}$ indicate these parameters for the two cases, respectively.
Our goal is to minimize the two proposed QoE metrics over the choice
of access decisions and auxiliary bound parameters. The objective
can be modeled as a convex combination of the two QoE metrics since
this is a multi-objective optimization.

To incorporate for weighted fairness and differentiated services,
we assign a positive weight $w_{i}$ for each QoE for file $i$. Without
loss of generality, each file $i$ is weighted by the arrival rate
$\lambda_{i}$ in the objective (so larger arrival rates are weighted
higher). However, any other weights can be incorporated to accommodate
for weighted fairness or differentiated services. Let $\overline{\lambda}=\sum_{i}\lambda_{i}$
be the total arrival rate. Hence, $w_{i}=\lambda_{i}/\overline{\lambda}$
is the ratio of file $i$ requests. The first objective is the minimization
of the mean latency, averaged over all the file requests, and is given
as $\sum_{i}\frac{\lambda_{i}}{\overline{\lambda}}\,\mathbb{Q}_{i}$.
The second objective is the minimization of latency tail probability,
averaged over all the file requests, and is given as $\sum_{i}\frac{\lambda_{i}}{\overline{\lambda}}\,\text{Pr}\left(\mathbb{Q}_{i}\geq\sigma\right)$. 

By using a special case of the expressions for the mean latency and the latency tail
probability in Sections  \ref{mean_ps} and \ref{tail_ps}, optimization of a convex combination of the two QoE
metrics can be formulated as follows.

\begin{align}
\text{\textbf{min}\,\,\,\,\,} & \sum_{i=1}^{r}\frac{\lambda_{i}}{\overline{\lambda}}\left[\theta\,\frac{1}{\widetilde{t}_{i}}\text{log}\left(\sum_{j=1}^{m}q_{i,j}\frac{(1-\rho_{j})\widetilde{t}_{i}\mathbb{Z}_{j}(\widetilde{t}_{i})}{\widetilde{t}_{i}-\Lambda_{j}\left(\mathbb{Z}_{j}(\widetilde{t}_{i})-1\right)}\right)\right.\nonumber \\
& \left.+(1-\theta)\sum_{j=1}^{m}\frac{q_{i,j}}{e^{\overline{t}_{i}\sigma}}\frac{(1-\rho_{j})\overline{t}_{i}\mathbb{Z}_{j}(\overline{t}_{i})}{\overline{t}_{i}-\Lambda_{j}(\mathbb{Z}_{j}(\overline{t}_{i})-1)}\right]\label{eq:optFun-2}\\
\boldsymbol{\text{s.t.}}\nonumber 
\end{align}

\begin{align}
& \mathbb{Z}_{j}(t_{i})=\frac{\alpha_{j}}{\alpha_{j}-t_{i}}e^{\beta_{j}t_{i}}\,\,,\,\forall j\label{eq:mgf_obj}\\
& \rho_{j}=\frac{\Lambda_{j}}{\alpha_{j}}+\Lambda_{j}\beta_{j}<1\,\,,\,\forall j\label{eq:rho_obj}\\
& \Lambda_{j}=\sum_{i}\lambda_{i}q_{i,j}\,\,,\,\forall j\label{eq:rate_obj}\\
& \sum_{j}q_{i,j}=k_{i}\,\,,\,\forall i\label{eq:qij_obj}\\
& q_{i,j}=0,\,\,j\notin\mathcal{G}_{i}\,\,,\forall i,j\label{eq:qij_obj_2}\\
& q_{i,j}\in[0,1]\,\,,\,\forall i,j\label{eq:qij_obj3}\\
& \widetilde{t}_{i}>0\,\,,\,\forall i\label{eq:t1}\\
& \overline{t}_{i}>0\,\,,\,\forall i\label{eq:t2}\\
& \widetilde{t}_{i}(\widetilde{t}_{i}-\alpha_{j}+\Lambda_{j})+\Lambda_{j}\alpha_{j}(e^{\beta_{j}\widetilde{t}_{i}}-1)<0\label{eq:t3}\\
& \overline{t}_{i}(\overline{t}_{i}-\alpha_{j}+\Lambda_{j})+\Lambda_{j}\alpha_{j}(e^{\beta_{j}\overline{t}_{i}}-1)<0\label{eq:t4}\\
\nonumber \\
\boldsymbol{\text{var}}\text{\,\,\,\,\,} & \boldsymbol{q}\,,\boldsymbol{t},\nonumber 
\end{align}
where $\theta\in[0,1]$ is a trade-off factor that determines the
relative significance of mean latency and latency tail probability
in the objective function. By changing $\theta$ from $\theta=1$
to $\theta=0$, the solution for (\ref{eq:optFun-2}) spans the solutions
that minimize the mean latency to ones that minimize the tail latency
probability. Note that constraint (\ref{eq:rho_obj}) gives the load
intensity of server $j$. Constraint (\ref{eq:rate_obj}) gives the
aggregate arrival rate $\Lambda_{j}$ for each node for the given
probabilistic scheduling probabilities $q_{i,j}$ and arrival rates
$\lambda_{i}$. Constraints (\ref{eq:qij_obj})-(\ref{eq:qij_obj3})
guarantee that the scheduling probabilities are feasible. Also, Constraints
(\ref{eq:t1})-(\ref{eq:t4}) ensure that the moment generating function
given in (\ref{eq:mg1_waitingTime}) exists. Note that the optimization
over $\boldsymbol{q}$ helps decrease the overall latency which gives
significant flexibility over choosing the lowest-queue servers for
accessing the files. We further note that the optimization problem
in (\ref{eq:optFun-2}) is non-convex as, for instance, Constraint
(\ref{eq:t3}) is non-convex in $(\boldsymbol{q},\,\boldsymbol{t})$
jointly. In order to solve the problem, we can use an alternating optimization that divides the problem into two subproblems 
that optimize one variable while fixing the another. In order to solve each subproblem, we use the iNner cOnVex Approximation
(NOVA) algorithm proposed in \citep{scutNOVA}, which guarantees convergence to a stationary point. Based on this, it can be shown that the alternating optimization converges to a stationary point. 

To validate our proposed algorithm for joint mean-tail latency and
evaluate its performance, we simulate our algorithm in a distributed
storage system of $m=12$ distributed nodes, $r=1000$ files, all
of size $200$ MB and using $(7,4)$. However, our model can be used
for any given number of storage servers, any number of files, and
for any erasure coding setting. We consider a shifted-exponential
distribution for the chunk service times as it has been shown in real
system measurements on Tahoe and Amazon S3 servers \citep{Makowski:89,CS14,AmazonS3}.
The service time parameters $\alpha_{j}$ and $\beta_{j}$ are shown
in Table \ref{tab:Storage-Nodes-Parameters1}. Unless otherwise explicitly
stated, the arrival rate for the first $500$ files is $0.002s^{-1}$
while for the next $500$ files is set to be $0.003s^{-1}$.
\begin{table}[b]	
	\vspace{-1em}
	\vspace{-0.5em}	
	{\footnotesize{}}%
	\begin{tabular}{|c|c|c|c|c|c|c|}
		\multicolumn{1}{c}{} & \multicolumn{1}{c}{\textbf{\footnotesize{}Node 1}} & \multicolumn{1}{c}{\textbf{\footnotesize{}Node 2}} & \multicolumn{1}{c}{\textbf{\footnotesize{}Node 3}} & \multicolumn{1}{c}{\textbf{\footnotesize{}Node 4}} & \multicolumn{1}{c}{\textbf{\footnotesize{}Node 5}} & \multicolumn{1}{c}{\textbf{\footnotesize{}Node 6}}\tabularnewline
		\hline 
		{\footnotesize{}$\alpha_{j}$ } & {\footnotesize{}$18.23$ } & {\footnotesize{}$24.06$ } & {\footnotesize{}$11.88$ } & {\footnotesize{}$17.06$ } & {\footnotesize{}$20.19$ } & {\footnotesize{}$23.91$}\tabularnewline
		\hline 
	\end{tabular}{\footnotesize \par}
	
	{\footnotesize{}}%
	\begin{tabular}{|c|c|c|c|c|c|c|}
		\multicolumn{1}{c}{} & \multicolumn{1}{c}{\textbf{\footnotesize{}Node 7}} & \multicolumn{1}{c}{\textbf{\footnotesize{}Node 8}} & \multicolumn{1}{c}{\textbf{\footnotesize{}Node 9}} & \multicolumn{1}{c}{\textbf{\footnotesize{}Node 10}} & \multicolumn{1}{c}{\textbf{\footnotesize{}Node 11}} & \multicolumn{1}{c}{\textbf{\footnotesize{}Node 12}}\tabularnewline
		\hline 
		{\footnotesize{}$\alpha_{j}$ } & {\footnotesize{}$27.01$ } & {\footnotesize{}$21.39$ } & {\footnotesize{}$9.92$ } & {\footnotesize{}$24.96$ } & {\footnotesize{}$26.53$ } & {\footnotesize{}$21.80$}\tabularnewline
		\hline 
	\end{tabular}{\footnotesize \par}
{\caption{Storage Node Parameters Used in our Simulation (Shift $\beta_{j}=10\,msec,\,\forall j$
		and rate $\alpha$ in 1/s).  }\label{tab:Storage-Nodes-Parameters1}
}

\end{table}

In order to initialize our algorithm, we use a random placement of
each file on $7$ out of the $12$ servers. Further, we set $q_{i,j}=k/n$
on the placed servers with $t_{i}=0.01$ $\forall i$ and $j\in\mathcal{G}_{i}$.
However, these choices of $q_{i,j}$ and $t_{i}$ may not be feasible.
Thus, we modify the initialization to be closest norm feasible solution.

We compare the proposed approach with two baselines. 
\begin{enumerate}
\item PSP (\textit{Projected Service-Rate Proportional }Access) Policy:
The access probabilities $\boldsymbol{q}$ are assigned proportional
to the service rates of the storage nodes, i.e., $q_{i,j}=k_{i}\frac{\mu_{j}}{\sum_{j}\mu_{j}}$,
where $\mu_{j}=1\left/(\frac{1}{\alpha_{j}}+\beta_{j})\right.$. This
policy assigns servers proportional to their service rates. These
access probabilities are projected toward feasible region in (\ref{eq:optFun-2})
to ensure stability of the storage system. With these fixed access
probabilities, the QoE metrics are optimized over the auxiliary variables
$\mathbf{t}$ using NOVA. 
\item PEA (Projected Equal Access) Policy: In this strategy, we set $q_{i,j}=k/n$
on the placed servers with $t_{i}=0.01$ $\forall i$ and $j\in\mathcal{G}_{i}$.
We then modify the initialization of $\boldsymbol{q}$ to be closest
norm feasible solution given above values of $\boldsymbol{t}$. Finally,
an optimization over $\boldsymbol{t}$ is performed to the objective
using NOVA.
\end{enumerate}

\begin{figure}[t]
\begin{center}	\includegraphics[trim=0in 0in 5in 0in, clip,scale=0.4]{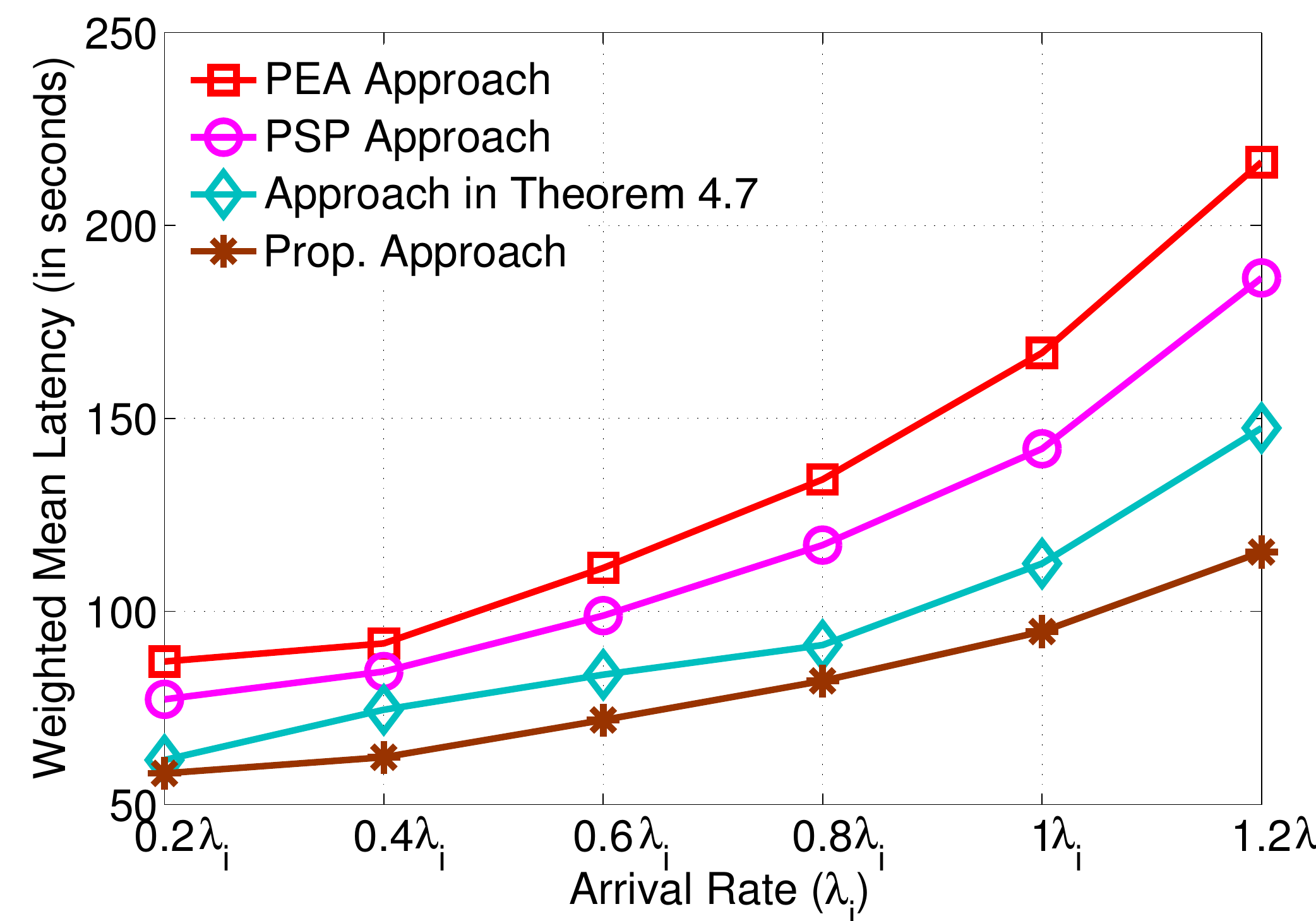}
	\end{center}
	\caption{Weighted mean latency for different file arrival rates. We vary the
		arrival rate of file $i$ from $0.2\times\lambda_{i}$ to 1.$2\times\lambda_{i}$,
		where $\lambda_{i}$ is the base arrival rate.\label{lat_vs_arrRate}}
	
\end{figure}
{\bf Mean Latency: }  We first let $\theta=1$. We also compare with the third policy, that is based on optimizing the mean latency upper bound in Theorem \ref{old_upper_bound}. Figure \ref{lat_vs_arrRate} plots the effect of different arrival
rates on the upper bound of the mean latency where we compare our proposed algorithm
with three different policies. Here, the arrival rate of each file
$\lambda_{i}$ is varied from $0.2\times\lambda_{i}$ to $1.2\times\lambda_{i}$,
where $\lambda_{i}$ is the base arrival rate. We note that our proposed
algorithm outperforms all these strategies for the QoE metric of mean
latency. Thus, both access and file-based auxiliary variables of files
are both important for the reduction of mean latency. We also note
that uniformly accessing servers (PEA) and simple service-rate-based
scheduling (PSP) are unable to optimize the request based on different
factors like arrival rates, different latency weights, thus leading
to much higher latency. As expected, the mean latency increases with
arrival rates. However, at high arrival rates, we see significant
reduction in mean latency for our proposed approach. For example,
we see, at the highest arrival rate, approximately $25\%$ reduction
in weighted mean latency as compared to the proposed approach in \citep{Yu_TON,Xiang:2014:Sigmetrics:2014}, given in Theorem \ref{old_upper_bound}.

\begin{figure}[t]
	\centering\includegraphics[trim=0in 0in 4in 0in,clip,scale=0.4]{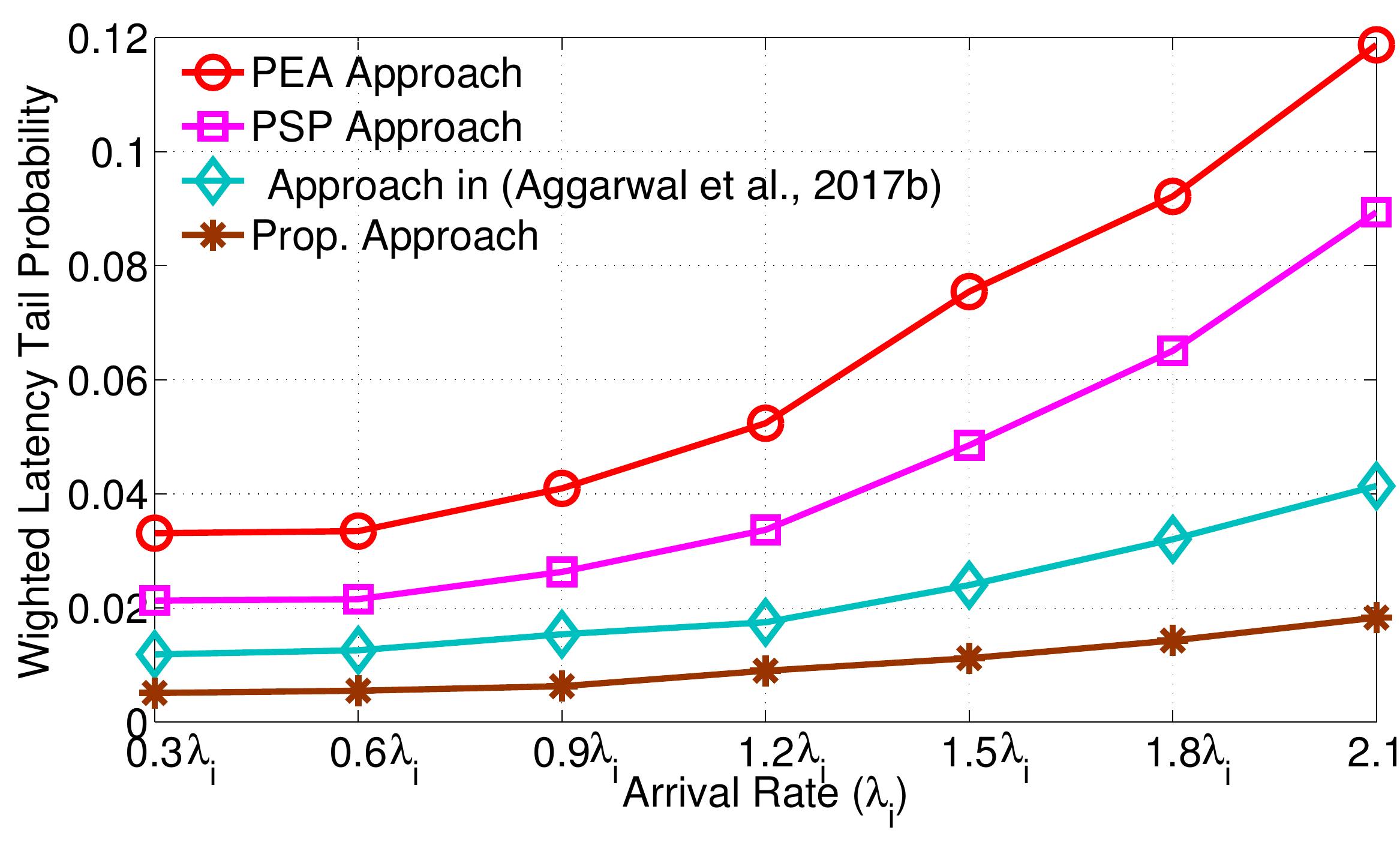}
	
	\caption{Weighted latency tail probability for different file arrival rates.
		We vary the arrival rate of file $i$ from $0.3\times\lambda_{i}$
		to $2.1\times\lambda_{i}$, where $\lambda_{i}$ is the base arrival
		rate.\label{WLTP_vs_arrRate}}
	
\end{figure}

{\bf Tail Latency: } For $\lambda_{i}$'s as the base arrival rates and $\sigma=80$ seconds,
we increase the arrival rate of all files from $0.3\lambda$ to $2.1\lambda$
and plot the weighted latency tail probability in Figure \ref{WLTP_vs_arrRate}.
We note that our algorithm assigns differentiated latency for different
files to keep low weighted latency tail probability. We also observe
that our proposed algorithm outperforms all strategies for all arrival
rates. For example, at the highest arrival rate, the proposed approach
performs much better compared to \citep{Jingxian,al2019ttloc}, i.e.,
a significant reduction in tail probability from $0.04$ to $0.02$.
Hence, reducing the latency of the high arrival rate files and exploiting
the role of auxiliary variables result in reducing the overall weighted
latency tail probability.

\begin{figure}[t]
	\centering\includegraphics[trim=0in 0in 4in 0in,clip,scale=0.4]{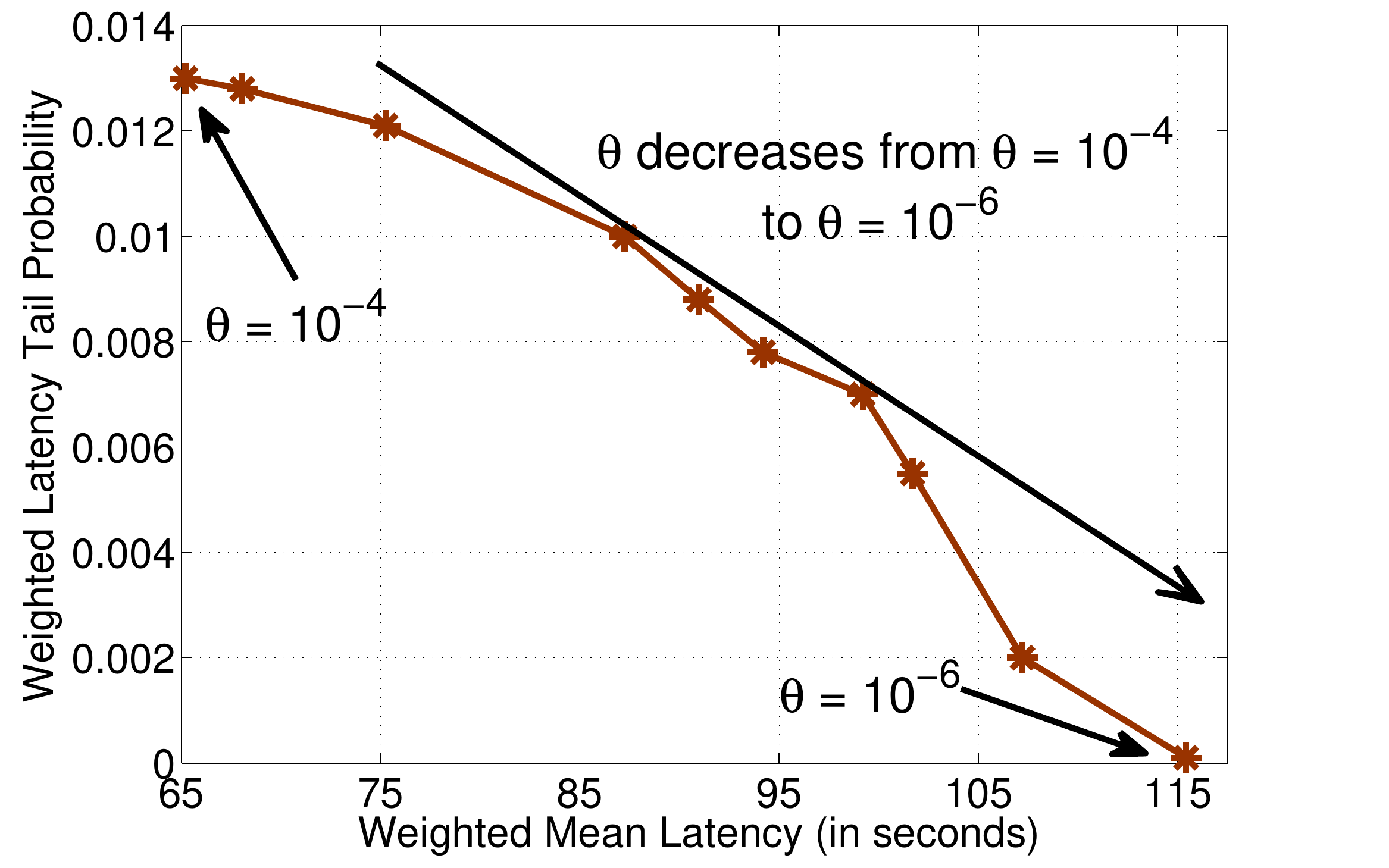}
	
	\caption{Tradeoff between weighted mean latency and weighted latency tail probability
		obtained by varying $\theta$ in the objective function given by (\ref{eq:optFun-2}).
		We vary $\theta$ (coefficient of weighted mean latency) from $\theta=10^{-4}$
		to $\theta=10^{-6}$. These values are chosen carefully to bring the
		two QoE metrics to a comparable scale, since weighted mean latency
		is orders of magnitude higher than weighted latency tail probability.\label{tradeoff}}
	
\end{figure}

{\bf Tradeoff: } We investigate the tradeoff between weighted mean latency and weighted
latency tail probability in Figure \ref{tradeoff}. Intuitively, if
the mean latency decreases, the latency tail probability also reduces.
Thus, one may wonder whether the optimal point for decreasing the
mean latency and the latency tail probability is the same? From Figure
\ref{tradeoff}, we answer this question negatively since for $r=1000$
and $m=12$, we find out that the optimal mean latency is approximately
$43\%$ lower as compared to the mean latency at the value of ($\boldsymbol{q}$,$\boldsymbol{t}$)
that optimizes the weighted latency tail probability. Hence, an efficient
tradeoff point between the two QoE metrics can be chosen based on
the point on the curve that is appropriate for the clients.

\section{Notes and Open Problems}\label{sec:ps_notes}

Probablistic scheduling for erasure-coded storage system was first proposed in \citep{Yu_TON,Xiang:2014:Sigmetrics:2014}, where the mean latency was characterized. The theoretical analysis on  joint latency-plus-cost optimization is evaluated in Tahoe \citep{Tahoe}, which is an open-source, distributed
file system based on the zfec erasure coding library for fault tolerance. The mean latency expressions are further extended in \citep{8406958}. Differentiated latency in erasure-coded storage by investigating weighted queue and priority queue policies was considered in \citep{Yu-ICDCS,Yu-TNSM16}. The problem of erasure-coded storage in a  data center network needs to account for the  limited bandwidth available at both top-of-the-rack and aggregation switches, and differentiated service requirements of the tenants. This is accounted via efficient splitting of network bandwidth among different intra- and inter-rack data flows for different service classes in line with their traffic statistics \citep{Yu-CCGRID,Yu-TCC16}. Erasure coding can lead to new caching designs, for which the latency has been characterized \citep{Sprout,Yu-TON16}. The proposed approach is prototyped using
Ceph, an open-source erasure-coded storage system \citep{Ceph} and
tested on a real-world storage testbed with an emulation of
real storage workload, as will de detailed in Chapter \ref{chpt:impl}. %

The evaluation of tail latency in erasure-coded storage systems using probalistic scheduling was first considered in \citep{Jingxian}. This was further extended in \citep{al2019ttloc}, where the probabilistic scheduling-based algorithms were shown to be (asymptotically) optimal since they are able to achieve the exact tail index. The analysis in this chapter further extends these works, and is shown to outperform these previous analysis in this chapter. These extended results appear for the first time in this monograph.  The authors of \citep{wang2019delay} considered asymptotic regime for $n$ in the case of uniform probabilistic scheduling. 

The results for mean and tail latency in this Chapter have been extended from the works above, with an aim of giving a concise representation for general service process. The approach in this chapter has also been used in \citep{Abubakr_meantail}, where TTL based caching is also considered, and \citep{al2018video}, where the results are extended to stall duration (and will be covered in detail in Chapter \ref{sec:video}).

The approach could be further extended on the following directions:
\begin{enumerate}
	\item {\bf Placement of multiple chunks on the same node}: This case arises when a group of storage nodes share a single bottleneck (e.g., outgoing bandwidth at a regional datacenter) and must be modeled by a single queue, or in small clusters the number storage node is less than that of created file chunks (i.e., $n_i>m$). As a result, multiple chunk requests corresponding to the same file request can be submitted to the same queue, which processes the requests sequentially and results in dependent chunk service times.  The analysis in this chapter can be extended on the lines of \citep{Yu_TON}. 
	\item {\bf File can be retrieved from more than $k$ nodes}: We first note that file can be retrieved by obtaining $F_i/d_i$ amount of data from $d_i\ge k_i$ nodes with the same placement and the same $(n_i,k_i)$ MDS code. To see this, consider that the content at each node is subdivided into $B=\binom{d_i}{k_i}$ sub-chunks (We assume that each chunk can be perfectly divided and ignore the effect of non-perfect division). Let ${\cal L} = \{{\cal L}_1, \cdots, {\cal L}_{ B}\}$ be the list of all ${B}$ combinations of $d_i$ servers  such that each combination is of size $k_i$. In order to access the data, we get $m^{\text{th}}$ sub-chunks from all the servers in ${\cal L}_m$ for all $m = 1, 2, \cdots {B}$. Thus, the total size of data retrieved is of size $F_i$, which is evenly accessed from all the $d_i$ nodes. In order to obtain the data, we have enough data to decode since $k_i$ sub-chunks are available for each $m$ and we assume a linear MDS code. %
	
In this case, smaller amount of data can be
	obtained from more nodes. Obtaining data from more nodes
	has an effect of considering worst ordered statistics having
	an effect on increasing latency, while the smaller file size
	from each of the node facilitating higher parallelization, and thus
	decreasing latency. The optimal value of the number of disks
	to access can then be optimized. However, the analysis of the mean and tail latency can be easily extended following the approach in \citep{Yu_TON}.
	
	\item {\bf Asymptotic Independence Results for Heterogenous Files and Servers: } We note that the result in Section \ref{ps_asymp} states that the steady state job delay is upper bounded by the delay given by independent task delays. Further, the steady-state job delay holds was characterized for $n$ large. However, these results hold for all files of same size and same erasure code. Further, the servers are assumed to be homogenous. Finally, the probabilistic scheduling probabilities are assumed to be equal. Extension of these results when such assumptions do not hold are open. 
	
	\item {\bf  Efficient Caching Mechanisms: } Efficient caching can help reduce both mean and tail latency. Different caching mechanisms based on Least-Resenly-Used strategy and its adaptations have been studied \citep{Abubakr_meantail,friedlander2019generalization}. Erasure-coded mechanisms of caching have also been explored \citep{Yu-TON16}. In an independent line of work, coded caching strategies have been proposed which use a single central server \citep{pedarsani2015online}, with extensions to distributed storage \citep{luo2019coded}. Integrating efficient caching mechanisms with distributed storage and evaluating them in terms of latency is an interesting problem. 
\end{enumerate}

\chapter{Delayed-Relaunch Scheduling Approach}
\label{chp:relauch}

In this Chapter, we introduce the model of Delayed Relaunch Scheduling in Section~\ref{drs_defn}, which is first proposed for distributed storage in this monograph. This model generalizes the models of Fork-Join Scheduling in Chapter \ref{sec:fj_sch} and the Probabilistic Scheduling in Chapter \ref{sec:prob_sch}, and thus the guarantees in those chapters hold for the relevant parameters. Such a model is generalized from the earlier works on speculative execution for cloud computing \citep{aktacs2019straggler,Straggler2020Inf}. Even though queueing analysis is not available in general for this strategy, the analysis is provided for a single job. The inter-service time of different chunks is provided in Section \ref{drs_inter}, which is used to characetrize two metrics - Mean Service Completion Time and Mean Server Utilization Cost in Section \ref{dsr_metrics}, for shifted exponential service times of the homogenous servers. Sections~\ref{sec:dsr_sims} and \ref{sec:dsr_notes} contain simulation results and notes on future directions, respectively.

\section{Delayed-Relaunch Scheduling}\label{drs_defn}

We consider the system model introduced in Chapter~\ref{model}. In Fork-Join scheduling, the request was sent to all $n_i$ servers and the job completed when $k_i$ servers finish execution. In Probabilistic scheduling, the request was sent to $k_i$ servers using a probabilistic approach. The delayed relaunch scheduling sends the requests to the servers in stages. In stage $d$, the request is sent to $n_{i,d}$ servers. The job is complete when $k_i$ servers have finished execution. The time between the stages can either be deterministic, random variable independent of server completion times, or a random variable based on the different task completion times. 

Since in Fork-Join scheduling, all $n_i$ servers may be busy processing the file wasting time at all of $n_i-k_i$ servers that will eventually not be used, the delayed scheduling aims to reduce the additional time spent at $n_i-k_i$ servers by launching them at a later time.

In order to see the concept of delayed relaunch scheduling, see Figure \ref{Fig:timeStampsOnRealLine} (where file index $i$ is supressed). For a file request, $n_{i,0}$ tasks are requested at time $t_0=0$, $n_{i,1}$ tasks are requested at time $t_1$, and so on. Based on these requests, the overall job is complete when $k_i$ servers have finished execution.

\begin{figure}[hhh]
	\centering
	\scalebox{0.8}{\begin{tikzpicture}
font=\small,
\draw (0,0) -- (8,0);
\foreach \i in {0,1,...,8} %
\draw (\i,0.1) -- + (0,-0.2) node[below] {$\i$};
\fill[blue] (0,0) circle (0.6 mm) node[above] {$t_0$};
\fill[blue] (2,0) circle (0.6 mm) node[above] {$t_1$};
\fill[blue] (4,0) circle (0.6 mm) node[above] {$t_2$};

\draw [->] (0,-1) -- (7.3,-1) ;
\draw [->] (0,-1.2) -- (6.75,-1.2) ;
\draw [->] (0,-1.4) -- (6.4,-1.4) ;
\draw [->] (0,-1.6) -- (5.4,-1.6) ;

\draw [->] (2,-2) -- (6.6,-2) ;
\draw [->] (2,-2.2) -- (6.25,-2.2) ;
\draw [->] (2,-2.4) -- (7,-2.4) ;
\draw [->] (2,-2.6) -- (6,-2.6) ;
\draw [->] (2,-2.8) -- (5.25,-2.8) ;

\draw [->] (4,-3.2) -- (5,-3.2) ;
\draw [->] (4,-3.6) -- (7.5,-3.6) ;
\draw [->] (4,-3.4) -- (6.5,-3.4) ;

\draw [blue, dotted] (0,0) -- (0,-5);
\draw [blue, dotted] (2,0) -- (2,-5);
\draw [blue, dotted] (4,0) -- (4,-5);
\draw [green, densely dashed] (5,0) -- (5,-5);

\draw [red, dashed, ->] (5,-1) -- (7.3,-1);
\draw [red, dashed, ->] (5,-1.2) -- (6.75,-1.2);
\draw [red, dashed, ->] (5,-1.4) -- (6.4,-1.4);
\draw [red, dashed, ->] (5,-1.6) -- (5.4,-1.6);

\draw [red, dashed, ->] (5,-2) -- (6.6,-2);
\draw [red, dashed, ->] (5,-2.2) -- (6.25,-2.2);
\draw [red, dashed, ->] (5,-2.4) -- (7,-2.4);
\draw [red, dashed, ->] (5,-2.6) -- (6,-2.6);
\draw [red, dashed, ->] (5,-2.8) -- (5.25,-2.8);

\draw [red, dashed, ->] (5,-3.4) -- (6.5,-3.4);
\draw [red, dashed, ->] (5,-3.6) -- (7.5,-3.6);

\draw[|<->|, loosely dashed] (0,-4.2) -- (5,-4.2) node [above,text width=3cm,align=center,midway] {\textbf{$s_1$}};

\draw [violet, decorate,decoration={brace,amplitude=1pt, mirror},xshift=-2pt,yshift=0pt]
(0,-1) -- (0,-1.6) node [black,midway,xshift=-0.4cm] {\footnotesize $n_0$};
\draw [violet, decorate,decoration={brace,amplitude=1pt, mirror},xshift=-2pt,yshift=0pt]
(2,-2) -- (2,-2.8) node [black,midway,xshift=-0.4cm] {\footnotesize $n_1$};
\draw [violet, decorate,decoration={brace,amplitude=1pt, mirror},xshift=-2pt,yshift=0pt]
(4,-3.2) -- (4,-3.6) node [black,midway,xshift=-0.4cm] {\footnotesize $n_2$};

\end{tikzpicture}}
	\caption{
		This figure illustrates  two-forking, 
		by plotting the different completion times on the real line, 
		with  the forked servers $n_0=4$, $n_1=5$, $n_2=3$ at forking points $t_0=0$, $t_1=2$, $t_2=4$. 
		The first task completes at $s_1$. 
	}
	\label{Fig:timeStampsOnRealLine}
\end{figure}
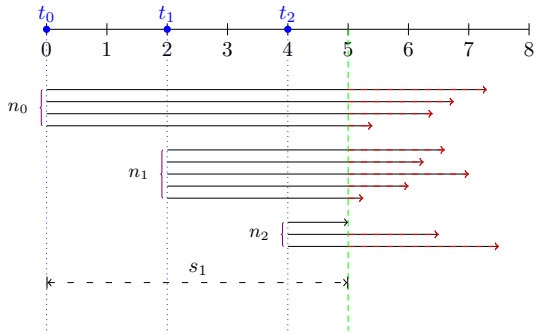

In this Chapter, we assume $n_i=n$  and $k_i=k$ for all $i$, and thus index $i$ will be supressed. We assume homogeneous servers with shifted exponential service times, ${\rm Sexp}(c,\mu)$, where ${\rm Sexp}(c,\mu)$ is as defined in \eqref{eq:Sexp}. Analysis of delayed relaunch scheduling is not as straightforward due to the added complexity in choosing $t_{d}$. For Fork-Join scheduling, $n_0=n$ and others are zero, while in the case of probabilistic scheduling $n_0=k$. In addition, the choice of servers has similar challenges as in Probabilistic Scheduling. The cancellation of remaining tasks after $k$ finished execution is akin to the Fork-Join scheduling. Another challenge in the scheduling is as to where in the queue should the tasks requested at time $t_1$ be placed - at the tail or based on the job request times or some other approach? These challenges make the problem hard to analyze. We note that since the approach is a generalization of Fork-Join scheduling and Probabilistic Scheduling, the latency optimized over the different parameters of Delayed Relaunch scheduling is lower than that for Fork-Join scheduling and Probabilistic Scheduling. Thus, the tail index optimality holds also for delayed relaunch scheduling.

In order to make progress, we do not consider the scheduling approach in this chapter while focus on a single job. The analysis for a stream of jobs and general service time distributions is left as a future work.  We assume a single-fork scheduling, 
where a file request starts at $n_0$ parallel servers at time $t_0 = 0$, 
and adds $n_1 = n - n_0$ servers at a random time instant $t_1$ corresponding to service completion time of the $\ell_0$th {coded sub-task} out of $n_0$ initial servers. 
The total service completion time is given by $t_2$ when the remaining {coded sub-tasks} at $\ell_1 = k-\ell_0$ servers are completed. 
Since we can't have more service completions than the number of servers in service, 
we have $\ell_0 \le n_0$ and $\ell_0 + \ell_1 = k \le n$. The overall scheduling approach choses the first $n_0$ servers using probabilistic scheduling among the $n$ servers having the corresponding file. Further, the choice of $n_1$ servers is from the rest of $n-n_0$ servers  that were not chosen for the first phase. The request is cancelled from $n_0+n_1-k$ servers that still have the file chunk in queue/service when $k$ chunks have been received.

We note that the latency analysis of the delayed relaunch scheduling is open. In this chapter, we will analyse the delayed relaunch scheduling for a single job. In this case, there is no queue and earlier jobs in the system. We will consider two metrics - the service completion time and the service utilization cost. The first indicate the latency in extremely lightly loaded system while the second indicate the amount of load created by the job for other jobs in the system.

The service completion time for $k$ requested chunks is denoted by $t_{2}$ 
and the server utilization cost by $W$. 
We denote the service completion time of $r$th {coded sub-task} in $i$th stage $[t_i , t_{i+1})$ by $t_{i,r}$ where $i \in \{0,1\}$. 
Since each stage consists of $\ell_i$ service completions, we have $r \in \{0, \dots, \ell_i\}$ such that $t_{i,0} = t_i$ and $t_{i, \ell_{i}} = t_{i+1,0} = t_{i+1}$. 

Assuming that a server is discarded after its chunk completion, 
we can write the utilization cost in this case as the time-integral of number of servers that are ON during the service completion $[0, t_2]$, multiplied by the server utilization cost per unit time
\EQN{ 
	\label{eqn:ServerUtilizationCost}
	W = \lambda \sum_{i=0}^{1} \Big[\sum_{r = 0}^{\ell_{i}-1}(t_{i, r+1} - t_{i, r}) \Big(\sum_{j=0}^{i}(n_j - \ell_j)+ \ell_i - r \Big)\Big].
}
The total service completion time $S = t_2$ can be written as the following telescopic sum
\EQN{
	\label{eqn:ServiceCompletionTime}
	S = \sum_{i=0}^{1} \left[\sum_{r = 1}^{\ell_{i}}(t_{i, r} - t_{i, r-1}) \right].
}

Thus, both the metrics rely on inter-service times $t_{i, r} - t_{i, r-1}$, which will be characterized in the next section, followed by the  result on the two metrics.

\section{Characterization of  Inter-Service Times of Different Chunks for Single Job}\label{drs_inter}

Before we start the analysis, we describe some preliminary results and definitions. Let $T_i$, $i\in \{1, \cdots, n\}$,  be a shifted exponential with rate $\mu$ and shift $c$, 
such that the complementary distribution function $\bar{F} = 1 - F$ can be written 
\EQN{
	\label{eqn:ShiftedExp}
	\bar{F}(x) \triangleq P\{T_i > x\} = {\bf 1}_{x \in [0, c]} + e^{-\mu(x-c)}{\bf 1}_{ x \ge c}.
}
We see that $T'_i \triangleq T_i-c$ are \emph{i.i.d.} random variables distributed exponentially with rate $\mu$. 
We denote the $j$th order statistic of $(T'_1, \dots, T'_n)$ by $X^n_j$. The $j$th order statistic of  $(T_1, \dots, T_n)$ is $c + X^n_j$.  The distribution of ordered statistics is given as follows. 
\begin{lemma}
	\label{rem:DistOS}
	Let $(X_1, \dots, X_n)$ be $n$ \emph{i.i.d.} random variables with common distribution function $F$, 
	and we denote the  $j$th order statistics of this collection by $X_j^n$. 
	Then the distribution of $X_j^n$ is given by 
	$
	P\set{X_j^n \le x} = \sum_{i=j}^n\binom{n}{i}F(x)^i\bar{F}(x)^{n-i}. 
	$
\end{lemma}

The next result provides the mean gap between two inter-arrivals. 
\begin{lemma}
	\label{rem:DiffOS}
	Denoting $X_0^n = 0$, 
	from the memoryless property of $T'_i$, we observe the following equality in joint distribution of two vectors
\begin{equation}
		(X_{j}^n-X_{j-1}^n:  j \in [n]) = \left(\frac{T'_{j}}{n-j+1}: j \in [n]\right). 
	\end{equation}
\end{lemma}

We next introduce a definition of the Pochhammer function, which will be used in the analysis. 
\begin{definition}
	\label{def:hypergeometricSeries}
	We denote the Pochhammer function $(a)_n \triangleq  \frac{\Gamma(a+n)}{\Gamma(a)}$ 
	to define the $z$-transform of hypergeometric series as 
	\EQN{
		\label{eqn:HGF}
		_p F_q(z) \triangleq {_p F_q\Big[{{a_1, \dots, a_p} \atop {b_1,\dots, b_q}};z \Big]} = \sum_{n=0}^{\infty} \frac{\prod_{i=1}^{p} (a_i)_n z^n}{\prod_{j=1}^{q} (b_j)_n (n)!}.
	}
	Because generalizations of the above series also exist~\citep{gasper2004basic}, 
	this series is referred to here as the hypergeometric series rather than as the generalized hypergeometric series.
\end{definition}
\begin{remark}
	\label{rem:integral1}
	For positive integers $p,q$ and positive reals $c,\mu$, we have the following identity in terms of the hypergeometric series $_p F_q$ defined in Definition~\ref{def:hypergeometricSeries}
	\begin{eqnarray}
	& \int_{0}^{c} x e^{-\mu x} \frac{(1-e^{-\mu x})^q (e^{-\mu x}-e^{-\mu c})^{p-q}}{(1-e^{-\mu c})^{p+2}} dx \nonumber \\
	&= \left(\frac{1}{(p+2)\mu^2\binom{p+1}{q+1}}\right) {_3 F_2\Big[{{1, 1, q+2} \atop {2, p+3}};1-e^{-\mu c} \Big]}.
	\end{eqnarray}
\end{remark}
Using the definition of hypergeometric series, it can be verified that expression in Remark~\ref{rem:integral1} simplifies to the  expression below for $p = q = m-1$.
\begin{corollary}
		\label{rem:integral2}
		For positive integer $m$ and positive reals $c,\mu$, 
		we have the following identity 
		$m \mu\int_{0}^{c}  x  e^{-\mu x} (1-e^{-\mu x})^{m-1} dx \\
		= c (1-e^{-\mu c})^{m}   - c + \sum_{i=1}^{m} \frac{(1-e^{-\mu c})^i}{i\mu}.$

	\end{corollary}
%
Having provided some definitions and the basic results, we analyze the time between two service completions of the chunks. 
Recall that we have two contiguous stages. 
The time interval $[t_0, t_1)$ corresponds to the stage~$0$, and 
the interval $[t_1, t_2]$ corresponds to the stage~$1$. In stage~$0$, we switch on $n_0$ initial servers at instant $t_0 = 0$. 
This stage is completed at the single forking point denoted by the instant $t_1$, 
when $\ell_0$ chunks out of $n_0$ are completed. 
At the beginning of stage~$1$, additional $n_1 = n-n_0$ servers are switched on, each working on a unique chunk. 
The job is completed at the end of this second stage, when remaining $k-\ell_0$ chunks are completed. 
The $k$th service completion time is denoted by $t_2$. We will separately analyze these two stages in the following.

%
%
%
%
%

%
We will first compute the mean of the interval $[t_{0,r-1}, t_{0,r})$ for each $r \in [\ell_0]$. 

\begin{lemma}
\label{lem:InterArrivalStage0}
The mean time between two coded sub-task completions in the single forking scheme for \emph{i.i.d.} shifted exponential coded sub-task completion times in stage~$0$ is  
\EQN{ 
\label{eqn:MeanInterServiceStage0}
\E{t_{0,r}-t_{0,r-1}} = 
\begin{cases}
c + \frac{1}{\mu n_0}, & r = 1,\\
\frac{1}{\mu (n_0-r+1)}, & r \in \{2, \dots, \ell_0\}.  
\end{cases}
}
\end{lemma}
\begin{proof}
	Since $t_{0,r}$ is the completion time of first $r$ coded sub-tasks out of $n_0$ parallel coded sub-tasks, 
	we have $t_{0,r} = c + X_{r}^{n_0}$. 
	Hence, for each $r \in [\ell_0]$, we have 
\begin{equation}
		t_{0,r}-t_{0,r-1} = (c + X_{r}^{n_0})-(c+X_{r-1}^{n_0}). 
	\end{equation}
	The chunk requests are initiated at time $t_{0,0} = t_0 = 0$ and hence the first {chunk} is completed at $t_{0,1} -t_{0,0} = c + X_1^{n_0}$. 
	
	From Lemma~\ref{rem:DiffOS}, 
	we can write the following equality in distribution 
\begin{equation}
		t_{0,r} - t_{0,r-1} = \begin{cases}
			c + \frac{T'_1}{n_0}, &r = 1,\\
			\frac{T'_{r}}{(n_0-r+1)}, & r \in \set{2, \dots, \ell_0},
		\end{cases}
	\end{equation}
	where $(T'_1, \dots, T'_n)$ are \emph{i.i.d.} exponentially distributed random variables with rate $\mu$. 
	Taking expectations on both sides, we get the result.

\end{proof}

Having analyzed the Stage 0, we now  compute the mean of the interval $[t_{1,r-1}, t_{1,r})$ for each $r \in [\ell_1]$. 
The difficulty in this computation is that additional $n_1$ servers that start working on {chunk requests} at the single forking-time $t_1$, 
have an initial start-up time of $c$ due to the shifted exponential service distribution. 
Hence, none of these additional $n_1$ servers can complete service before time $t_1+c$. 
Whereas, some of the $n_0-\ell_0$ servers with unfinished {chunk requests} from stage~$0$
can finish their {chunks} in this time-interval $(t_1, t_1+c]$. 
In general, the number of chunk completions in the interval $(t_1, t_1+c]$ is a random variable, 
which we denote by $N(t_1, t_1+c) \in \set{0, \dots, n_0-\ell_0}$.

We first compute the probability mass function of this discrete valued random variable $N(t_1, t_1+c)$. 
We denote the event of $j-\ell_0$ {chunk} completions in this interval $(t_1, t_1+c]$ for any $\ell_0 \le j \le n_0$ by 
\EQN{
E_{j-\ell_0} \triangleq	\set{N(t_1, t_1+c) = j-\ell_0, t_1 =  c+X_{\ell_0}^{n_0}}.
}

\begin{lemma}
\label{lem:pdfOfNumCompletions}
The probability distribution of the number of {chunk} completions $N(t_1, t_1+c)$ in the interval $(t_1, t_1+c]$  for $\ell_0 \le j \le n_0$ is given by $p_{j-\ell_0} \triangleq P(E_{j-\ell_0})$ where 
\EQN{
\label{eqn:ProbSCStart}
p_{j-\ell_0} =  \binom{n_0-\ell_0}{j-\ell_0}(1-e^{-\mu c})^{j-\ell_0}e^{-(n_0-j)\mu c}. 
}
\end{lemma}

\begin{proof}
Let the number of service completions until time $t_1+c$ be $j \in \set{\ell_0, \dots, n_0}$. 
We can write the event of $j-\ell_0$ service completions in the interval $(t_1, t_1+c]$ as 
\EQN{
\set{X_j^{n_0} - X_{\ell_0}^{n_0} \le c}\cap\set{X_{j+1}^{n_0} - X_{\ell_0}^{n_0} \le c}^c.
}
From the definition of order statistics for continuous random variables, we have $X_j^{n_0} < X_{j+1}^{n_0}$. 
This implies that the intersection of events $\set{X_j^{n_0} - X_{\ell_0}^{n_0} \le c}$ and $\set{X_{j+1}^{n_0} - X_{\ell_0}^{n_0} \le c}$ is $\set{X_{j+1}^{n_0} - X_{\ell_0}^{n_0} \le c}$.

Therefore, from the disjointness of complementary events and probability axiom for summation of disjoint events, it follows %
\EQN{
p_{j-\ell_0}  = P \set{X_j^{n_0} - X_{\ell_0}^{n_0} \le c} - P \set{X_{j+1}^{n_0} - X_{\ell_0}^{n_0} \le c}. 
}
Due to memoryless property, we can write the above as
\EQN{
p_{j-\ell_0}  = P \set{X_{j-\ell_0}^{n_0-\ell_0} \le c} - P \set{X_{j+1-\ell_0}^{n_0-\ell_0} \le c}. 
}
From ordered statistics for exponentially distributed random variables with rate $\mu$, 
we get the required form of $p_{j-\ell_0}$ (Let $(X_1, \dots, X_n)$ be $n$ \emph{i.i.d.} random variables with common distribution function $F$, 
and we denote the  $j$th order statistics of this collection by $X_j^n$. 
Then the distribution of $X_j^n$ is given by 
$
P\set{X_j^n \le x} = \sum_{i=j}^n\binom{n}{i}F(x)^i\bar{F}(x)^{n-i}. 
$). 
\end{proof}

%

%
%
%
%
%
%

Let $\{s_1, s_2, \dots, s_{n_0-\ell_0}\}$ be the {chunk} completion times in stage~1 after the forking time $t_1$, which in definition correspond to $\{t_{1,1}= t_1+s_1, t_{1,2}= t_1+s_2, \dots, t_{1,n_0-\ell_0}= t_1+s_{n_0-\ell_0}\}$.
In stage~1, the {chunk} completions  numbered $r \in [j-\ell_0]$ are finished only by the $n_0-\ell_0$ servers within time $t_1+c$, since none of the $n_1$ servers started at forking point $t_1$ are able to finish even a single {chunk} with in the time $t_1+c$, whereas the {chunk} completions numbered $r \in \{j-\ell_0+1, \dots, k-\ell_0\}$ are finished by $n-j$ servers which include subset of combination of both left over initial servers and all forked servers.

We next find mean of $r$th completion time in stage~1 conditioned on the event $E_{j-\ell_0}$ .%

\begin{lemma} 
\label{lem:meanConditionalCompletion}
For any $r \in [j-\ell_0]$ and  $\alpha = 1-e^{-c\mu}$, we have
\EQN{
\E{s_r| E_{j-\ell_0}} = \begin{cases}
_3 F_2 \Big( {{1,1,r+1} \atop {2,j-\ell_0+2}} ; \alpha \Big)\frac{r \alpha}{\mu (j-\ell_0+1)},& r < j-\ell_0,\\
c \Big[1 - \alpha^{-r} + \sum_{i=1}^{r} \frac{\alpha^{i-r}}{ic\mu} \Big],& r = j - \ell_0.
\end{cases}
}
\end{lemma}

\begin{proof}
We denote $m = j -\ell_0$ for convenience.
Let $N(t_1, t_1+c) = m$, then $t_1+s_1, \dots, t_1+s_{m}$ are the chunk completion times of the first $m$ servers out of $n_0-\ell_0$ parallel servers in their memoryless phase in the duration $[t_1, t_1+c)$. 
In the duration $[t_{1,r-1}, t_{1,r})$ for $r \in [m]$, there are $n_0-\ell_0-r+1$ parallel servers in their memoryless phase, 
and hence the inter-service completion times $(t_{1,r}-t_{1,r-1}: r \in [m])$ are independent and distributed exponentially with parameter $\mu_r \triangleq (n_0-\ell_0 -r+1)\mu$.
Denoting $s_0=0$,  we have $s_r - s_{r-1} = t_{1,r}-t_{1,r-1}$ for each $r \in [m]$. 
From the definition of $\mu_r$'s and $p_{m}$, the independence of $s_r-s_{r-1}$, and rearrangement of terms we can write the conditional joint density of vector $s = (s_1, s_2, \dots, s_{m})$ given event $E_{m}$ as
\EQN{
	\label{eqn:conditionalpdf}
	f_{s_1, \dots, s_m|E_m} 
	= \prod_{i=1}^m\frac{i\mu e^{-\mu s_i}}{1-e^{-c\mu}}.
}
From the definition of the task completion times, 
the possible values of the vector $s = (s_1, \dots, s_{m})$ 
satisfy the constraint $0 < s_1 < \dots < s_{m} < c$. 
That is, we can write the set of possible values for vector $s$ as $A_{m}$, 
where $A_m$ is a vector of increasing co-ordinates bounded between $(0,c)$, 
and can be written as 
\EQN{
	A_m \triangleq \set{s \in \mathbb{R}^m: 0 < s_1 < \dots < s_m < c}.
}
This constraint couples the set of achievable values for the vector $s$, 
and hence even though the conditional density has a product form, 
the random variables $(s_1, \dots, s_{m})$ are not conditionally independent given the event $E_{m}$. 

To compute the conditional expectation $\E{s_r|E_m}$, we find the conditional marginal density of $s_r$ given the event $E_{m}$. 
To this end,  we integrate the conditional joint density of vector $s$ over variables without $s_r$. 
In terms of $s_r \in (0, c)$, we can write the region of integration  
as the following intersection of regions, 
\EQN{
	A_m^{-r} = \cap_{i < r}\set{0< s_i < s_{i+1}}\cap_{i >r}\set{s_{i-1} <  s_i < c}. %
}
Using the conditional density of vector $s$ defined in~\eqref{eqn:conditionalpdf} in the above equation, 
and denoting $\alpha \triangleq 1-e^{-c\mu}$ and $\alpha_r \triangleq 1- e^{-\mu s_r}$ for clarity of presentation,  
we can compute the conditional marginal density function~\citep{Ross2019}
\EQN{
	\label{eqn:ConditionalMarginalDensity}
	f_{s_r|E_{m}} = \frac{m\mu(1-\alpha_r)}{\alpha^m}\binom{m-1}{r-1}(\alpha_r)^{r-1}(\alpha-\alpha_r)^{m-r}. 
}
The conditional mean $\E{s_r|E_{m}} = \int_{0}^cf_{s_r|E_{m}}ds_r$ is obtained by integrating the conditional marginal density in~\eqref{eqn:ConditionalMarginalDensity}, over $s_r \in (0, c)$. 
For $r \in [m-1]$, the result follows from the integral identity of Remark~\ref{rem:integral1} for $x = s_r, q = r-1, p = m-1$ and $\alpha = 1-e^{-\mu c}$. 
Similarly, the result for $r = j-\ell_0$ follows from Corollary~\ref{rem:integral2}  for $x = s_m$ and $m = j-\ell_0$.
\end{proof}

In stage~1, for $1 \le r \le j-\ell_0$, we have 
\EQN{
t_{1,r}-t_{1,r-1} = (X^{n_0-\ell_0}_r-  X^{n_0-\ell_0}_{r-1}){\bf 1}_{E_{j-\ell_0}}.%
}
For $j-\ell_0+2 \le r \le k-\ell_0$, the difference $t_{1,r}-t_{1,r-1}$ is equal to
\EQN{
\label{eqn:PostShift}
(X^{n-j}_{r-j+\ell_0} - X^{n-j}_{r-j+ \ell_0-1}){\bf 1}_{E_{j-\ell_0}}.%
}
When $r = j-\ell_0+1$, we write the time difference between $r$th and $(r-1)$th {chunk} completion instants as 
\EQN{
t_{1,r}-t_{1,r-1}  =  t_{1,r}-(t_1 + c) + (t_1 + c) - t_{1,r-1}.
}
For $r = j-\ell_0+1$, we have $t_{1,r-1} \le t_1 + c < t_{1,r}$. 
In the disjoint intervals $[t_{1,r-1}, t_1+c)$ and $[t_1+c,  t_{1,r})$, there are $n_0-j$ and $n-j$ \emph{i.i.d.} exponentially distributed parallel servers respectively. 
Since the age and excess service times of exponential random variables are independent at any constant time, 
we have independence of $t_{1,r}-(t_1 + c)$ and $(t_1 + c) - t_{1,r-1}$ for $r=j-\ell_0+1$.

Conditioned on the event $E_{j-\ell_0}$ of $j-\ell_0$ {chunk} completions in the interval $(t_1, t_1+c]$, 
the conditional mean of inter-{chunk} completion time in stage~1 is 
\begin{eqnarray}
&\E {(t_{1,r}-t_{1,r-1}) | E_{j-\ell_0}} = \mathbb{E}[(s_r-s_{r-1})({\bf 1}_\{j-\ell_0 > r -1\}\\
&+ {\bf 1}_\{j-\ell_0 = r -1\} + {\bf 1}_\{j-\ell_0 < r -1\})|E_{j-\ell_0}].
\end{eqnarray}

\begin{lemma}
\label{lem:conditionalInterServiceStage1}
For any $r \in [k-\ell_0]$, and  $\alpha = 1-e^{-c\mu}$
the conditional mean $\E {(t_{1,r}-t_{1,r-1}) | E_{j-\ell_0} } $ equals
\EQN{
\label{eqn:ConditionalMeanInterServiceStage1_c1}
\begin{cases}
 {_2 F_1} \Big( {{1,r} \atop {j-\ell_0+2}} ; \alpha \Big) \frac{r \alpha}{\mu (j-\ell_0+1)}, & r < j-\ell_0+1,\\
c \Big[\frac{1}{\alpha^{(r-1)}} - \sum\limits_{i=1}^{r-1} \frac{\alpha^{i-r+1}}{ic\mu} \Big]+\frac{1}{\mu(n-j)}, & r = j -\ell_0+1,\\
\frac{1}{\mu(n-\ell_0-r+1)}, & r > j -\ell_0+ 1.
\end{cases}
}
\end{lemma}

\begin{proof}
Recall that, we have $n_0-\ell_0$ parallel servers in their memoryless phase working on individual {chunks} in the interval $(t_1, t_1+c]$. 
In this duration, $N(t_1, t_1+c)$ {chunks} are completed and additional $n_1$ parallel servers start their memoryless phase at time $t_1+c$. 

We first consider the case when $r -1 > N(t_1, t_1+c) = j -\ell_0$.  
This implies that $t_{1,r-1} > t_1 + c$ and there are $n-\ell_0-r+1$ parallel servers in their memoryless phase working on remaining {chunks}. 
From Lemma~\ref{rem:DiffOS}, the following equality holds in distribution  
\EQN{
t_{1,r}-t_{1,r-1} = \frac{T'_r}{n-\ell_0-r+1}.
}
Recall that %
$E_{j-\ell_0} \in \sigma(T'_1, \dots, T'_{j-\ell_0+1})$, 
and since $(T'_i: i \in {\mathbb{N}})$ is an \emph{i.i.d.} sequence, 
it follows that $t_{1,r}-t_{1,r-1}$ is independent of the event $E_{j-\ell_0} $ for $r > j-\ell_0+1$
and hence $\E{t_{1,r}-t_{1,r-1}| E_{j-\ell_0} } =\E{t_{1,r}-t_{1,r-1}}$. 
The result follows from the fact that $\E{T'_i} = \frac{1}{\mu}$. 

We next consider the case when $r -1 = N(t_1,t_1+c) = j - \ell_0$. 
By definition of $N(t_1,t_1+c)$, we have $t_{1,r-1} \le t_1+c < t_{1,r}$. 
In the disjoint intervals $(t_{1,r-1}, t_1+c]$ and $(t_{1}+c, t_{1,r}]$, there are $n_0-j$ and $n-j$ \emph{i.i.d.} exponentially distributed parallel servers respectively. 
Therefore, writing $t_{1,r} - t_{1,r-1}$ as $(t_{1,r} - (t_1+c) ) + ((t_1+c) - t_{1,r-1})$,  
and using Lemma~\ref{rem:DiffOS}, 
we compute the conditional mean of the first part as %
\EQN{
\E{ t_{1, r}-(t_1+c)\vert E_{j-\ell_0} } = \E{\frac{T'_{r+1}}{n-j}} = \frac{1}{\mu (n-j)}.
}
By using the fact $t_{1,r-1} = t_1+s_{r-1}$, we can write the conditional mean of the second part as %
$\E{t_1+c-t_{1,r-1}\vert E_{j-\ell_0} } %
= c - \E{s_{r-1} \vert E_{j-\ell_0} },$
where $\E{s_{r-1} \vert E_{j-\ell_0} }$ is given by Lemma~\ref{lem:meanConditionalCompletion}.
Summing these two parts, we get the conditional expectation for $r=j-\ell_0+1$. 

For the case when $r \in [j-\ell_0]$, the result follows from Lemma~\ref{lem:meanConditionalCompletion} 
and the fact $t_{1,r} = t_1+s_{r}$.
\end{proof}

We next compute the unconditional mean of inter-{chunk} completion time $\E{(t_{1,r}-t_{1,r-1})}$ by averaging out the conditional mean $\E {(t_{1,r}-t_{1,r-1}) | E_{j-\ell_0}}$ over all possible values of $j$. 
We denote $m=j-\ell_0$ for convenience.
\begin{corollary}
\label{cor:StageLength1r}
For each $r \in [k-\ell_0]$, by considering all possible values of $m$ from the set $\{0,1, \dots, n-\ell_0\}$, the mean inter-service completion time in stage~1, is
\begin{eqnarray}
& \E{t_{1,r}-t_{1,r-1}} = \sum_{m:m+1<r}^{}  p_{m} \frac{1}{\mu(n-\ell_0-r+1)}\\
&+ \sum_{m:m+1=r}^{}  p_{m} \Big[c \Big[\frac{1}{\alpha^{(r-1)}} - \sum\limits_{i=1}^{r-1} \frac{\alpha^{i-r+1}}{ic\mu} \Big]+\frac{1}{\mu(n-j)}\Big] \\
&+ \sum_{m:m+1>r}^{} {_2 F_1} \Big( {{1,r} \atop {m+2}} ; \alpha \Big) \frac{r \alpha}{\mu (m+1)} p_{m}.
\end{eqnarray}
\end{corollary}
\begin{proof}
The result follows by using Lemma~\ref{lem:conditionalInterServiceStage1} and from the tower property of nested expectations
\EQN{
\E{(t_{1,r}-t_{1,r-1})} = \E{\E{(t_{1,r}-t_{1,r-1})|E_{j-\ell_0} }},
}
and the fact that $N(t_1, t_1+c) \in \set{0, \dots, n_0-\ell_0}$ and $p_{m}$ is defined in~\eqref{eqn:ProbSCStart}, as the probability of the number of service completions $N(t_1, t_1+c)$ in the interval $(t_1,t_1+c]$ being $m=j-\ell_0$ where $t_1$ is the time of $\ell_0$ completions of initial $n_0$ chunks. 
\end{proof}

\section{Characterization of Mean Service Completion Time and Mean Server Utilization Cost  for Single Job}\label{dsr_metrics}
We are now ready to compute the means of service completion time and server utilization cost. We first consider the metrics in the Stage 0, based on Lemma \ref{lem:InterArrivalStage0}.

\begin{lemma}
\label{cor:Stage1Cost}
Consider single-forking with \emph{i.i.d.} shifted exponential {coded sub-task} completion times and initial number of servers $n_0$ in stage~$0$.  
The mean forking time is given by 
 \EQN{ 
\label{eqn:MeanForkingPoint}
\E{t_1} = c + \sum_{r=1}^{\ell_0} {\frac{1}{\mu(n_0-r+1)}}.
}
The mean server utilization cost in stage~$0$ is given by 
\EQN{
\label{eqn:MeanServerUtilizationCostStage0}
\E{W_0} = \frac{\lambda}{\mu}(\ell_0 + \mu n_0 c).
}
\end{lemma}
\begin{proof}
We can write the completion time $t_1$ of $\ell_0$th coded sub-task out of $n_0$ in parallel, 
as a telescopic sum of length of coded sub-task completions given in~\eqref{eqn:ServiceCompletionTime}.
Taking expectations on both sides, 
the mean forking point $\E{t_1}$ follows from the the linearity of expectations and 
the mean length of each coded sub-task completion~\eqref{eqn:MeanInterServiceStage0}.  

Taking expectation of the server utilization cost in~\eqref{eqn:ServerUtilizationCost}, 
the mean server utilization cost $\E{W_0}$ in stage~$0$ follows from the linearity of expectations 
and the mean length of each coded sub-task completion~\eqref{eqn:MeanInterServiceStage0}. 
\end{proof}

Next, we consider two possibilities for the initial number of servers $n_0$: when $n_0 < k$ and otherwise.

Note that when $n_0 < k$, then $t_2 > t_1 + c$ almost surely, since $k$ {coded sub-tasks} can never be finished by initial $n_0$ servers. The next result computes the mean service completion time and mean server utilization cost for $n_0 < k$ case.

%
%
%
%
%
%
%
%
%
%
%
%

%
%

%

%
%

%
%

%
\begin{theorem}
\label{thm:AnalyticalResults_case1_ShiftedExp}
For the single forking case with $n$ total servers for $k$ {sub-tasks} and initial number of servers $n_0 < k$, 
the mean server utilization cost is 
\EQN{ 
\label{eqn:MeanServerUtilizationCost_fp1_c1_Exp1}
\E W = \lambda n c + \frac{\lambda k}{\mu},
}
and the mean service completion time is
\EQN{  
\label{eqn:MeanCompletionTime_fp1_c1_1}
\E{t_2} = c + \E{t_1} + \frac{1}{\mu}\sum_{j=\ell_0}^{n_0}  p_{j-\ell_0}\sum_{i=j}^{k-1}\frac{1}{(n-i)},
}
where $\E{t_1}$ is given in~\eqref{eqn:MeanForkingPoint} and $p_{j-\ell_0}$ is given in~\eqref{eqn:ProbSCStart}.
\end{theorem}
\begin{proof}
The proof follows by substituting inter-chunk times in \eqref{eqn:ServerUtilizationCost} and \eqref{eqn:ServiceCompletionTime} and simplifications. The details are omitted, while can be seen in \citep{Straggler2020Inf}. 
\end{proof}

From  Theorem~\ref{thm:AnalyticalResults_case1_ShiftedExp}, we observe that the mean server utilization cost remains same for all values of initial number of servers $n_0 < k$ and forking threshold $\ell_0$.  The mean service completion time decreases as we increase the number of initial servers, and thus at $n_0=n$ will have lower mean service completion time. Further, the mean service utilization cost for $n_0=n$ can be easily shown to be $ \lambda n c + \frac{\lambda k}{\mu}$ which is the same as that for all $n_0<k$. Thus, as compared to no forking ($n_0=n$), the single forking with $n_0<k$ has the same mean server utilization cost while it has higher mean service completion time. 
Thus, this regime doesn't provide any tradeoff point between service completion time and server utilization cost which is worse than no-forking, and hence the only region of interest for a system designer is $n_0\ge k$, which is studied in the following.

%
%
%
For $n_0\ge k$,  the number of completed {chunks} $\ell_0$ at the forking point $t_1$ are in $\{0, 1, \dots, k\}$. 
There are three different possibilities for completing $k$ {chunks}. 
First possibility is $\ell_0=k$, when all the required $k$ {chunks} are finished on initial $n_0$ servers without any forking. 
In this case, $t_2 = t_1$. 
For the next two possibilities, $\ell_0 < k$ and hence forking is needed. 

Second possibility is $\ell_0 < k$ and $\ell_0 + N(t_1, t_1+c) = j \le k-1$, where $j-\ell_0$ service completions occur in the duration $[t_1, t_1+c)$ and $\ell_0 \le j \le k-1$. 
This implies that even though $n_0 > k$, the total {chunks} finished until instant $t_1+c$ are still less than $k$ and remaining $k-j > 0$ {chunks} among the required $k$ are completed only after $t_1+c$, 
when $n-j$ parallel servers are in their memoryless phase.  
In this case, $t_2 = t_1+ c + X_{k-j}^{n-j}$ for $N(t_1, t_1+c) = j-\ell_0 \in \set{0, \dots, k-\ell_0-1}$. 

Third possibility is when $\ell_0< k$ and $\ell_0 + N(t_1, t_1+c) \ge k$. %
That is, even though the {chunks} are forked on additional $n_1$ servers at time $t_1$, 
the {job} is completed at $k$ out of $n_0$ initial servers before the constant start-up time of these additional $n_1$servers is finished. 
This happens when $s_{k-\ell_0} \le c$ and in this case, $t_2 = t_1 + s_{k-\ell_0}$ for $N(t_1, t_1+c) \ge k-\ell_0$. Recall that $s_{k-\ell_0}$ is the $(k-\ell_0)$th service completion in stage~1 after $t_1$.
Summarizing all the results, we write the service completion time in the case $n_0 \ge k$ and $N(t_1, t_1+c) = j - \ell_0$ as 
\EQN{
t_2 = t_1 + s_{k-\ell_0}{\bf 1}_{\{\ell_0 < k \le j\}} + (c + X^{n-j}_{k-j}){\bf 1}_{\{\ell_0 \le j < k\}}.
}
For $n_0 \ge k$, the mean service completion time and the mean server utilization cost are given in the following theorem.

\begin{theorem}
\label{thm:AnalyticalResults_case2_ShiftedExp}
In single forking scheme, for $n_0 \ge k$ case, the mean service completion time $\E{t_2}$ is
\EQN{ 
\label{eqn:MeanCompletionTime_fp1_c2_Exp2}
\E{t_1} + \Big[\sum\limits_{r=1}^{k-\ell_0} \E{t_{1,r}-t_{1,r-1}}\Big] {\bf 1}_{\{\ell_0 < k\}}
}
and the mean server utilization cost $\E{W}$ is
\EQN{  
\label{eqn:MeanServerUtilizationCost_fp1_c2_Exp1}
\E{W_0} + \lambda \sum\limits_{r=1}^{k-\ell_0} (n-\ell_0-r+1) \E{t_{1,r} - t_{1,r-1}} {\bf 1}_{\{\ell_0 < k\}}.
}
Where $\E{t_{1,r}-t_{1,r-1}}$ in the above expressions is given by Corollary~\ref{cor:StageLength1r}.
\end{theorem}
\begin{proof}
In Lemma~\ref{cor:Stage1Cost}, we have already computed the mean completion time $\E{t_1}$ of stage~0, 
and the mean server utilization cost $\E{W_0}$ in stage~0. 
Recall that since completion of any $k$ {chunks} suffice for the job completion, the forking threshold $\ell_0 \le k$. 

We first consider the case when $\ell_0=k$. 
In this case, we do not need to add any further servers because all the required tasks are already finished in stage~0 itself. 
Hence, there is no need of forking in this case, and the mean service completion time is given by $\E{t_1}$ and the mean server utilization cost is given by $\E{W_0}$. 

We next consider the case when $\ell_0 < k$. 
In this case, the job completion occurs necessarily in stage~1. 
Thus, we need to compute $\E{t_2-t_1}$ and $\E{W_1}$ in order to evaluate the mean service completion time $\E{t_2}$ and the mean server utilization cost $\E{W_0+W_1}$.  
The duration of stage~1 can be written as a telescopic sum of inter service times 
\EQN{
t_2 - t_1 = \sum_{r=1}^{k-\ell_0-1}(t_{1,r}-t_{1,r-1}). 
}
Further for $\ell_0 < k$, the number of servers that are active in stage~1 after $(r-1)$th service completions are $n-\ell_0-r+1$ and the associated cost incurred in the interval $[t_{1,r-1},t_{1,r})$ is $\lambda(t_{1,r}-t_{1,r-1})(n-\ell_0-r+1)$.
Therefore, we can write the server utilization cost in stage~1 as 
\EQN{
W_1 = \lambda\sum_{r=1}^{k-\ell_0-1}(n-\ell_0-r+1)(t_{1,r}-t_{1,r-1}). 
}
The result follows from taking mean of the duration $t_2-t_1$ and server utilization cost $W_1$, 
from the linearity of expectations, 
and considering both possible cases. 
\end{proof}

We observe that when $n_0\ge k$, the mean service utilization cost depends on the initial number of servers $n_0$ as well as the total number of servers $n$, 
unlike the case $n_0<k$ where this cost depends only on the total number of servers $n$. 

\section{Simulations}\label{sec:dsr_sims}
For numerical evaluation of mean service completion time and mean server cost utilization for single forking systems, 
we choose the following system parameters. 
We select the sub-task fragmentation of a single job as $k=12$, 
and a maximum redundancy factor of $n/k = 2$.
That is, we choose the total number of servers $n=24$. 
We take the server utilization cost rate to be $\lambda = 1$. Coded-task completion time at each server was chosen to be an \emph{i.i.d.} random variable having a shifted exponential distribution. 
For numerical studies in this section, 
we choose the shift parameter $c=1$ and the exponential rate $\mu = 0.5$. Since it was already shown that $n_0<k$ is not a useful regime, we consider the case where $n_0\ge k$. 

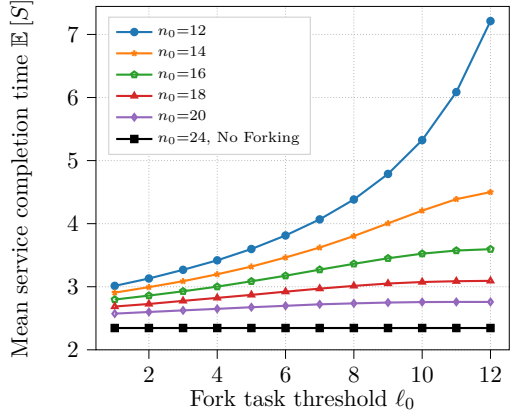
\begin{figure}[hbt]
	\centering
	\scalebox{0.8}{\begin{tikzpicture}

\definecolor{color0}{rgb}{0.12156862745098,0.466666666666667,0.705882352941177}
\definecolor{color1}{rgb}{1,0.498039215686275,0.0549019607843137}
\definecolor{color2}{rgb}{0.172549019607843,0.627450980392157,0.172549019607843}
\definecolor{color3}{rgb}{0.83921568627451,0.152941176470588,0.156862745098039}
\definecolor{color4}{rgb}{0.580392156862745,0.403921568627451,0.741176470588235}

\begin{axis}[
legend cell align={left},
legend style={at={(0.03,0.97)}, anchor=north west, draw=white!80.0!black, font=\scriptsize},
tick align=outside,
tick pos=left,
x grid style={white!69.01960784313725!black, densely dotted},
xlabel={Fork task threshold $\ell_0$},
xmajorgrids,
xmin=0.45, xmax=12.55,
xminorgrids,
xtick style={color=black},
y grid style={white!69.01960784313725!black, densely dotted},
ylabel={Mean service completion time $\E S$},
ymajorgrids,
ymin=2, ymax=7.44410480988649,
yminorgrids,
ytick style={color=black}
]
\addplot [semithick, color0, mark=*, mark size=1.5, mark options={solid}, line width=1pt]
table {%
1 3.01369871303481
2 3.13064058306519
3 3.26805703470636
4 3.41868219759926
5 3.59808696609605
6 3.81288393707227
7 4.06765553390932
8 4.38295352681885
9 4.78823797768967
10 5.32399156499496
11 6.08611240297693
12 7.2121738197126
};
\addlegendentry{$n_0$=12}
\addplot [semithick, color1, mark=star, mark size=1.5, mark options={solid}, line width=1pt]
table {%
1 2.90639937570915
2 2.99354020673869
3 3.08749928139773
4 3.19752864384908
5 3.31955575335801
6 3.46390508335684
7 3.62149319412232
8 3.80377886932076
9 4.00495511413662
10 4.20686095143819
11 4.38824476115839
12 4.50055674864267
};
\addlegendentry{$n_0$=14}
\addplot [semithick, color2, mark=pentagon, mark size=1.5, mark options={solid}, line width=1pt]
table {%
1 2.79689765752927
2 2.8596294889118
3 2.92768797465149
4 3.00153575524783
5 3.08639154494779
6 3.174278416553
7 3.27083044763372
8 3.36299870808061
9 3.45116419999778
10 3.52439540064159
11 3.57378672343268
12 3.5953827491408
};
\addlegendentry{$n_0$=16}
\addplot [semithick, color3, mark=triangle, mark size=1.5, mark options={solid}, line width=1pt]
table {%
1 2.68632562785407
2 2.726749567034
3 2.77583348917678
4 2.82321398611927
5 2.87082918988658
6 2.92188779575518
7 2.97000413507778
8 3.01330382855873
9 3.05089786250224
10 3.07478749509631
11 3.08780035066057
12 3.09344398527792
};
\addlegendentry{$n_0$=18}
\addplot [semithick, color4, mark=diamond, mark size=1.5, mark options={solid}, line width=1pt]
table {%
1 2.57355401623466
2 2.59974378281649
3 2.62502477923003
4 2.64937159650102
5 2.67509959477866
6 2.69767030139652
7 2.72127641146894
8 2.73632231644662
9 2.74926272369104
10 2.75686473292883
11 2.75914419091818
12 2.75888603645032
};
\addlegendentry{$n_0$=20}
\addplot [semithick, black, mark=square*, mark size=1.5, mark options={solid}, line width=1pt]
table {%
1 2.3454949990856573
2 2.3454949990856573
3 2.3454949990856573
4 2.3454949990856573
5 2.3454949990856573
6 2.3454949990856573
7 2.3454949990856573
8 2.3454949990856573
9 2.3454949990856573
10 2.3454949990856573
11 2.3454949990856573
12 2.3454949990856573
};
\addlegendentry{$n_0$=24, No Forking}
\end{axis}

\end{tikzpicture}}
	\caption{
		For the setting $n_0 \ge k$, this graph displays the mean service completion time $\E S$ as a function of fork task threshold $\ell_0$ for single forking with the total number of servers $n=24$,  the total needed {coded sub-tasks} $k=12$, and different numbers of initial servers $n_0 \in \{12,14,16,18,20\}$. 
		The single {coded sub-task} execution time at servers are assumed to be \emph{i.i.d.} shifted exponential distribution with shift $c=1$ and rate $\mu=0.5$.
	}
	\label{Fig:MeanServiceSingleFork_c2}
\end{figure}

\begin{figure}[hhh]
	\centering
	\scalebox{0.8}{\begin{tikzpicture}

\definecolor{color0}{rgb}{0.12156862745098,0.466666666666667,0.705882352941177}
\definecolor{color1}{rgb}{1,0.498039215686275,0.0549019607843137}
\definecolor{color2}{rgb}{0.172549019607843,0.627450980392157,0.172549019607843}
\definecolor{color3}{rgb}{0.83921568627451,0.152941176470588,0.156862745098039}
\definecolor{color4}{rgb}{0.580392156862745,0.403921568627451,0.741176470588235}

\begin{axis}[
legend cell align={left},
legend style={at={(0.03,0.03)}, anchor=south west, draw=white!80.0!black, font=\scriptsize},
tick align=outside,
tick pos=left,
x grid style={white!69.01960784313725!black, densely dotted},
xlabel={Fork task threshold $\ell_0$},
xmajorgrids,
xmin=0.45, xmax=12.55,
xminorgrids,
xtick style={color=black},
y grid style={white!69.01960784313725!black, densely dotted},
ylabel={Mean server utilization cost $\E W$},
ymajorgrids,
yminorgrids,
ytick style={color=black}
]
\addplot [semithick, color0, mark=*, mark size=1.5, mark options={solid}, line width=1pt]
table {%
1 47.9883301079489
2 47.9717837439924
3 48.0156764441777
4 47.9950859658202
5 47.9935730379279
6 48.0099177994271
7 47.9869548034981
8 47.9177052769452
9 47.7644096898755
10 47.2776875894062
11 45.4310247350799
12 36.0110720716353
};
\addlegendentry{$n_0$=12}
\addplot [semithick, color1, mark=star, mark size=1.5, mark options={solid}, line width=1pt]
table {%
1 48.0282116991027
2 48.0259046783065
3 47.9661336534825
4 47.9832382446899
5 47.9367421361971
6 47.9213532131114
7 47.7825353994387
8 47.5337477045851
9 47.0161177852256
10 45.7502350813319
11 43.1677119308395
12 37.9970002044163
};
\addlegendentry{$n_0$=14}
\addplot [semithick, color2, mark=pentagon, mark size=1.5, mark options={solid}, line width=1pt]
table {%
1 47.998631803841
2 47.9990804395517
3 47.9688589210967
4 47.909261338186
5 47.8757490708588
6 47.7331798818781
7 47.5178293191262
8 47.0359515478747
9 46.2311062961805
10 44.9487057812259
11 42.9353848313557
12 40.0082132848927
};
\addlegendentry{$n_0$=16}
\addplot [semithick, color3, mark=triangle, mark size=1.5, mark options={solid}, line width=1pt]
table {%
1 47.9754382437841
2 47.9207753366371
3 47.9407928854363
4 47.8756616593541
5 47.714143186272
6 47.5371545032379
7 47.2114987886577
8 46.7242893986762
9 46.0244489784556
10 45.0135621590285
11 43.673570816701
12 42.0332857675115
};
\addlegendentry{$n_0$=18}
\addplot [semithick, color4, mark=diamond, mark size=1.5, mark options={solid}, line width=1pt]
table {%
1 47.9471447741816
2 47.9329419593561
3 47.8790269239654
4 47.7613287460646
5 47.6407180297185
6 47.416848866953
7 47.1815051215408
8 46.7701851044489
9 46.2809790346811
10 45.6531294502228
11 44.8803534534508
12 43.9863979002725
};
\addlegendentry{$n_0$=20}
\addplot [semithick, black, mark=square*, mark size=1.5, mark options={solid}, line width=1pt]
table {%
1 48
2 48
3 48
4 48
5 48
6 48
7 48
8 48
9 48
10 48
11 48
12 48
};
\addlegendentry{$n_0$=24, No Forking}
\end{axis}

\end{tikzpicture}}
	\caption{
		For the setting $n_0 \ge k$, this graph displays the mean server utilization cost $\E W$ as a function of fork task threshold $\ell_0$ for single forking with the total number of servers $n=24$, the total needed {coded sub-tasks} $k=12$, and different numbers of initial servers $n_0 \in \{12,14,16,18,20\}$. 
		The single {coded sub-task} execution time at servers are assumed to be \emph{i.i.d.} shifted exponential distribution with shift $c=1$ and rate $\mu=0.5$.
	}
	\label{Fig:MeanCostSingleFork_c2}
\end{figure}
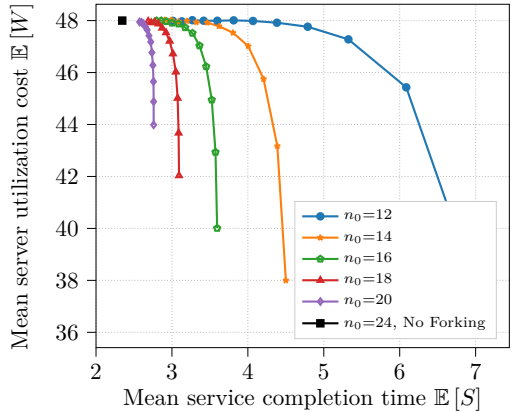
\begin{figure}[hhh]
	\centering
	\scalebox{0.8}{\begin{tikzpicture}

\definecolor{color0}{rgb}{0.12156862745098,0.466666666666667,0.705882352941177}
\definecolor{color1}{rgb}{1,0.498039215686275,0.0549019607843137}
\definecolor{color2}{rgb}{0.172549019607843,0.627450980392157,0.172549019607843}
\definecolor{color3}{rgb}{0.83921568627451,0.152941176470588,0.156862745098039}
\definecolor{color4}{rgb}{0.580392156862745,0.403921568627451,0.741176470588235}

\begin{axis}[
legend cell align={left},
legend style={at={(0.48,0.025)}, anchor=south west, draw=white!80.0!black, font=\scriptsize},
tick align=outside,
tick pos=left,
x grid style={white!69.01960784313725!black, densely dotted},
xlabel={Mean service completion time $\E S$},
xmajorgrids,
xmin=2, xmax=7.44410480988649,
xminorgrids,
xtick style={color=black},
y grid style={white!69.01960784313725!black, densely dotted},
ylabel={Mean server utilization cost $\E W$},
ymajorgrids,
ymin=35.4102150902619, ymax=48.6290686804761,
yminorgrids,
ytick style={color=black}
]
\addplot [semithick, color0, mark=*, mark size=1.5, mark options={solid}, line width=1pt]
table {%
3.01369871303481 47.9883301079489
3.13064058306519 47.9717837439924
3.26805703470636 48.0156764441777
3.41868219759926 47.9950859658202
3.59808696609605 47.9935730379279
3.81288393707227 48.0099177994271
4.06765553390932 47.9869548034981
4.38295352681885 47.9177052769452
4.78823797768967 47.7644096898755
5.32399156499496 47.2776875894062
6.08611240297693 45.4310247350799
7.2121738197126 36.0110720716353
};
\addlegendentry{$n_0$=12}
\addplot [semithick, color1, mark=star, mark size=1.5, mark options={solid}, line width=1pt]
table {%
2.90639937570915 48.0282116991027
2.99354020673869 48.0259046783065
3.08749928139773 47.9661336534825
3.19752864384908 47.9832382446899
3.31955575335801 47.9367421361971
3.46390508335684 47.9213532131114
3.62149319412232 47.7825353994387
3.80377886932076 47.5337477045851
4.00495511413662 47.0161177852256
4.20686095143819 45.7502350813319
4.38824476115839 43.1677119308395
4.50055674864267 37.9970002044163
};
\addlegendentry{$n_0$=14}
\addplot [semithick, color2, mark=pentagon, mark size=1.5, mark options={solid}, line width=1pt]
table {%
2.79689765752927 47.998631803841
2.8596294889118 47.9990804395517
2.92768797465149 47.9688589210967
3.00153575524783 47.909261338186
3.08639154494779 47.8757490708588
3.174278416553 47.7331798818781
3.27083044763372 47.5178293191262
3.36299870808061 47.0359515478747
3.45116419999778 46.2311062961805
3.52439540064159 44.9487057812259
3.57378672343268 42.9353848313557
3.5953827491408 40.0082132848927
};
\addlegendentry{$n_0$=16}
\addplot [semithick, color3, mark=triangle, mark size=1.5, mark options={solid}, line width=1pt]
table {%
2.68632562785407 47.9754382437841
2.726749567034 47.9207753366371
2.77583348917678 47.9407928854363
2.82321398611927 47.8756616593541
2.87082918988658 47.714143186272
2.92188779575518 47.5371545032379
2.97000413507778 47.2114987886577
3.01330382855873 46.7242893986762
3.05089786250224 46.0244489784556
3.07478749509631 45.0135621590285
3.08780035066057 43.673570816701
3.09344398527792 42.0332857675115
};
\addlegendentry{$n_0$=18}
\addplot [semithick, color4, mark=diamond, mark size=1.5, mark options={solid}, line width=1pt]
table {%
2.57355401623466 47.9471447741816
2.59974378281649 47.9329419593561
2.62502477923003 47.8790269239654
2.64937159650102 47.7613287460646
2.67509959477866 47.6407180297185
2.69767030139652 47.416848866953
2.72127641146894 47.1815051215408
2.73632231644662 46.7701851044489
2.74926272369104 46.2809790346811
2.75686473292883 45.6531294502228
2.75914419091818 44.8803534534508
2.75888603645032 43.9863979002725
};
\addlegendentry{$n_0$=20}
\addplot [semithick, black, mark=square*, mark size=1.5, mark options={solid}, line width=1pt]
table {%
2.3454949990856573 48
};
\addlegendentry{$n_0$=24, No Forking}

\end{axis}

\end{tikzpicture}}
	\caption{
		For the setting $n_0 \ge k$, we have plotted the mean server utilization cost $\E W$ as a function of the mean service completion time $\E S$ by varying fork task threshold {$\ell_0 \in [n_0]$} in single forking. 
		The total number of servers considered are $n=24$, the total {coded sub-task} needed are $k=12$. 
		The single {coded sub-task} execution time at servers are assumed to be \emph{i.i.d.} shifted exponential distribution with shift $c=1$ and rate $\mu=0.5$.
		We have plotted the same curve for different values of initial servers $n_0 \in \{12,14,16,18,20\}$. 
		For each curve, $\ell_0$ increasing from left to right. 
	}
	\label{Fig:TradeoffSingleFork_c2}
\end{figure}

To this end, we plot the  mean service completion time in Figure~\ref{Fig:MeanServiceSingleFork_c2} and mean server utilization in Figure~\ref{Fig:MeanCostSingleFork_c2},
both as a function of fork-task threshold $\ell_0 \in [k]$, 
for different values of initial servers $n_0 \in \{12,14,16,18,20\}$. 
The analytical results in Theorem~\ref{thm:AnalyticalResults_case2_ShiftedExp} are substantiated by observing that the mean service completion time $\E{S}$ increases with increase in fork-task threshold $\ell_0$ and decreases with increase in initial number of servers $n_0$. 
Further, the mean server utilization cost $\E{W}$ decreases with increase in fork-task threshold $\ell_0$. 
Thus, there is a tradeoff between the two performance measures as a function of fork-task threshold $\ell_0$. 
The tradeoff between the two performance metrics of interest is plotted in Figure~\ref{Fig:TradeoffSingleFork_c2}, which suggests that the number of initial servers $n_0$ and the forking threshold $\ell_0$ affords a true tradeoff between these metrics.

It is interesting to observe the behavior of mean server utilization cost as a function of initial number of servers $n_0$ in Figure~\ref{Fig:MeanCostSingleFork_c2}.
We note that for each fork-task threshold $\ell_0$, there exists an optimal number of initial servers $n_0$ that minimizes the server utilization cost. 
We further observe in Figure~\ref{Fig:TradeoffSingleFork_c2} that for $n_0=20$, the mean service completion time increases only $17.635\%$ while the mean server utilization cost can be decreased $8.3617\%$ by an appropriate choice of $\ell_0$ as compared to choosing no forking case of $n_0=n$.  
However, a value of $\ell_0$ cannot be chosen for $n_0=20$ that reduces the mean server utilization cost beyond $8.3617\%$. 
In order to have further reduction in mean server utilization cost, we can choose $n_0$ to $18$ which helps to decrease mean server completion time by $12.43\%$ at an expense of $31.888\%$ increase in mean service completion time as compared to the no forking case $n_0=n$. 
The intermediate points on the curve of $n_0=18$ further provide tradeoff points that can be chosen based on the desired combination of the two measures as required by the system designer. 
The choice of $n_0=12$ further helps decrease the mean server utilization cost by $24.976\%$ by having $207.49\%$ times increase in the mean service completion time as compared to the no forking case $n_0=n$. 
Thus, we see that appropriate choice of $n_0$ and $\ell_0$ provide tradeoff points that help minimizing the mean server utilization cost at the expense of the mean service completion time.

\section{Notes and Open Problems}\label{sec:dsr_notes}
This problem has been studied in the context of straggler mitigation problem, where some tasks have run-time variability. Existing solution techniques for straggler mitigation 
fall into two categories: i) Squashing runtime variability via
preventive actions such as blacklisting faulty machines that
frequently exhibit high variability \citep{dean2008mapreduce,dean2012achieving} or learning the
characteristics of task-to-node assignments that lead to high
variability and avoiding such problematic task-node pairings
\citep{yadwadkar2012proactive}, ii) Speculative execution by launching the tasks together with replicas and waiting only for the fastest copy to complete \cite{ananthanarayanan2013effective}
\citep{ananthanarayanan2010reining,zaharia2008improving,melnik2010dremel}. Because runtime variability is caused by intrinsically complex reasons, preventive measures for stragglers
could not fully solve the problem and runtime variability continued plaguing the compute workloads \citep{dean2012achieving,ananthanarayanan2013effective}. Speculative
task execution on the other hand has proved to be an effective
remedy, and indeed the most widely deployed solution for
stragglers \citep{dean2013tail,ren2015hopper}.

Even though the technique of delayed relaunch with erasure coding of tasks was originally proposed for straggler mitigation, it is also applicable to accessing erasure coded chunks from distributed storage. The authors of \citep{aktacs2019straggler} provided a single fork analysis with coding, where $k$ chunk requests are started at $t=0$. Further, after a fixed deterministic time $\Delta$, additional $n-k$
chunk requests are started. While this lays an important problem, this chapter considers the following differences to the approach: (i) we allow for general number of starting chunk requests,
(ii) the start time of new chunk requests is random and
based on the completion time of certain number of chunk requests rather than a fixed constant, and
(iii) our framework allows for an optimization of different parameters
to provide a tradeoff between service utilization
cost and service completion time. As shown in the evaluation results, the choice of $n_0=k$ is not always optimal, which additionally motivates such setup. This analysis has been considered in \citep{Straggler2020Inf,Straggler2020TON}. 

The approach could be further extended on the following directions:
\begin{enumerate}
	\item {\bf Queueing Analysis}: The proposed framework in this chapter considers a non-queueing system with a single job. The analysis with multiple arrivals is open. 
	\item {\bf General Service Distribution}: The analysis in this chapter is limited to shifted-exponential service times. Even though Parteo distribution has been considered in \citep{aktacs2019straggler}, considering general service time distribution is important. 
	\item {\bf Multiple Forking Points}: We only considered single forking. Additional benefits to more forking points is an open problem. 
	\item {\bf Use for Distributed Gradient Descent}: The approach can be used for straggler mitigation with gradient codes \citep{sasi2019straggler}, where it has been shown that the delay allows for lower amount of computation per node. Thus, the results of delayed relaunch scheduling can be used for distributed gradient computations in addition to that for erasure-coded storage.
\end{enumerate}

\chapter{Analyzing Latency for Video Content}\label{sec:video}

In this Chapter, we extend the setup to assume that the servers store video content. Rather than downloading the content, the users are streaming the content, which makes the notion of stall duration more important. We explain the system model in Section \ref{vs_model}. The downlaod and play times of different segments in a video is characterized in Section \ref{sec:dtpt_vs}. This is further used to characterize upper bounds on mean stall duration and tail stall duration in Sections \ref{sec:msd} and \ref{sec:tsd}, respectively. Sections~\ref{sec:vsr_sims} and \ref{sec:vsr_notes} contain simulation results and notes on future directions, respectively.

\section{Modeling Stall Duration for Video Requests}\label{vs_model}
We consider a distributed storage system consisting of $m$ heterogeneous servers
(also called storage nodes), denoted by $\mathcal{M}=1,2,...,m$. Each video file $i$, where $i=1,2,...r,$ is divided into $L_{i}$ equal segments, $G_{i,1}, \cdots, G_{i,L_i}$,  each of length $\tau$ sec. Then, each segment $G_{i,j}$ for $j\in\left\{ 1,2,\ldots,L_{i}\right\} $ is  partitioned  into $k_i$ fixed-size chunks  and then
 encoded  using an $(n_i, k_i)$ Maximum Distance Separable (MDS) erasure code to generate $n_i$  distinct chunks for each segment $G_{i,j}$. These coded chunks are denoted as $C_{i,j}^{(1)}, \cdots, C_{i,j}^{(n_i)}$. The encoding setup is illustrated in Figure~\ref{fig:videoEncoding}.

The encoded chunks are stored on the disks of $n_i$ distinct storage nodes. These storage nodes are represented by a set $\mathcal{S}_{i}$, such that 
$\mathcal{S}_{i}\subseteq\mathcal{M}$ and $n_{i}=\left|\mathcal{S}_{i}\right|$. Each server $z\in \mathcal{S}_{i}$ stores all the chunks $C_{i,j}^{(g_z)}$ for all $j$ and for some $g_z\in \{1, \cdots, n_i\}$. In other words, each of the $n_i$ storage nodes  stores one of the coded chunks for the entire duration of the video. The placement on the servers is illustrated in Figure \ref{fig:plcOnServ}, where  server $1$ is shown to store first coded chunks of file $i$, third coded chunks of file $u$ and first coded chunks for file $v$.  

The use of $\left(n_{i},k_{i}\right)$ of MDS erasure code introduces a redundancy factor of $n_{i}/k_{i}$ which allows the video to be reconstructed from the video chunks from any subset  of $k_{i}$-out-of-$n_{i}$ servers.  We note that the erasure-code can also help in recovery of the content $i$  as long as $k_i$ of the servers containing file $i$ are available \citep{dimakis2010network}.  Note that replication along $n$ servers is equivalent to choosing $(n,1)$ erasure code.  Hence, when a video $i$ is requested, the request goes to a set $\mathcal{A}_{i}$ of the storage nodes, where  $\mathcal{A}_{i}\subseteq\mathcal{S}_{i}$ and $k_{i}=\left|\mathcal{A}_{i}\right|$. From each server $z \in \mathcal{A}_{i}$, all chunks $C_{i,j}^{(g_z)}$ for all $j$ and the value of $g_z$ corresponding to that placed on server $z$ are requested. The request is illustrated in  Figure \ref{fig:plcOnServ}. In order to play a segment $q$ of video $i$, $C_{i,q}^{(g_z)}$ should have been downloaded from all $z\in \mathcal{A}_{i}$.  We assume that an edge router which is a combination of multiple users is requesting the files. Thus, the connections between the servers and the edge router is considered as the bottleneck. Since the service provider only has control over this part of the network and the last hop may not be under the control of the provider, the service provider can only guarantee the quality-of-service till the edge router. %

\begin{figure}
\centering\includegraphics[trim=0in 0in 0in .8in, clip, scale=0.35]{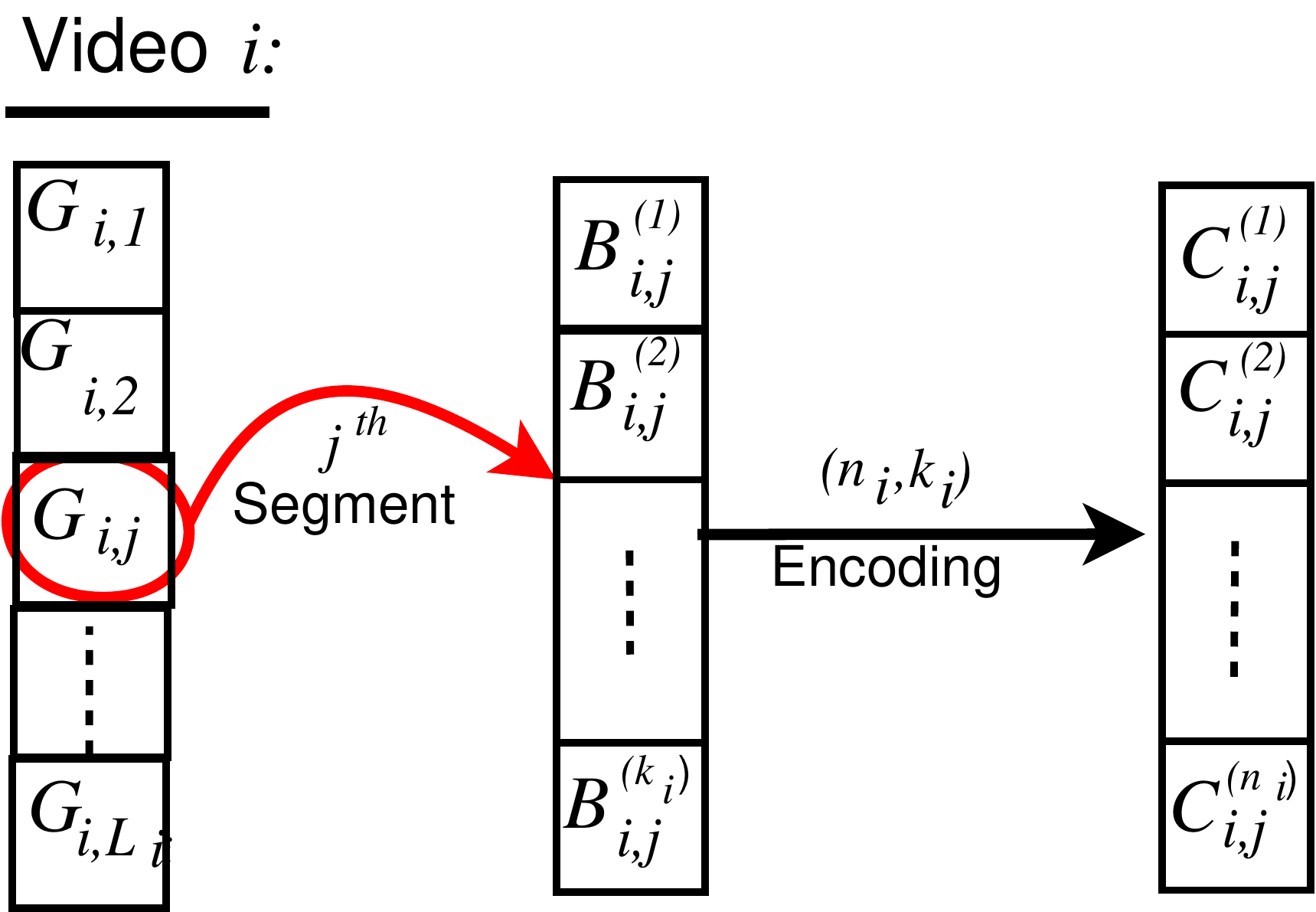}
\caption{A schematic illustrates video fragmentation and erasure-coding processes. Video $i$ is composed of $L_{i}$ segments. Each segments is partitioned into $k_{i}$ chunks and then encoded using an $(n_{i},k_{i})$ MDS code.\label{fig:videoEncoding}}
\end{figure}

\begin{figure}
\centering\includegraphics[scale=0.28,angle=-90]{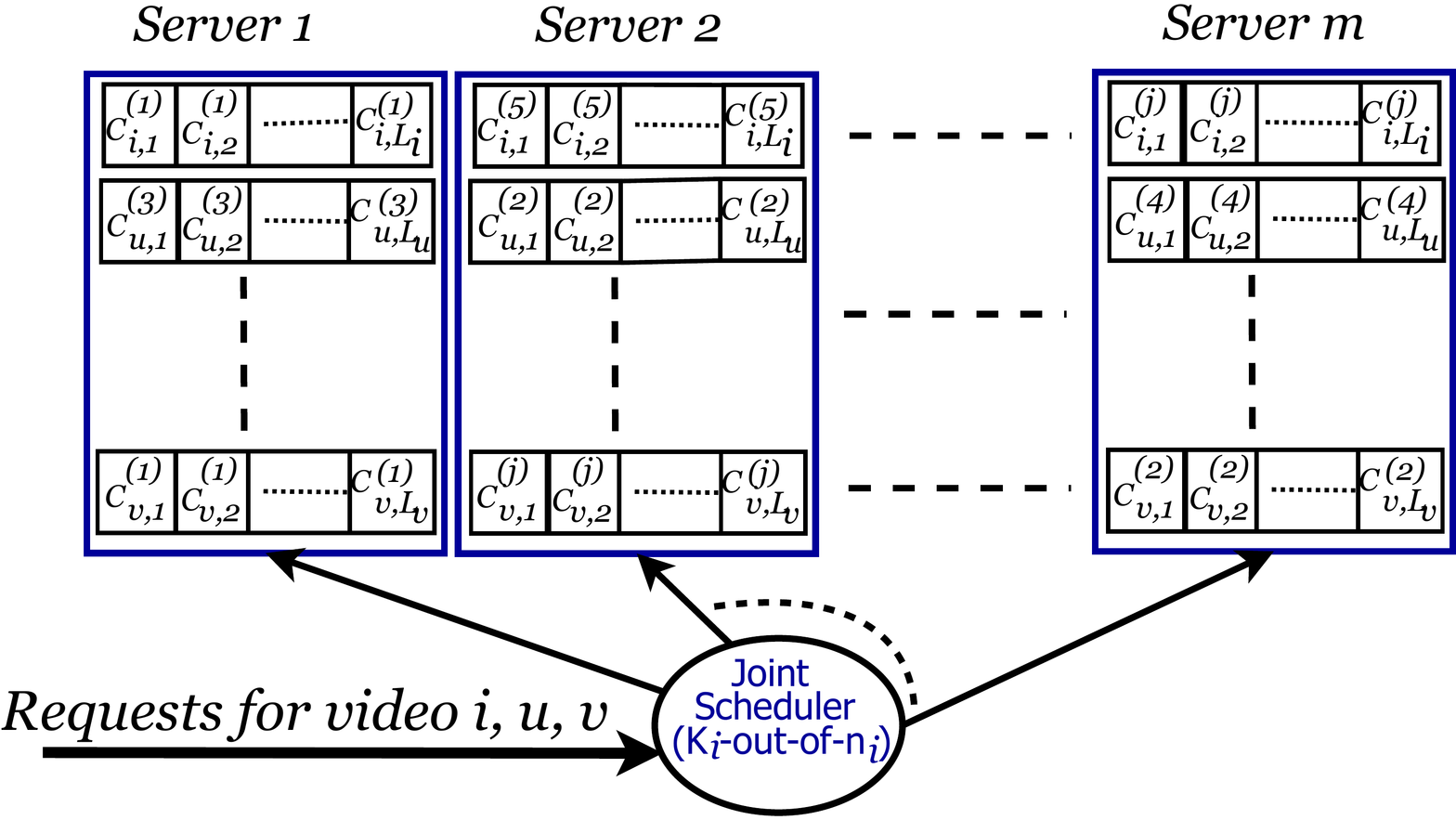}
\vspace{-.3in}
\caption{An Illustration of a distributed storage system
equipped with $m$ nodes and storing $3$ video files assuming $(n_{i},k_{i})$ erasure codes.\label{fig:plcOnServ}}
\end{figure}

\begin{figure}
\centering\includegraphics[trim=0.2in .5in 0in .3in, clip, width=.45\textwidth]{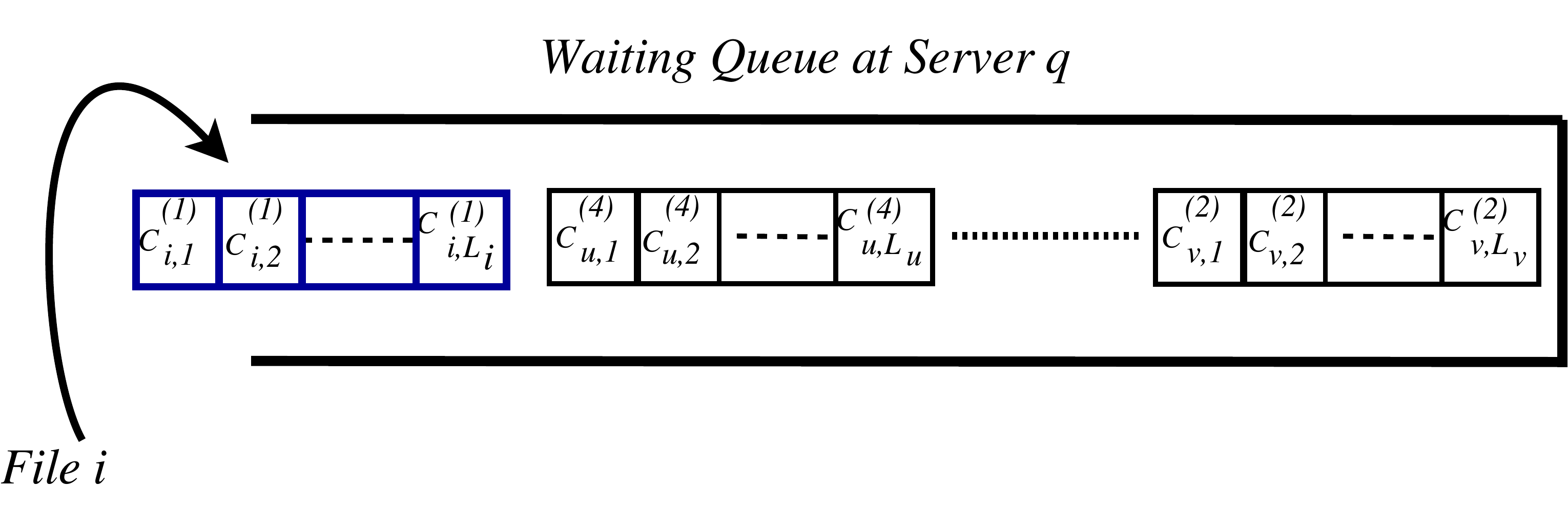}
\caption{An Example of the instantaneous queue status at
server $q$, where $q\in{1,2,...,m}$.\label{fig:sysModel}}
\end{figure}

We assume that the files at each server are served in order of the request in a first-in-first-out (FIFO) policy. Further, the different chunks are processed in order of the duration. This is depicted in  Figure \ref{fig:sysModel}, where for a server $q$, when a file $i$ is requested, all the chunks are placed in the queue where other video requests before this that have not yet been served are waiting. 

In order to schedule the requests for video file $i$ to the $k_i$ servers, the choice of $k_i$-out-of-$n_i$ servers is important. Finding the optimal choice of these servers to compute the latency expressions is an open problem to the best of our knowledge. Thus, this paper uses a policy, called Probabilistic Scheduling, given in Chapter \ref{sec:prob_sch}. %
This policy  allows choice of every possible subset of $k_i$ nodes with certain probability. Upon the arrival of a video file $i$,  we randomly dispatch the batch of $k_{i}$ chunk requests  to appropriate a set of nodes (denoted by set $\mathcal{A}_{i}$ of servers for file $i$) with predetermined probabilities
($P\left(\mathcal{A}_{i}\right)$ for set $\mathcal{A}_{i}$ and file $i$). Then, each node buffers requests in a local queue and processes in order and independently as explained before. From Chapter \ref{sec:prob_sch}, we note that probabilistic scheduling policy with feasible probabilities
$\left\{ P\left(\mathcal{A}_{i}\right):\,\forall_{i},\,\mathcal{A}_{i}\right\} $ exists
if and only if there exists conditional probabilities $\pi_{ij}\in\left[0,1\right]$
$\forall i,j$ satisfying

\[
\sum_{j=1}^{m}\pi_{ij}=k_{i}\,\,\,\,\forall i\,\,\,\,\,\,\,\,\,\,\,\,\,\,\,\,\,\,\mbox{and}\,\,\,\,\,\,\pi_{ij}=0\,\,\,\,\,\mbox{if\,\,\,\ensuremath{j\notin \mathcal{S}_{i}}}.
\]

In other words, selecting each node $j$ with probability $\pi_{ij}$ would yield a feasible choice of $\left\{ P\left(\mathcal{A}_{i}\right):\,\forall_{i},\,\mathcal{A}_{i}\right\} $. Thus, we consider the request probabilities $\pi_{ij}$ as the probability that the request for video file $i$ uses server $j$. While the probabilistic scheduling have been used to give bounds on latency of file download, this paper uses the scheduling to give bounds on the QoE for video streaming.

We note that it may not be ideal in practice  for a server to finish one video request before starting another since that increases delay for the future requests. However, this can be easily alleviated by considering that each server has multiple queues (streams) to the edge router which can all be considered as separate servers. These multiple streams can allow multiple parallel videos from the server.  The probabilistic scheduling can choose $k_i$ of the overall queues to access the content. This extension can be seen in \citep{al2018video}. 

We now describe a queuing model of the distributed storage system.
We assume that the arrival of client requests for each video $i$ form
an independent Poisson process with a known rate $\lambda_{i}$.  The arrival of file requests at node $j$ forms a Poisson Process
with rate $\varLambda_{j}=\sum_{i}\lambda_{i}\pi_{i,j}$ which is
the superposition of $r$ Poisson processes each with rate $\lambda_{i}\pi_{i,j}$.

We assume that the chunk service time for each coded chunk $C_{i,l}^{(g_j)}$ at server $j$,  $X_{j}$,  follows a shifted exponential distribution as has been demonstrated in realistic systems \citep{Yu_TON,CS14}. The service time distribution for the chunk service time at server $j$, $X_{j}$, is given by
the probability distribution function $f_{j}(x)$, which is
\begin{equation}
f_{j}(x)=\begin{cases}
\begin{array}{cc}
\alpha_{j}e^{-\alpha_{j}\left(x-\beta_{j}\right)}\,, & \,\,\,\,\,x\geq\beta_{j}\\
0\,, & \,\,\,\,\,\,x<\beta_{j}
\end{array}\end{cases}.
\end{equation}

We note that exponential distribution is a special case with $\beta_{j}=0$.  We note that the constant delays like the networking delay, and the decoding time can be easily factored into the shift of the shifted exponential distribution. Let $M_{j}(t)=\mathbb{E}\left[e^{tX_{j}}\right]$ be the moment generating function of $X_{j}$.  Then, $M_{j}(t)$ is given as

\begin{equation}
M_{j}(t)=\frac{\alpha_{j}}{\alpha_{j}-t}\,e^{\beta_{j}t}\,\,\,\,\,\,\,\,\,\,t<\alpha_{j} 
\label{M_j_t_1}
\end{equation}

We note that the arrival rates are given in terms of the video files, and the service rate above is provided in terms of the coded chunks at each server. The client plays the video segment after all the $k_i$ chunks for the segment have been downloaded and the previous segment has been played. We also assume that there is a start-up delay of $d_{s}$ (in seconds) for the video which is the duration in which the content can be buffered but not played.  This paper will characterize the stall duration and  stall duration tail probability for this setting.
\section{Modeling Download and Play Times}\label{sec:dtpt_vs}
In order to understand the stall duration, we need to see the download time of different coded chunks and the play time of the different segments of the video. %

\subsection{Download Times of the Chunks from each Server}
In this subsection, we will quantify the download time of chunk for video file $i$ from server $j$ which has  chunks $C_{i,q}^{(g_j)}$ for all $q = 1, \cdots L_i$. We consider download of $q^{\text{th}}$  chunk $C_{i,q}^{(g_j)}$. As seen in Figure \ref{fig:sysModel}, the download of $C_{i,q}^{(g_j)}$ consists of two components - the waiting time of all the video files in queue before file $i$ request and the service time of all chunks of video file $i$ up to the $q^{\text{th}}$ chunk. Let  $\ensuremath{W_{j}}$ be the random variable corresponding to the waiting time of all the video files in queue before file $i$ request and $Y_{j}^{(q)}$ be the (random) service time of coded chunk $q$ for file $i$ from server $j$. Then, the (random) download time for coded chunk $q\in \{1, \cdots, L_i\}$ for file $i$ at server $j\in \mathcal{A}_{i}$, $D_{i,j}^{(q)}$, is given as 
\begin{equation}
D_{i,j}^{(q)} = W_j + \sum_{v=1}^q Y_{j}^{(v)}. \label{dije}
\end{equation}

We will now find the distribution of $W_j$. We note that this is the waiting time for the video files whose arrival rate is given as $\varLambda_{j}=\sum_{i}\lambda_{i}\pi_{i,j}$. Since the arrival rate of video files is Poisson, the waiting time for the start of video download from a server $j$, $W_j$,  is given by an M/G/1 process.  In order to find the waiting time, we would need to find the service time statistics of the video files. Note that $f_{j}(x)$ gives the service time distribution of only a chunk and not of the video files. 

Video file $i$ consists of $L_{i}$ coded chunks at server $j$ ($j\in \mathcal{S}_{i}$). The total service time for video file $i$ at server $j$ if requested from server $j$, $ST_{i,j} $, is given as  
\begin{equation}
ST_{i,j} = \sum_{v=1}^{L_i} Y_{j}^{(v)}.
\end{equation}

The service time of the video files is given as 
\begin{equation}
R_{j} = \begin{cases}
ST_{i,j}  \quad \text{ with probability } \frac{\pi_{ij}\lambda_{i}}{\Lambda_j} \quad \forall i,
\end{cases}
\end{equation}
since the service time is $ST_{i,j} $ when file $i$ is requested from server $j$. Let $\overline{R}_{j}(s) = {\mathbb E}[e^{-sR_{j} }]$ be the Laplace-Stieltjes Transform of $R_{j}$. 

\begin{lemma}\label{ljlemma}
	The Laplace-Stieltjes Transform of  $R_{j}$, $\overline{R}_{j}(s)=\mathbb{E}\left[e^{-s\overline{R}_{j}}\right]$ is given as
	\begin{equation}
\overline{R}_{j}(s)  = \sum_{i=1}^r \frac{\pi_{ij}\lambda_i}{\Lambda_j}	\left(\frac{\alpha_{j}e^{-\beta_{j}s}}{\alpha_{j}+s}\right)^{L_{i}}\label{eq:servTimeofFile}
	\end{equation}
	\end{lemma}

\begin{proof}
\begin{align}
	\overline{R}_{j}(s) & =\sum_{i=1}^r \frac{\pi_{ij}\lambda_{i}}{\Lambda_{j}}\mathbb{E}\left[e^{-s\left(ST_{i,j}\right)}\right]\nonumber \\
	& \overset{}{=}\sum_{i=1}^r \frac{\pi_{ij}\lambda_{i}}{\Lambda_{j}}\mathbb{E}\left[e^{-s\left(\sum_{\nu=1}^{L_{i}}Y_{j}^{(\nu)}\right)}\right]\nonumber \\
	& =\sum_{i=1}^r \frac{\pi_{ij}\lambda_{i}}{\Lambda_{j}}\left(\mathbb{E}\left[e^{-s\left(Y_{j}^{(1)}\right)}\right]\right)^{L_{i}}\nonumber \\
	& =\sum_{i=1}^r \frac{\pi_{ij}\lambda_{i}}{\Lambda_{j}}\left(\frac{\alpha_{j}e^{-\beta_{j}s}}{\alpha_{j}+s}\right)^{L_{i}}
\end{align}

\end{proof}

\begin{corollary}
	The moment generating function for the service time of video files when requested from server $j$, $B_{j}(t)$, is given by
	\begin{equation}
B_{j}(t)  = \sum_{i=1}^r \frac{\pi_{ij}\lambda_i}{\Lambda_j}	\left(\frac{\alpha_{j}e^{\beta_{j}t}}{\alpha_{j}-t}\right)^{L_{i}}\label{eq:servTimeofFileB_j_i}
	\end{equation}
	for any $t>0$, and $t< \alpha_j$.
	\end{corollary}	
\begin{proof}
This corollary follows from (\ref{eq:servTimeofFile})
 by setting $t=-s$.
\end{proof}

The server utilization for the video files at server $j$ is given as $\rho_{j}=\varLambda_{j}\mathbb{E}\left[R_j\right]$. Since $\mathbb{E}\left[R_j\right] = B_j'(0)$, using Lemma \ref{eq:servTimeofFile}, we have

\begin{equation}
\rho_{j}=\sum_{i}\pi_{ij}\lambda_{i}L_{i}\left(\beta_{j}+\frac{1}{\alpha_{j}}\right)\label{eq:rho_j}.
\end{equation}

Having characterized the service time distribution of the video files via a Laplace-Stieltjes Transform $\overline{R}_{j}(s) $, the Laplace-Stieltjes Transform of the waiting time $W_{j}$ can be characterized using Pollaczek-Khinchine formula for M/G/1 queues \citep{zwart2000sojourn1}, since the request pattern is Poisson and the service time is general distributed. Thus, the Laplace-Stieltjes Transform of the waiting time $W_{j}$ is given as 

\begin{equation}
\mathbb{E}\left[e^{-sW_{j}}\right]=\frac{\left(1-\rho_{j}\right)s}{s-\Lambda_{j}\left(1-\overline{R}_{j}(s)\right)}\label{eq:E_W_j_laplace}
\end{equation}

Having characterized the Laplace-Stieltjes Transform of the waiting time $W_{j}$ and knowing the distribution of  $Y_{j}^{(v)}$, the Laplace-Stieltjes Transform of the download time $D_{i,j}^{(q)}$ is given as

\begin{equation}
{\mathbb E}[e^{-sD_{i,j}^{(q)}}] = \frac{\left(1-\rho_{j}\right)s}{s-\Lambda_{j}\left(1-\overline{R}_{j}(s)\right)}\left( \frac{\alpha_{j}}{\alpha_{j}+s}\,e^{-\beta_{j}s}\right)^q.\label{LapOfE_D_ij}
\end{equation}

We note that the expression above holds only in the range of $s$ when  $s-\Lambda_{j}\left(1-\overline{R}_{j}(s)\right)>0$ and $\alpha_{j}+s>0$. Further, the server utilization $\rho_j $ must be less than $1$. The overall download time of all the chunks for the segment $G_{i,q}$ at the client,  $D_{i}^{(q)}$, is given by 

\begin{equation}
D_{i}^{(q)} = \max_{j\in \mathcal{A}_{i}}  D_{i,j}^{(q)}. \label{deq}
\end{equation}

\subsection{Play Time of Each Video Segment}
 Let $T_{i}^{\left(q\right)}$ be the  time at which the segment $G_{i,q}$ is played (started) at the client. The startup delay of the video is $d_s$. Then, the first segment can be played at the maximum of the time the first segment can be downloaded and the startup delay. Thus, 

\begin{eqnarray}
T_{i}^{(1)} & = & \mbox{max }\left(d_{s},\,D_{i}^{(1)}\right).
\label{eq:T_i_qi1}
\end{eqnarray}

For $1<q\le L_i$, the  play time of  segment $q$ of file $i$ is given by the maximum of the time it takes to download the segment and the time at which the previous segment is played plus the time to play a segment ($\tau$ seconds). Thus, the play time of segment $q$ of file $i$, $T_{i}^{(q)}$  can be expressed as 
\begin{eqnarray}
T_{i}^{(q)} & = & \mbox{max }\left(T_{i}^{(q-1)}+\tau,\,D_{i}^{(q)}\right).
\label{eq:T_i_qi}
\end{eqnarray}

Equation \eqref{eq:T_i_qi} gives a recursive equation, which can yield

\begin{eqnarray}
T_{i}^{(L_i)}
 & = &  \mbox{max }\left(T_{i}^{(L_i-1)}+\tau,\,D_{i}^{(L_i)}\right)\nonumber\\
& = &  \mbox{max }\left(T_{i}^{(L_i-2)}+2\tau,\, D_{i}^{(L_i-1)}+\tau,\,D_{i}^{(L_i)}\right)\nonumber\\
& = & \! \mbox{max} \left(d_s+(L_{i}-1)\tau,\right.\nonumber\\
&&\left. \max_{z=2}^{L_i+1}D_{i}^{(z-1)} + (L_i-z+1)\tau \right)\
\label{eq:T_i_qi2}
\end{eqnarray}

Since $D_{i}^{(q)} = \max_{j\in \mathcal{A}_{i}}  D_{i,j}^{(q)}$ from \eqref{deq}, $T_{i}^{(L_{i})}$ can  be written as

\begin{eqnarray}
T_{i}^{(L_{i})}= \max_{z=1}^{L_i+1}\max_{j\in \mathcal{A}_i}\left(p_{i,j,z}\right)\label{eq:T_i_L_i}, \label{eq:pjstart}
\end{eqnarray}where
\begin{eqnarray}
p_{i,j,z}=\begin{cases}
d_{s}+\left(L_{i}-1\right)\tau &,\,\, z=1\\
\\
D_{i,j}^{(z-1)} + (L_i-z+1)\tau \ &,\,\, 2\leq z\leq (L_{i}+1)
\end{cases}\label{eq:pjz}
\end{eqnarray}

We next give the moment generating function of $p_{i,j,z}$ that will be used in the calculations of the QoE metrics in the next sections. Hence, we define the following lemma. 

\begin{lemma}\label{lemma_pijz}
	The moment generating function for $p_{i,j,z}$, is given as
	\begin{equation}
\mathbb{E}\left[e^{tp_{i,j,z}}\right]=\begin{cases}
e^{t\left(d_{s}+\left(L_{i}-1\right)\tau\right)} & ,\,z=1\\
e^{t\left(L_{i}+1-z\right)\tau}Z_{i,j}^{(z-1)}\left(t\right) & ,2\leq z\leq L_{i}+1
\end{cases}
\label{eq:momntPjz}
	\end{equation}
where  

\begin{equation}
Z_{i,j}^{(\ell)}\left(t\right) = {\mathbb E}[e^{tD_{i,j}^{(\ell)}}]  = \frac{\left(1-\rho_{j}\right)t\left(M_{j}(t)\right)^{\ell}}{t-\Lambda_{j}\left(B_{j}(t)-1\right)}\label{eq:M_D_ij} 
\end{equation}
	\end{lemma}

\begin{proof}
This follows by substituting
$t=-s$ in (\ref{LapOfE_D_ij}) and $B_{j}(t)$ is given by  (\ref{eq:servTimeofFileB_j_i}) and $M_j(t)$ is given by (\ref{M_j_t_1}). This expressions holds when $t-\Lambda_{j}\left(B_{j}(t)-1\right)>0$ and $t<0  \,\forall j$,  since the moment generating function does not exist if the above does not hold.
\end{proof}

Ideally, the last segment should be completed by time $d_s + L_i \tau$. The difference between $T_{i}^{(L_i)}$ and $d_s + (L_i-1) \tau$ gives the stall duration. Note that the stalls may occur before any segment. This difference will give the sum of durations of all the stall  periods before any segment. Thus, the stall duration for the request of file 
$\delta^{(i)}$ is given as
\begin{equation}
\Gamma^{(i)} = T_{i}^{(L_i)} - d_s - (L_i-1) \tau. \label{eq:base}
\end{equation}
In the next two sections, we will use this stall time to determine the bounds on the mean stall duration and the stall duration tail probability.

\section{Characterization of Mean Stall Duration}\label{sec:msd}
In this section, we will provide a bound for the first QoE metric, which is the mean stall duration for a file $i$. We will find the bound through probabilistic scheduling and since probabilistic scheduling is one feasible  strategy, the obtained bound is an upper bound to the optimal strategy. 

Using \eqref{eq:base}, the expected stall time for file $i$ is given as follows
\begin{eqnarray}
\mathbb{E}\left[\Gamma^{(i)}\right] & = & \mathbb{E}\left[T_{i}^{(L_{i})}-d_{s}-\left(L_{i}-1\right)\tau\right]\nonumber \\
\nonumber \\
& = & \mathbb{E}\left[T_{i}^{(L_{i})}\right]-d_{s}-\left(L_{i}-1\right)\tau\label{eq:E_T_s_2}
\end{eqnarray}

 An exact evaluation for the play time of segment $L_{i}$ is hard due to the dependencies between  $p_{jz}$ random variables for different values of $j$ and $z$, where $z\in{(1,2,...,L_{i}+1)}$ and $j\in \mathcal{A}_{i}$. Hence, we derive an upper-bound on the playtime of the segment $L_{i}$ as follows. Using Jensen's inequality \citep{kuczma2009introduction}, we have for $t_i>0$, 
 
 \begin{equation}
e^{t_{i}\mathbb{E}\left[T_{i}^{\left(L_{i}\right)}\right]} \leq \mathbb{E}\left[e^{t_{i}T_{i}^{\left(L_{i}\right)}}\right]. \label{eq:jensen}
 \end{equation}
 
 Thus, finding an upper bound on the moment generating function for $T_i^{(L_i)}$ can lead to an upper bound on the mean stall duration. Thus, we will now bound the moment generating function for $T_i^{(L_i)}$.

\begin{eqnarray}
  \mathbb{E}\left[e^{t_{i}T_{i}^{\left(L_{i}\right)}}\right] & \overset{(a)}{=} & \mathbb{E}\left[\underset{z}{\mbox{max}}  \,\underset{j\in\mathcal{A}_{i}}{\mbox{max}}\,e^{t_{i}p_{ijz}}\right]\nonumber\\
 & = & \mathbb{E}_{\mathcal{A}_{i}}\left[\mathbb{E}\left[\underset{z}{\mbox{max}}\,\underset{j\in{\mathcal{A}_{i}}}{\mbox{max}}\,e^{t_{i}p_{ijz}}|\,\mathcal{A}_{i}\right]\right]\nonumber
\\
 &\overset{(b)}{\leq} & \mathbb{E}_{\mathcal{A}_{i}}\left[\sum_{j\in \mathcal{A}_{i}}\mathbb{E}\left[\underset{z}{\mbox{max}}\,e^{t_{i}p_{ijz}}\right]\right]\nonumber
\\
 & = & \mathbb{E}_{\mathcal{A}_{i}}\left[\sum_{j}F_{ij}\mathbf{1}_{\left\{ j\in\mathcal{A}_{i}\right\} }\right]\nonumber
\\
 & = & \sum_{j}F_{ij}\,\mathbb{E}_{\mathcal{A}_{i}}\left[\mathbf{1}_{\left\{ j\in\mathcal{A}_{i}\right\} }\right]\nonumber
\\
 & = & \sum_{j}F_{ij}\,\mathbb{P}\left(j\in\mathcal{A}_{i}\right)\nonumber
\\
 & \overset{(c)}{=} & \sum_{j}F_{ij}\pi_{ij}
\label{eq:mgf_bound}
\end{eqnarray}
where (a) follows from \eqref{eq:pjstart}, (b) follows by upper bounding $\max_{j\in \mathcal{A}_{i}} $ by $\sum_{j\in \mathcal{A}_{i}} $, (c) follows by probabilistic scheduling where $\mathbb{P}\left(j\in\mathcal{A}_{i}\right) = \pi_{ij}$,   and  $F_{ij}=\mathbb{E}\left[\underset{z}{\mbox{max}}  \, e^{t_{i}p_{ijz}}\right]$. We note that the only inequality here is for replacing the maximum by the sum. Since this term will be inside the logarithm for the mean stall latency, the gap between the term and its bound becomes additive rather than multiplicative.

Substituting \eqref{eq:mgf_bound} in \eqref{eq:jensen}, we have 

\begin{equation}
\mathbb{E}\left[T_{i}^{(L_{i})}\right]\leq\frac{1}{t_{i}}\text{log}\left(\sum_{j=1}^{m}\pi_{ij}F_{ij}\right).\label{eq:ET_i}
\end{equation}

 Let $H_{ij}=\sum_{\ell=1}^{L_{i}}e^{-t_{i}\left(d_{s}+\left(\ell-1\right)\tau\right)}Z_{{i,j}}^{(\ell)}(t_{i})$, where $Z_{{i,j}}^{(\ell)}(t)$ is defined in equation \eqref{eq:M_D_ij}. We  note that $H_{ij}$ can be simplified using the geometric series formula as follows.
 
 \begin{lemma}\label{hijlem}
 	\begin{equation}
 	H_{ij} =  \frac{e^{-t_{i}\left(d_{s}-\tau\right)}\left(1-\rho_{j}\right)t_{i}\,\widetilde{M}_{j}(t_{i})}{t_{i}-\Lambda_{j}\left(B_{j}(t_{i})-1\right)}\frac{1-\left(\widetilde{M}_{j}(t_{i})\right)^{L_{i}}}{\left(1-\widetilde{M}_{j}(t_{i})\right)},
 	\label{eq:H}
 	\end{equation}
 	where $\widetilde{M}_{j}(t_{i})=M_{j}(t_{\text{i}})e^{-t_{i}\tau}$, $M_j(t_i)$ is given in (\ref{M_j_t_1}), and $B_j(t_i)$ is given in (\ref{eq:servTimeofFileB_j_i}). 
 \end{lemma}
 \begin{proof}
\begin{eqnarray}
	H_{ij} 
	& = & \sum_{\ell=1}^{L_{i}}\left(\frac{e^{-t_{i}\left(d_{s}+\left(\ell-1\right)\tau\right)}\left(1-\rho_{j}\right)t_{i}}{t_{i}-\Lambda_{j}\left(B_{j}(t_{i})-1\right)}\left(\frac{\alpha_{j}e^{t_{i}\beta_{j}}}{\alpha_{j}-t_{i}}\right)^{\ell}\right)\nonumber\\
	& = & \frac{e^{-t_{i}d_{s}}\left(1-\rho_{j}\right)t_{i}}{t_{i}-\Lambda_{j}\left(B_{j}(t_{i})-1\right)}\sum_{\ell=1}^{L_{i}}\left(e^{-t_{i}\left(\ell-1\right)\tau}\left(\frac{\alpha_{j}e^{t_{i}\beta_{j}}}{\alpha_{j}-t_{i}}\right)^{\ell}\right)\nonumber\\
	& = & \frac{e^{-t_{i}\left(d_{s}-\tau\right)}\left(1-\rho_{j}\right)t_{i}}{t_{i}-\Lambda_{j}\left(B_{j}(t_{i})-1\right)}\sum_{\ell=1}^{L_{i}}\left(e^{-t_{i}\tau}\frac{\alpha_{j}e^{t_{i}\beta_{j}}}{\alpha_{j}-t_{i}}\right)^{\ell}\nonumber\\
	& = & \frac{e^{-t_{i}\left(d_{s}-\tau\right)}\left(1-\rho_{j}\right)t_{i}}{t_{i}-\Lambda_{j}\left(B_{j}(t_{i})-1\right)}\sum_{\ell=1}^{L_{i}}\left(\frac{\alpha_{j}e^{t_{i}\beta_{j}-t_{i}\tau}}{\alpha_{j}-t_{i}}\right)^{\ell}\nonumber\\
	& = & \frac{e^{-t_{i}\left(d_{s}-\tau\right)}\left(1-\rho_{j}\right)t_{i}}{t_{i}-\Lambda_{j}\left(B_{j}(t_{i})-1\right)}\times\nonumber\\
	&  & \left(M_{j}(t_{i})e^{-t_{i}\tau}\frac{1-\left(M_{j}(t_{i})\right)^{Li}e^{-t_{i}L_{i}\tau}}{1-M_{j}(t_{i})e^{-t_{i}\tau}}\right)\nonumber\\
	& = & \frac{e^{-t_{i}\left(d_{s}-\tau\right)}\left(1-\rho_{j}\right)t_{i}\widetilde{M}_{j}(t_{i})}{t_{i}-\Lambda_{j}\left(B_{j}(t_{i})-1\right)}\frac{1-\left(\widetilde{M}_{j}(t_{i})\right)^{L_{i}}}{\left(1-\widetilde{M}_{j}(t_{i})\right)}
\end{eqnarray}

 \end{proof}

 Substituting \eqref{eq:ET_i} in \eqref{eq:E_T_s_2} and some manipulations, the  mean stall duration is bounded as follows.

\begin{theorem}
	The mean stall duration time for file $i$ is bounded by 
	
	\begin{equation}
	\mathbb{E}\left[\Gamma^{(i)}\right]\leq\frac{1}{t_{i}}\text{log}\left(\sum_{j=1}^{m}\pi_{ij}\left(1+H_{ij}\right)\right)\label{eq:T_s_main_stall}
	\end{equation}
	for any $t_{i}>0$, $\rho_{j}=\sum_{i}\pi_{ij}\lambda_{i}L_{i}\left(\beta_{j}+\frac{1}{\alpha_{j}}\right)$, $\rho_{j}<1,\,\text{and }$ \\
	$\,\sum_{f=1}^r\pi_{fj}\lambda_{f}\left(\frac{\alpha_{j}e^{-\beta_{j}t_{i}}}{\alpha_{j}-t_{i}}\right)^{L_{f}}-\left(\Lambda_{j}+t_{i}\right)<0,\,\forall j$.\label{meanthm}
\end{theorem}
\begin{proof}
We first find an upper bound on $F_{ij}$   as follows.
\begin{align}
	F_{ij} & =\mathbb{E}\left[\underset{z}{\text{max}}\, e^{t_{i}pijz}\right]\nonumber \\
	& \overset{(d)}{\leq}\sum_{z}\mathbb{E}\left[e^{t_{i}pijz}\right]\nonumber \\
	& \overset{(e)}{=}e^{t_{i}(d_{s}+(L_{i}-1)\tau)}+ \nonumber  \\
	& \sum_{z=2}^{L_{i}+1}\frac{e^{t_{i}\left(L_{i}-z+1\right)\tau}\left(1-\rho_{j}\right)t_{i}}{t_{i}-\Lambda_{j}\left(B_{j}(t_{i})-1\right)}\left(\frac{\alpha_{j}e^{t_{i}\beta_{j}}}{\alpha_{j}-t_{i}}\right)^{z-1}\nonumber \\
	& \overset{(f)}{=}e^{t_{i}(d_{s}+(L_{i}-1)\tau)}+ \nonumber  \\
	& \sum_{\ell=1}^{L_{i}}\frac{e^{t_{i}\left(L_{i}-\ell\right)\tau}\left(1-\rho_{j}\right)t_{i}}{t_{i}-\Lambda_{j}\left(B_{j}(t_{i})-1\right)}\left(\frac{\alpha_{j}e^{t_{i}\beta_{j}}}{\alpha_{j}-t_{i}}\right)^{\ell}\label{eq:F_ij}
\end{align}
where (d) follows by bounding the maximum by the sum, (e) follows from (\ref{eq:momntPjz}), and (f) follows by substituting $\ell=z-1$.

Further, substituting the bounds \eqref{eq:F_ij} and \eqref{eq:ET_i} in \eqref{eq:E_T_s_2}, the  mean stall duration is bounded as follows.  

\begin{eqnarray}
	&&\mathbb{E}\left[\Gamma^{(i)}\right] \nonumber\\
	&\leq& \frac{1}{t_{i}}\text{log}\left(\sum_{j=1}^{m}\pi_{ij}\left(e^{t_{i}(d_{s}+(L_{i}-1)\tau)}\right.\right.\nonumber \\
	&& \left.\left.+\sum_{\ell=1}^{L_{i}}e^{t_{i}\left(L_{i}-\ell\right)\tau}Z_{{i,j}}^{(\ell)}(t_{i})\right)\right)-\left(d_{s}+\left(L_{i}-1\right)\tau\right)\nonumber \\
	&=& \frac{1}{t_{i}}\text{log}\left(\sum_{j=1}^{m}\pi_{ij}\left(e^{t_{i}(d_{s}+(L_{i}-1)\tau)}\right.\right.\nonumber \\
	&& \left.\left.+\sum_{\ell=1}^{L_{i}}e^{t_{i}\left(L_{i}-\ell\right)\tau}Z_{{i,j}}^{(\ell)}(t_{i})\right)\right)-\frac{1}{t_{i}}\text{log}\left(e^{t_{i}\left(d_{s}+\left(L_{i}-1\right)\tau\right)}\right)\nonumber \\
	&=& \frac{1}{t_{i}}\text{log}\left(\sum_{j=1}^{m}\pi_{ij}\left(1+\sum_{\ell=1}^{L_{i}}e^{-t_{i}\left(d_{s}+\left(\ell-1\right)\tau\right)}Z_{{i,j}}^{(\ell)}(t_{i})\right)\right)\label{eq:ET_s_i_ap}
\end{eqnarray}
\end{proof}

Note that Theorem \ref{meanthm} above holds only in the range of $t_i$ when  $t_i-\Lambda_{j}\left(B_{j}(t_i)-1\right)>0$ which reduces to 
$\,\sum_{f=1}^r\pi_{fj}\lambda_{f}\left(\frac{\alpha_{j}e^{-\beta_{j}t_{i}}}{\alpha_{j}-t_{i}}\right)^{L_{f}}-\left(\Lambda_{j}+t_{i}\right)<0,\,\forall i, j$,
and $\alpha_{j}-t_i>0$. Further, the server utilization $\rho_j $ must be less than $1$ for stability of the system.

We note that for the scenario, where the files are downloaded rather than streamed, a metric of interest is the mean download time. This is a special case of our approach when the number of segments of each video is one, or $L_i=1$. Thus, the mean download time of the file follows as a special case of Theorem \ref{meanthm}. This special case was discussed in detail in Section \ref{mean_ps}.

\section{Characterization of Tail Stall Duration}\label{sec:tsd}
The stall duration tail probability of a file $i$ is defined as the probability that the stall duration tail $\Gamma^{(i)}$ is greater than (or equal) to $x$. Since evaluating  $\text{\text{Pr}}\left(\Gamma^{(i)}\geq x\right)$ in closed-form is hard \citep{MG1:12,Joshi:13,MDS_queue,Xiang:2014:Sigmetrics:2014,Yu_TON,CS14}, we derive an upper bound on the stall duration tail probability considering Probabilistic Scheduling as follows. 

\begin{align}
\text{\text{Pr}}\left(\Gamma^{(i)}\geq x\right) & \overset{(a)}{=}\text{\text{Pr}}\left(T_{i}^{(L_{i})}\geq x+d_{s}+\left(L_{i}-1\right)\tau\right)\nonumber \\
 & =\text{\text{Pr}}\left(T_{i}^{(L_{i})}\geq\overline{x}\right)\label{eq:Pr_T_s,i}
\end{align}
where $(a)$ follows from (\ref{eq:E_T_s_2}) and $\overline{x}=x+d_{s}+\left(L_{i}-1\right)\tau$. Then,

\begin{eqnarray}
\text{\text{Pr}}\left(T_{i}^{(L_{i})}\geq\overline{x}\right) & \overset{(b)}{=} & \text{\text{Pr}}\left(\underset{z}{\text{max}\,\,}\text{\ensuremath{\underset{j\in\mathcal{A}_{i}}{\text{max}}p_{ijz}}}\geq\overline{x}\right)\nonumber\\
 & = & \mathbb{E}_{\mathcal{A}_{i},p_{ijz}}\left[\text{\ensuremath{\boldsymbol{1}_{\left(\underset{z}{\text{max}}\,\,\text{\ensuremath{\underset{j\in\mathcal{A}_{i}}{\text{max}}p_{ijz}}}\geq\overline{x}\right)}}}\right]
\end{eqnarray}
\begin{eqnarray}
 & \overset{(c)}{=} & \mathbb{E}_{\mathcal{A}_{i},p_{ijz}}\left[\text{\ensuremath{\underset{j\in\mathcal{A}_{i}}{\text{max}}\,\,\boldsymbol{1}_{\left(\underset{z}{\text{max}}\,p_{ijz}\geq\overline{x}\right)}}}\right]\nonumber\\
 & \overset{(d)}{\leq} & \mathbb{E}_{\mathcal{A}_{i},p_{ijz}}\sum_{j\in\mathcal{A}_{i}}\,\boldsymbol{1}_{\left(\underset{z}{\text{max}}p_{ijz}\geq\overline{x}\right)}\nonumber\\
 & \overset{(e)}{=} & \sum_{j}\pi_{ij}\mathbb{\mathbb{E}}_{pijz}\left[\,\boldsymbol{1}_{\left(\underset{z}{\text{max}}p_{ijz}\geq\overline{x}\right)}\right]\nonumber\\
 & = & \sum_{j}\pi_{ij}\mathbb{P}\left(\underset{z}{\text{max}\,\,}p_{ijz}\geq\overline{x}\right)
\label{eq:Pr_T_i_L_i_x_bar}
\end{eqnarray}
where $(b)$ follows from (\ref{eq:T_i_L_i}), (c) follows as both max over $z$ and max over $\mathcal{A}_{j}$ are discrete indicies (quantities) and do not depend on other so they can be exchanged, (d) follows by replacing the max by $\sum_{\mathcal{A}_{i}}$, (e) follows from probabilistic scheduling. Using Markov Lemma, we get 
\begin{equation}
\mathbb{P}\left(\underset{z}{\text{max}\,\,}p_{ijz}\geq\overline{x}\right)\leq\frac{\mathbb{E}\left[e^{t_{i}\left(\underset{z}{\text{max}\,\,}p_{ijz}\right)}\right]}{e^{t_{i}\overline{x}}}\label{eq:Mark_lem}
\end{equation}
We further simplify to get 
\begin{align}
\mathbb{P}\left(\underset{z}{\text{max}\,\,}p_{ijz}\geq\overline{x}\right) & \leq\frac{\mathbb{E}\left[e^{t_{i}\left(\underset{z}{\text{max}\,\,}p_{ijz}\right)}\right]}{e^{t_{i}\overline{x}}}\nonumber \\
 & =\frac{\mathbb{E}\left[\underset{z}{\text{max}\,\,}e^{t_{i}p_{ijz}}\right]}{e^{t_{i}\overline{x}}}\nonumber \\
 & \overset{\left(f\right)}{=}\frac{F_{ij}}{e^{t_{i}\overline{x}}}\label{eq:max_z_pjz}
\end{align}
where (f) follows from (\ref{eq:F_ij}). Substituting (\ref{eq:max_z_pjz}) in (\ref{eq:Pr_T_i_L_i_x_bar}), we get the stall duration tail  probability as described in the following theorem.

\begin{theorem}
The stall distribution tail probability  for video file $i$ is bounded by 

\begin{equation}
\sum_{j}\frac{\pi_{ij}}{e^{t_{i}x}}\left(1+e^{-t_{i}\left(d_{s}+(L_{i}-1)\tau\right)}\,H_{ij}\right)\
\label{eq:T_stall}
\end{equation}
for any $t_{i}>0$, $\rho_{j}=\sum_{i}\pi_{ij}\lambda_{i}L_{i}\left(\beta_{j}+\frac{1}{\alpha_{j}}\right)$, $\rho_{j}\leq1,$ \\
$\,\sum_{f=1}^r\pi_{fj}\lambda_{f}\left(\frac{\alpha_{j}e^{-\beta_{j}t_{i}}}{\alpha_{j}-t_{i}}\right)^{L_{f}}-\left(\Lambda_{j}+t_{i}\right)<0,\,\forall i,j$, \text{and} $H_{ij}$ is given by (\ref{eq:H}).\label{tailthm}
\end{theorem}
\begin{proof}
	Substituting (\ref{eq:max_z_pjz}) in (\ref{eq:Pr_T_i_L_i_x_bar}), we get 
	
	\begin{eqnarray}
	&&\text{\text{Pr}}\left(T_{i}^{(L_{i})}\geq\overline{x}\right) \nonumber\\
	& \leq & \sum_{j}\pi_{ij}\mathbb{P}\left(\underset{z}{\text{max}\,\,}p_{ijz}\geq\overline{x}\right)\nonumber\\
	& \leq & \sum_{j}\pi_{ij}\frac{F_{ij}}{e^{t_{i}\overline{x}}}\nonumber\\
	& \overset{(g)}{\leq} & \sum_{j}\frac{\pi_{ij}}{e^{t_{i}\overline{x}}}\left(e^{t_{i}(d_{s}+(L_{i}-1)\tau)}+H_{ij}\right)\nonumber\\
	& = & \sum_{j}\frac{\pi_{ij}}{e^{t_{i}\left(x+d_{s}+(L_{i}-1)\tau\right)}}\left(e^{t_{i}(d_{s}+(L_{i}-1)\tau)}+H_{ij}\right)\nonumber\\
	& = & \sum_{j}\frac{\pi_{ij}}{e^{t_{i}x}}\left(1+e^{-t_{i}\left(d_{s}+(L_{i}-1)\tau\right)}\,H_{ij}\right)
	\label{eq:Pr_T_i_L_i_x_bar_final}
	\end{eqnarray}
	where (g) follows from (\ref{eq:F_ij}) and $H_{ij}$ is given by (\ref{eq:H}). 
\end{proof}

We note that for the scenario, where the files are downloaded rather than streamed, a metric of interest is the latency tail probability which is the probability that the file download latency is greater than $x$.  This is a special case of our approach when the number of segments of each video is one, or $L_i=1$. Thus, the latency tail probability of the file follows as a special case of Theorem \ref{tailthm}. In this special case, the result reduces to that in \citep{Jingxian}.
\section{Simulations}\label{sec:vsr_sims}
Let $\boldsymbol{\pi} = (\pi_{ij} \forall i=1, \cdots, r \text{ and } j=1, \cdots, m)$, $\boldsymbol{\mathcal{S}}=\left(\mathcal{S}_{1},\mathcal{S}_{2},\ldots,\mathcal{S}_{r}\right)$, and $\boldsymbol{t}=\left(\widetilde{t}_{1},\widetilde{t}_{2},\ldots,\widetilde{t}_{r}; \overline{t}_{1},\overline{t}_{2},\ldots,\overline{t}_{r}\right)$. Note that the values of $t_i$'s used for mean stall duration and the stall duration tail probability can be different and the parameters $\widetilde{t}$ and $\overline{t}$ indicate these parameters for the two cases, respectively. We wish to minimize the two proposed QoE metrics over the choice of scheduling and access decisions. Since this is a multi-objective optimization, the objective can be modeled as a convex combination of the two QoE metrics.

Let $\overline{\lambda}=\sum_{i}\lambda_{i}$ be the total arrival rate. Then, $\lambda_{i}/\overline{\lambda}$ is the ratio of video $i$ requests. The first objective is the minimization of the mean stall duration, averaged over all the file requests, and is given as $\sum_{i}\frac{\lambda_{i}}{\overline{\lambda}}\,\mathbb{E}\left[\Gamma^{\left(i\right)}\right]$. The second objective is the minimization of stall duration tail probability, averaged over all the file requests, and is given as $\sum_{i}\frac{\lambda_{i}}{\overline{\lambda}}\,{\text{Pr}}\left(\Gamma^{(i)}\geq x\right)$. Using the expressions for the mean stall duration and the stall duration tail probability in Chapters \ref{sec:msd} and \ref{sec:tsd}, respectively, optimization of a convex combination of the two QoE metrics can be formulated as follows. 

\begin{eqnarray}
\text{min\,\,\,\,\,}&&\sum_{i}\frac{\lambda_{i}}{\overline{\lambda}}\left[\theta\,\frac{1}{\widetilde{t}_{i}}\text{log}\left(\sum_{j=1}^{m}\pi_{ij}\left(1+\widetilde{H}_{ij}\right)\right)\right.\nonumber \\
&&+\left.\left(1-\theta\right)\sum_{j}\frac{\pi_{ij}}{e^{\overline{t}_{i}x}}\left(1+e^{-\overline{t}_{i}\left(d_{s}+(L_{i}-1)\tau\right)}\,\overline{H}_{ij}\right)\right]\label{eq:joint_otp_prob}\\
\mbox{s.t.}\,\,\,\,\, &&\widetilde{H}_{ij}=\frac{e^{-\widetilde{t}_{i}\left(d_{s}-\tau\right)}\left(1-\rho_{j}\right)\widetilde{t}_{i}}{\widetilde{t}_{i}-\Lambda_{j}\left(B_{j}(\widetilde{t}_{i})-1\right)}\widetilde{Q}_{ij}\,\, ,
\label{eq:H_ij}\\
&&\overline{H}_{ij}=\frac{e^{-\overline{t}_{i}\left(d_{s}-\tau\right)}\left(1-\rho_{j}\right)\overline{t}_{i}}{\overline{t}_{i}-\Lambda_{j}\left(B_{j}(\overline{t}_{i})-1\right)}\overline{Q}_{ij}\,\, ,
\label{eq:H_ij2}\\
&&\widetilde{Q}_{ij} =\left[\frac{\widetilde{M}_{j}(\widetilde{t}_{i})\left(1-\left(\widetilde{M}_{j}(\widetilde{t}_{i})\right)^{L_{i}}\right)}{1-\widetilde{M}_{j}(\widetilde{t}_{i})}\right],\,\,
\label{eq:Q_ij}\\
&& \overline{Q}_{ij} =\left[\frac{\widetilde{M}_{j}(\overline{t}_{i})\left(1-\left(\widetilde{M}_{j}(\overline{t}_{i})\right)^{L_{i}}\right)}{1-\widetilde{M}_{j}(\overline{t}_{i})}\right],\,\,
\label{eq:Q_ij2}\\
&  & \widetilde{M}_{j}({t})=\frac{\alpha_{j}e^{\left(\beta_{j}-\tau\right){t}}}{\alpha_{j}-{t}},\,\, \label{eq:M_telda_opt2}
\end{eqnarray}

\begin{eqnarray}
&  &  {B}_{j}(t)=\sum_{f=1}^r\frac{\lambda_{f}\pi_{fj}}{\Lambda_{j}}\left(\frac{\alpha_{j}e^{\beta_{j}{t}}}{\alpha_{j}-{t}}\right)^{L_f}\, , \label{eq:Bj_const}\\
&  & \widetilde{M}_{j}({t})=\frac{\alpha_{j}e^{\left(\beta_{j}-\tau\right){t}}}{\alpha_{j}-{t}},\,\, \label{eq:M_telda_opt2}\\ 
&  & {B}_{j}(t)=\sum_{f=1}^r\frac{\lambda_{f}\pi_{fj}}{\Lambda_{j}}\left(\frac{\alpha_{j}e^{\beta_{j}{t}}}{\alpha_{j}-{t}}\right)^{L_f}\, , \label{eq:Bj_const}\\ 
&  & \rho_{j}=\sum_{f=1}^r\pi_{fj}\lambda_{f}L_{f}\left(\beta_{j}+\frac{1}{\alpha_{j}}\right)<1\,\,\,\,\,\,\forall j\label{eq:rho_j}\\
&  & \varLambda_{j}=\sum_{f=1}^r\lambda_{f}\pi_{f,j}\,\,\,\,\,\forall j\label{eq:Lambda_j}\\
&  & \sum_{j=1}^{m}\pi_{i,j}=k_{i}\,\,\,\,\label{eq:sum_ij}\\
&  & \mbox{ \ensuremath{\pi_{i,j}}=0}\,\,\,\mbox{if \ensuremath{j\notin S_{i}}}\,,\ensuremath{\pi_{i,j}}\in\left[0,1\right]\label{eq:pij}\\
&  & \left|\mathcal{S}_{i}\right|=n_{i},\,\,\forall i\label{eq:S_i_and_ni}\\
&  & 0<\widetilde{t}_{i}<\alpha_j,\,\forall j \label{eq:t_i_alpha_j}\\
&  & 0<\overline{t}_{i}<\alpha_j,\,\forall j \label{eq:t_i_alpha_j2}\\
&  & \alpha_{j}\left(e^{(\beta_j-\tau)\widetilde{t}_i}-1\right)+\widetilde{t}_i<0\,,\forall j \label{M_telda_less_1_2} \\
&  & \alpha_{j}\left(e^{(\beta_j-\tau)\overline{t}_i}-1\right)+\overline{t}_i<0\,,\forall j \label{M_telda_less_1} \\
&  & \sum_{f=1}^r\pi_{fj}\lambda_{f}\left(\frac{\alpha_{j}e^{\beta_{j}\widetilde{t}_{i}}}{\alpha_{j}-\widetilde{t}_{i}}\right)^{L_{f}}-\left(\Lambda_{j}+\widetilde{t}_{i}\right)<0,\,\forall i, j\label{eq:don_pos_cond}\\
&  & \sum_{f=1}^r\pi_{fj}\lambda_{f}\left(\frac{\alpha_{j}e^{\beta_{j}\overline{t}_{i}}}{\alpha_{j}-\overline{t}_{i}}\right)^{L_{f}}-\left(\Lambda_{j}+\overline{t}_{i}\right)<0,\,\forall i, j\label{eq:don_pos_cond2}\\
&  & \mbox{var.} \ \ \ \  \    \boldsymbol{\pi},\boldsymbol{t}, \mathcal{\boldsymbol{S}}
\label{eq:vars}
\end{eqnarray}

Here, $\theta\in [0,1]$ is a trade-off factor that determines the relative significance of mean and tail probability of the stall durations in the minimization problem. Varying $\theta=0$ to $\theta=1$, the solution for
(\ref{eq:joint_otp_prob}) spans the solutions that  minimize the mean stall duration to ones that minimize the stall duration tail probability. Note that constraint (\ref{eq:rho_j}) gives the load intensity of server $j$. Constraint (\ref{eq:Lambda_j}) gives the aggregate arrival rate $\Lambda_j$ for each node for the  given probabilistic scheduling probabilities $\pi_{ij}$ and arrival rates $\lambda_i$. Constraints \eqref{eq:pij}-\eqref{eq:S_i_and_ni} guarantees that the scheduling probabilities are feasible. Constraints (\ref{eq:t_i_alpha_j})-(\ref{M_telda_less_1})  ensure that $\widetilde{M}_{j}({t})$ exist for each $\widetilde{t}_i$ and $\overline{t}_i$.  Finally, Constraints (\ref{eq:don_pos_cond})-(\ref{eq:don_pos_cond2}) ensure that the moment generating function given in (\ref{eq:M_D_ij}) exists. We note that the the optimization over $\boldsymbol{\pi}$ helps decrease the objective function and gives significant flexibility over choosing the lowest-queue servers for accessing the files. The placement of the video files $\mathcal{\boldsymbol{S}}$ helps separate the highly accessed files on different servers thus reducing the objective. Finally, the optimization over the auxiliary variables  $\boldsymbol{t}$  gives a tighter bound on the objective function. We note that the QoE for file $i$ is weighed by the arrival rate $\lambda_i$ in the formulation. However, general weights can be easily incorporated for weighted fairness or differentiated services. 

Note that the proposed  optimization problem is a mixed integer non-convex  optimization
as we have the placement over $n$ servers and the constraints \eqref{eq:don_pos_cond} and \eqref{eq:don_pos_cond2} are non-convex in $(\boldsymbol{\pi},\boldsymbol{t})$. The problem can be solved using an optimization algorithm described in \citep{al2018video}, which in part uses NOVA  algorithm proposed in \citep{scutNOVA}. 

 We simulate our algorithm in a distributed storage system of $m=12$ distributed nodes, where each video file uses an $(10,4)$ erasure code. The parameters for storage servers are chosen as in Table \ref{tab:Storage-Nodes-Parameters1}, which were chosen in \citep{Yu_TON} in the experiments using Tahoe testbed.  Further,
$(10,4)$ erasure code is used in HDFS-RAID in Facebook \citep{HDFS_ec} and Microsoft \citep{Asure14}. Unless otherwise explicitly stated, we consider $r=1000$ files, whose sizes are generated based on Pareto distribution \citep{arnold2015pareto} with shape factor of $2$ and scale of $300$, respectively.  We note that the Pareto distribution is considered as it has been widely used in existing
literature  \citep{Vaphase} to
model video files, and file-size distribution over networks.  We also assume that the chunk service time  follows a shifted-exponential distribution with rate $\alpha_{j}$ and shift $\beta_{j}$, whose values are shown in Table I, which are generated at random and kept fixed for the experiments. Unless explicitly stated, the arrival rate for the first $500$ files is $0.002s^{-1}$ while for the next $500$ files is set to be $0.003s^{-1}$. Chunk size $\tau$ is set to be equal to $4$ s. When generating video files, the sizes of the video file sizes are rounded up to the multiple of $4$ sec. We note that a high load scenario is considered for the numerical results. In order to initialize our algorithm, we use a random placement of files on all the servers. Further, we set $\pi_{ij}=k/n$ on the placed servers with $t_{i}=0.01$ $\forall i$ and $j \in \mathcal{S}_{i}$. However, these choices of $\pi_{ij}$ and $t_{i}$ may not be feasible. Thus, we modify the initialization of  $\boldsymbol{\pi}$  to be closest norm feasible solution  given above values of $\boldsymbol{\mathcal{S}}$ and $\boldsymbol{t}$. We compare the proposed approach with some baselines: 
 \begin{enumerate}[leftmargin=0cm,itemindent=.5cm,labelwidth=\itemindent,labelsep=0cm,align=left]
	\item {\em Random Placement, Optimized Access (RP-OA): }  In this strategy, the placement is chosen at random where any $n$  out of $m$ servers are chosen for each file, where each choice is equally likely. Given the random placement, the variables $\boldsymbol{t}$ and $\boldsymbol{\pi}$ are optimized using the proposed algorithm, where $\boldsymbol{\mathcal{S}}$-optimization is not performed.
	
	\item {\em Optimized Placement, Projected Equal Access (OP-PEA): }
	The strategy utilizes $\boldsymbol{\pi}$, $\boldsymbol{t}$ and $\boldsymbol{\mathcal{S}}$ as mentioned in the setup. Then, alternating optimization over placement and $\boldsymbol{t}$ are performed using the proposed algorithm.

	\item {\em Random Placement, Projected Equal Access (RP-PEA): } In this strategy, the placement is chosen at random where any $n$  out of $m$ servers are chosen for each file, where each choice is equally likely. Further, we set $\pi_{ij}=k/n$ on the placed servers with $t_{i}=0.01$ $\forall i$ and $j \in \mathcal{S}_{i}$. We then modify the initialization of $\boldsymbol{\pi}$ to be closest norm feasible solution  given above values of $\boldsymbol{\mathcal{S}}$ and $\boldsymbol{t}$. Finally, an optimization over $\boldsymbol{t}$ is performed with respect to the objective using the proposed algorithm.
	
	\item  OP-PSP ({\em  Optimized Placement-Projected Service-Rate Proportional Allocation}) Policy: The joint request scheduler  chooses the access probabilities to be proportional to the service rates of the storage nodes, i.e., $\pi_{ij}=k_{i}\frac{\mu_{j}}{\sum_{j}\mu_{j}}$.  This policy assigns servers proportional to their service rates. These access probabilities are projected toward feasible region  for a uniformly random placed files to ensure stability of the storage system.  With these fixed access probabilities,  the weighted  mean stall duration and stall duration tail probability are optimized over the $\mathbf{t}$, and placement $\boldsymbol{\mathcal{S}}$.
	
	\item  RP-PSP   ({\em  Random Placement-PSP}) Policy: As compared to the OP-PSP Policy, the chunks are placed uniformly at random. The weighted  mean stall duration and stall duration tail probability are optimized over  the choice of auxiliary variables ${\mathbf t}$.
\end{enumerate}

	\begin{figure}
	\centering
	\includegraphics[trim=0.0in 0in 4.3in 0in, clip,width=0.7\textwidth]{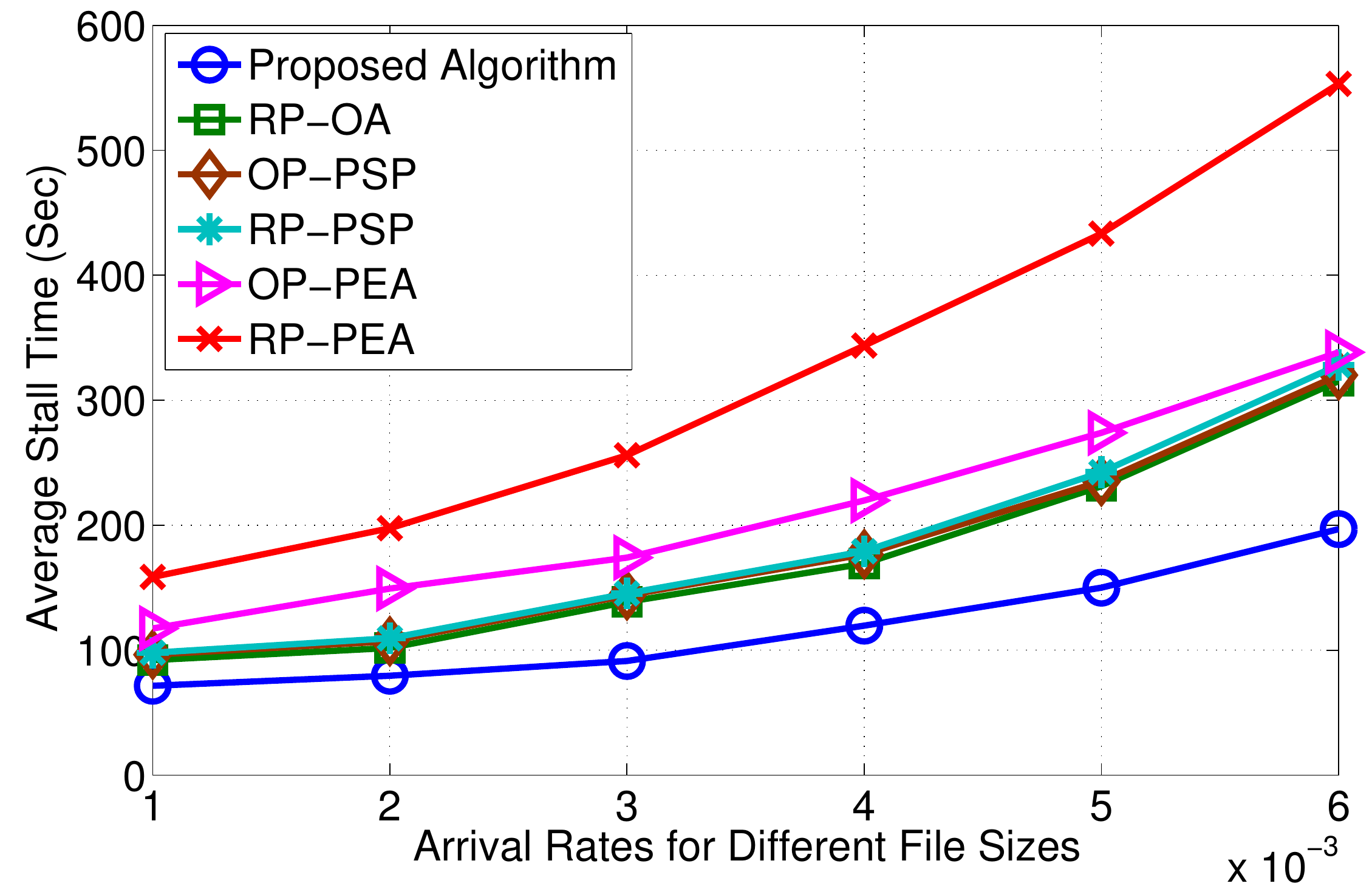}
	\caption{Mean stall duration for different video arrival rates with different  video lengths.  }
	\label{fig:meanArrRateDiffSize}
\end{figure}

Figure \ref{fig:meanArrRateDiffSize} shows the effect of different video arrival rates on the mean stall duration for different-size video length. The different size uses the Pareto-distributed lengths described above.  We compare our proposed algorithm with the five baseline  policies and we see that the proposed algorithm outperforms all baseline strategies for the QoE metric of mean stall duration. Thus, both access and placement of files are both important for the reduction of mean stall duration. Further, we see that the mean stall duration increases with arrival rates, as expected. Since the mean stall duration is more significant at high arrival rates, we notice a significant improvement in mean stall duration by about 60\% ( approximately 700s to about 250s) at the highest arrival rate in Figure \ref{fig:meanArrRateDiffSize} as compared to the random placement and projected equal access policy.

	\begin{figure}
	\centering
	\includegraphics[trim=0in 0in 4.1in 0in, clip,width=0.7\textwidth]{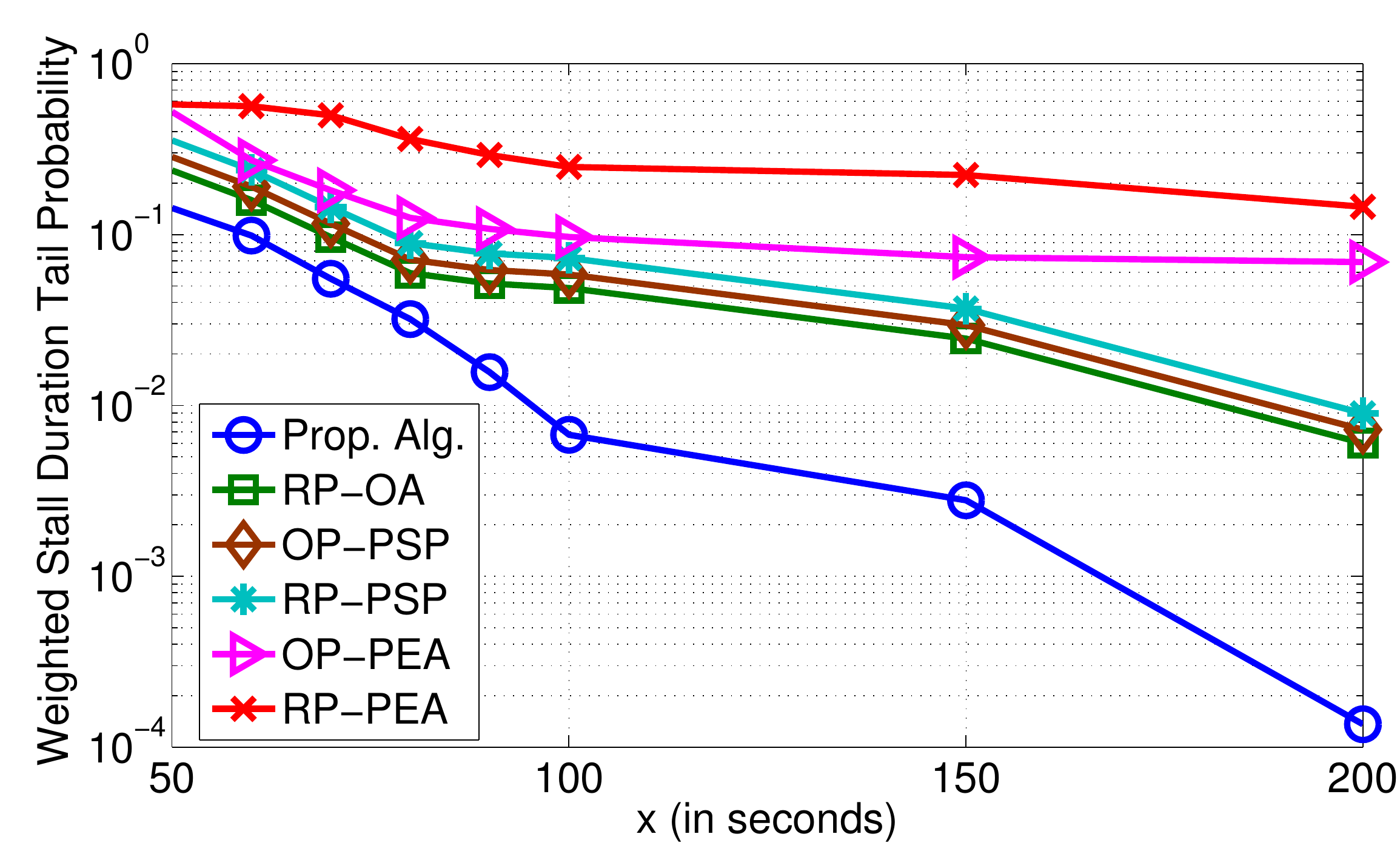}
	\caption{Stall duration tail probability for different values of $x$ (in seconds).  }
	\label{fig:tailProbDiffX}
\end{figure}

Figure \ref{fig:tailProbDiffX} shows the decay of weighted stall duration tail probability with respect to $x$ (in seconds) for the proposed and the baseline strategies. In order to signify (magnify) the small differences, we plot y-axis in logarithmic scale. We observe that the proposed algorithm gives orders improvement in the stall duration tail probabilities as compared to the baseline strategies. 

		\begin{figure}
	\centering
	\includegraphics[trim=0in 0in 0.0in 0in, clip,width=0.7\textwidth]{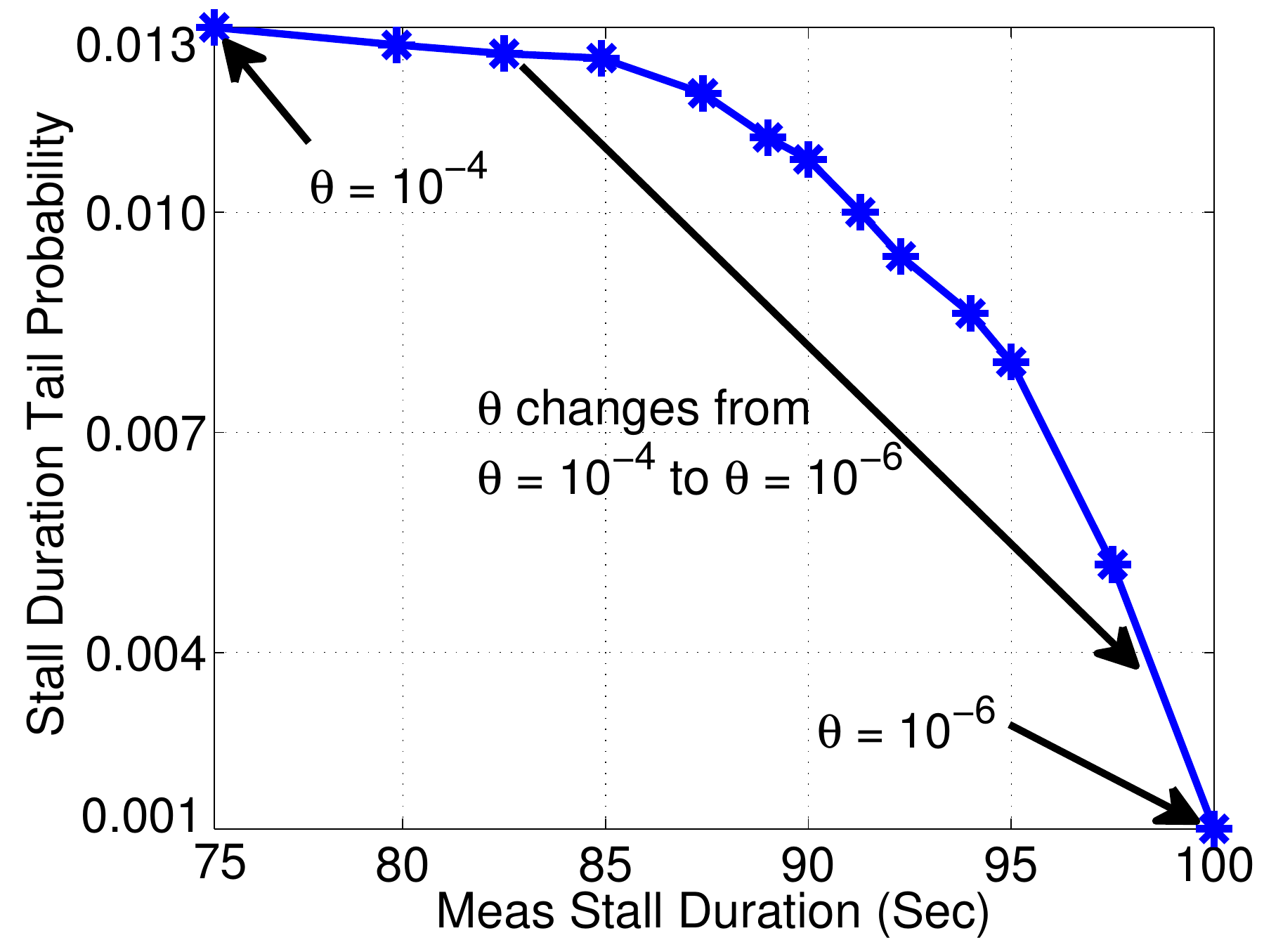}
	\caption{Tradeoff between mean stall duration and stall duration tail probability obtained by varying $\theta$.  }
	\label{fig:tradeoff}
\end{figure}

If the mean stall duration decreases, intuitively the stall duration tail probability also reduces. Thus, a question arises whether the optimal point for decreasing the mean stall duration and the stall duration tail probability is the same. We answer the question in negative since for $r=1000$ of equal sizes of length 300 sec, we find that at the values of  ($\boldsymbol{\pi}$, $\boldsymbol{\mathcal{S}}$) that optimize the mean stall duration, the stall duration tail probability is 12 times higher as compared to the optimal stall duration tail probability. Similarly, the optimal mean stall duration is 30\% lower as compared to the mean stall duration at the value of ($\boldsymbol{\pi}$, $\boldsymbol{\mathcal{S}}$) that optimizes the stall duration tail probability. Thus, an efficient tradeoff point between the QoE metrics can be chosen based on the point on the curve that is appropriate for the clients.

\section{Notes and Open Problems}\label{sec:vsr_notes}
Servicing Video on Demand and Live TV Content from cloud servers have been studied widely \citep{lee2013vod,huang2011cloudstream,he2014cost,chang2016novel,oza2016implementation}.  The reliability of content over the cloud servers have  been first considered for video streaming applications in \citep{al2018video}. In this work, the mean and tail stall duration are characterized. Based on this analysis of the metrics, joint placement of content and resource optimization over the cloud servers have been considered. Even though we consider single steam  from each server to the edge router, we can extend the approach to multiple parallel streams.  Multiple streams can help obtain parallel video files thus helping one file not wait behind the other. This extension has also been studied in \citep{al2018video}. The results have been further extended to a Virtualized Content Distribution Network (vCDN)
architecture in \citep{al2019fasttrack}, which consists of a remote datacenter that stores
complete original video data and multiple CDN sites (i.e.,
local cache servers) that only have part of those data and are
equipped with solid state drives (SSDs) for higher throughput.
A user request for video content not satisfied in the local
cache is directed to, and processed by, the remote datacenter. If the required video content/chunk is
not stored in cache servers, multiple parallel connections are
established between a cache server and the edge router, as
well as between the cache servers and the origin server, to
support multiple video streams simultaneously. This work uses cache at SDN servers, while not at the edge routers. Using caching at edge routers based on an adaptation of Least-Recently-Used (LRU) strategy, the stall duration metrics have been characterized in \citep{AbubakrInfW20181,al2019multi}.  These work assume a single quality video. An approach to select one of the different quality levels to have an efficient tradeoff between video quality and the stall duration metrics has been considered in  \citep{8724450,Abubakr_Stream}. 

These works lay the foundations for many important future problems, including
\begin{enumerate}
	\item {\bf Adaptive Streaming}: Adaptive streaming algorithms have  been considered for video streaming \citep{chen2012amvsc,wang2013ames,elgabli2018lbp,elgabli2019fastscan,elgabli2019smartstreamer,elgabli2018optimized,elgabli2018giantclient,elgabli2019groupcast}. However, the above works consider entire video streaming at the same quality. Considering the aspects of adavptive video streaming to compute the stall duration metrics and the video quality metrics jointly is an important problem. 
	\item {\bf Efficient Caching Algorithms}: In the vCDN environment, we consider optimized caching at the CDN servers and LRU based caching at the edge router. However, with different file sizes, different caching algorithms have been studied \citep{berger2017adaptsize,halalai2017agar,friedlander2019generalization}. Considering efficient caching strategies for video streaming is an important problem for the future.
\end{enumerate}

\chapter{Lessons from prototype implementation}\label{chpt:impl}

Various models and theories proposed in Chapters~\ref{chp:MDS} to \ref{chp:relauch} provide mathematical crystallization of erasure coded storage systems, by quantifying different performance metrics, revealing important control knobs and tradeoffs, illuminating opportunities for novel optimization algorithms, and thus enabling us to rethink the design of erasure coded storage systems in practice. On the other hand, implementing these algorithms and designs in practical storage systems allows us to validate/falsify different modeling assumptions and provides crucial feedback to bridge the divide between theory and practice. In this chapter, we introduce a practical implementation to demonstrate the path for realizing erasure-coded storage systems. Then, we provide numerical examples to illuminate key design tradeoffs in this system. Finally, another application of the models to distributed caching and content distribution is described, and numerical examples discussed. We highlight key messages learned from these experiments as remarks throughout this chapter.

\section{Exemplary implementation of erasure-coded storage}
\label{sec:impl_tahoe}

Distributed systems such as Hadoop, AT$\&$T Cloud Storage, Google File System and Windows Azure have evolved to support different types of erasure codes, in order to achieve the benefits of improved storage efficiency while providing the same reliability as replication-based schemes \citep{balaji2018erasure}. In particular, Reed-Solomon (RS) codes have been implemented in the Azure production cluster and resulted in the savings of millions of dollars for Microsoft \citep{Asure14,B180}. Later, Locally Recoverable (LR) codes were implemented in HDFS-RAID  carried out in Amazon EC2 and a cluster at Facebook in \citep{B2}. Various erasure code plug-ins and libraries have been developed in storage systems like Ceph \citep{Ceph,Yu-TON16}, Tahoe \citep{Yu_TON},  Quantcast (QFS) \citep{ovsiannikov2013quantcast}, and Hadoop (HDFS) \citep{B182}. In a separate line of work, efficient repair schemes and traffic engineering in erasure-coded storage systems are discussed in \citep{plank2009performance,dimakis2010network,zhou2020fast,code_repair3}.

In this chapter, we report an implementation of erasure coded storage in {\em Tahoe} \citep{Tahoe}, which is an open-source, distributed filesystem based on the {\em zfec} erasure coding library. It provides three special instances of a generic {\em node}: (a)  {\em Tahoe Introducer}: it keeps track of a collection of storage servers and clients and introduces them to each other.   (b) {\em Tahoe Storage Server}: it exposes attached storage to external clients and stores erasure-coded shares.  (c) {\em Tahoe Client}: it processes upload/download requests and connects to storage servers through a Web-based REST API and the Tahoe-LAFS (Least-Authority File System) storage protocol over SSL.

\begin{figure}[!thbp]
\begin{center}
\scalebox{0.3}{\includegraphics[draft=false]{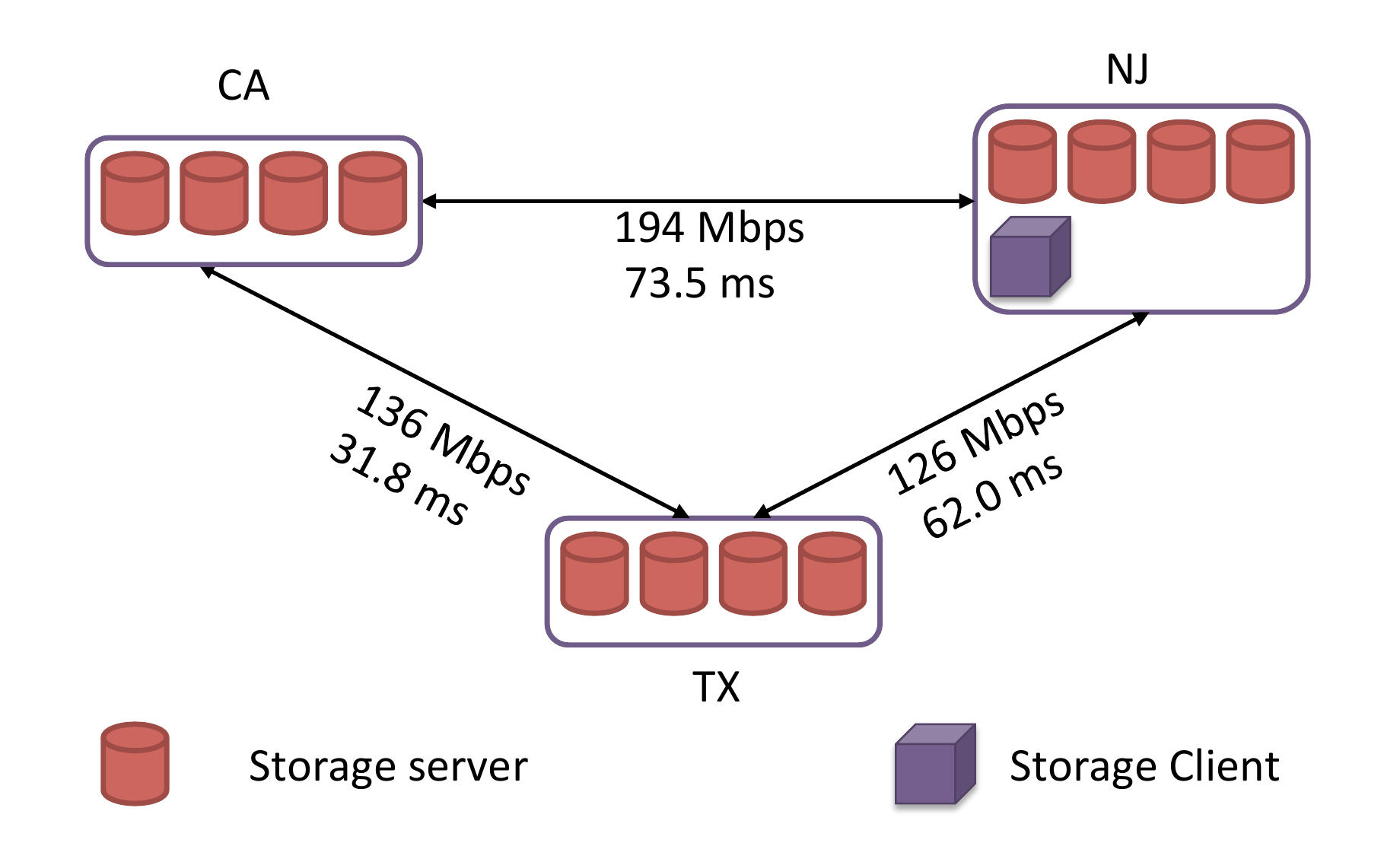}}
\caption{Our Tahoe testbed with average ping (RTT) and bandwidth measurements among three data centers in New Jersey, Texas, and California}
\label{fig:sl-testbed}
\end{center}
\end{figure}

While Tahoe uses a default $(10,3)$ erasure code, it supports arbitrary erasure code specification statically through a configuration file. In Tahoe, each file is
 encrypted, and is then broken into a set of segments, where each segment consists of $k$ blocks.  Each segment is then erasure-coded to produce $n$ blocks (using a $(n,k)$ encoding scheme) and then distributed to (ideally) $n$ distinct storage servers. The set of blocks on each storage server constitute a chunk. Thus, the file equivalently consists of $k$ chunks which are encoded into $n$ chunks and each chunk consists of multiple blocks\footnote{If there are not enough servers, Tahoe will store multiple chunks on one sever. Also, the term ``chunk'' we used in this chapter is equivalent to the term ``share'' in Tahoe terminology. The number of blocks in each chunk is equivalent to the number of segments in each file.}. For chunk placement, the Tahoe client randomly selects a set of available storage servers with enough storage space to store $n$ chunks. For server selection during file retrievals, the client first asks all known servers for the storage chunks they might have. Once it knows where to find the needed k chunks (from the k servers that responds the fastest), it downloads at least the first segment from those servers. This means that it tends to download chunks from the ``fastest'' servers purely based on round-trip times (RTT). %

To implement different scheduling algorithms in Tahoe, we would need to modify the upload and download modules in the Tahoe storage server and client to allow for customized and explicit server selection, which is specified in the configuration file that is read by the client when it starts.  In addition, Tahoe performance suffers from its single-threaded design on the client side for which we had to use multiple clients with separate ports to improve parallelism and bandwidth usage during experiments

We deployed 12 Tahoe storage servers as virtual machines in an OpenStack-based data center environment distributed in New Jersey (NJ), Texas (TX), and California (CA).  Each site has four storage servers.  One additional storage client was deployed in the NJ data center to issue storage requests.  The deployment is shown in Figure~\ref{fig:sl-testbed} with average ping (round-trip time) and bandwidth measurements listed among the three data centers.  We note that while the distance between CA and NJ is greater than that of TX and NJ, the maximum bandwidth is higher in the former case. The RTT time measured by ping does not necessarily correlate to the bandwidth number. Further, the current implementation of Tahoe does not use up the maximum available bandwidth, even with our multi-port revision.

\begin{figure}[!thbp]
\begin{center}
\scalebox{0.3}{\includegraphics[draft=false, trim = 0in 0in 0in 0in, clip]{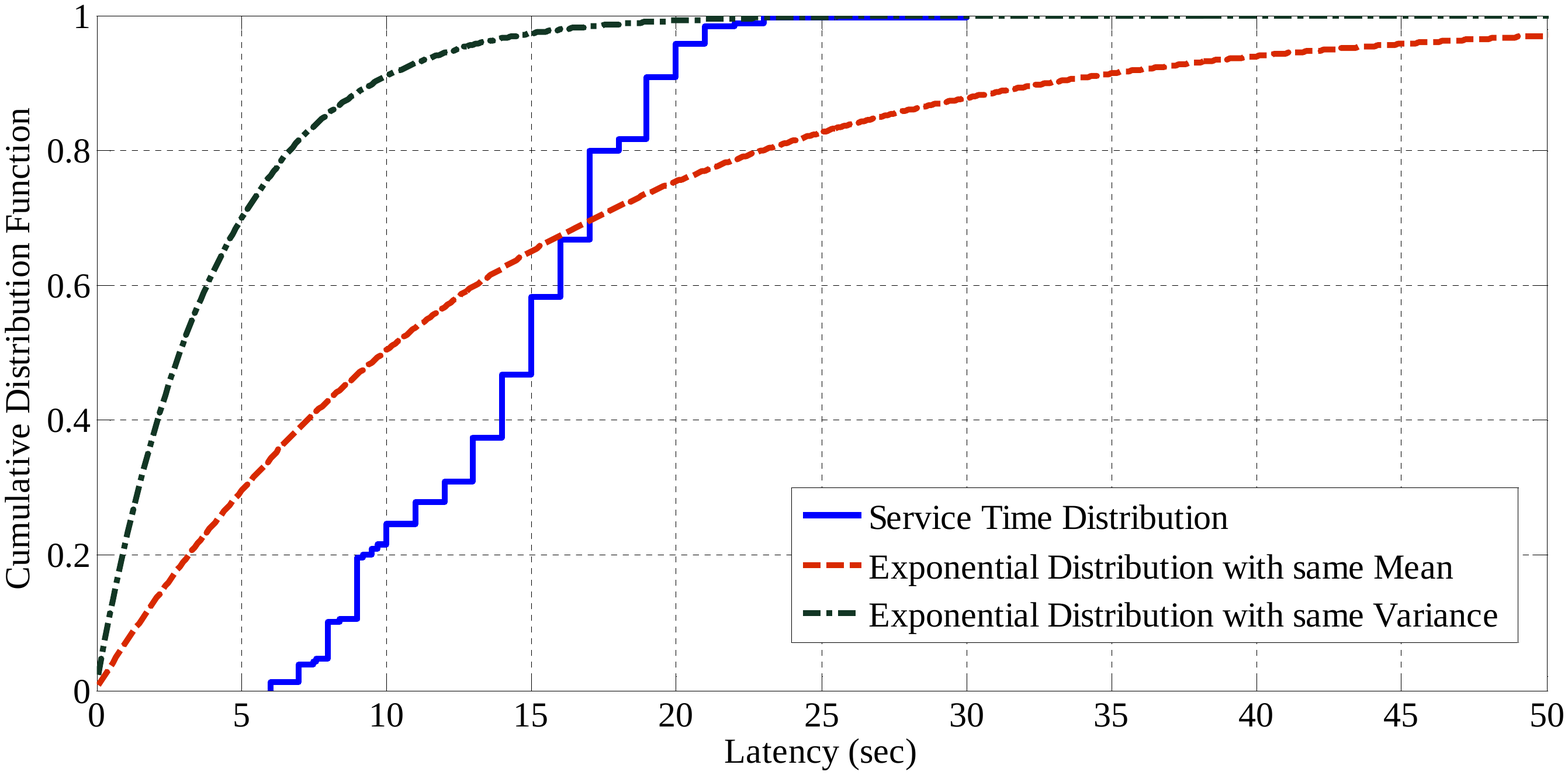}}
\caption{Comparison of actual service time distribution and an exponential distribution with the same mean. It verifies that actual service time does not follow an exponential distribution, falsifying the assumption in previous work \citep{MG1:12}. }
\label{fig:service_dist}
\end{center}
\end{figure}

\vspace{0.07in}
{\noindent \bf Remark 1}: Actual service time can be approximated well by a shifted exponential distribution. Using this testbed, we can run experiments to understand actual service time distribution on our testbed. We upload a 50MB file using a $(7,4)$ erasure code and measure the chunk service time. Figure \ref{fig:service_dist} depicts the Cumulative Distribution Function (CDF) of the chunk service time. Using the measured results, we get the mean service time of $13.9$ seconds with a standard deviation of 4.3 seconds, second moment of 211.8 $s^2$ and the third moment of 3476.8 $s^3$. We compare the distribution to the exponential distribution( with the same mean and the same variance, respectively) and note that the two do not match. It verifies that actual service time does not follow an exponential distribution, and therefore, the assumption of exponential service time in \citep{MG1:12} is falsified by empirical data. The observation is also evident because a distribution never has positive probability for very small service time. Further, the mean and the standard deviation are very different from each other and cannot be matched by any exponential distribution.

\section{Illuminating key design tradeoffs}

We leverage the implementation prototype in Section \ref{sec:impl_tahoe} to illustrate a number of crucial design tradeoffs in erasure coded storage systems. While the exact tradeoff curves could vary significantly in different systems/environments, the numerical examples presented in this section nevertheless provide a visualization of the various design space available in erasure-coded storage systems.

\begin{figure}[!thbp]
\begin{center}
\includegraphics[scale=.43]{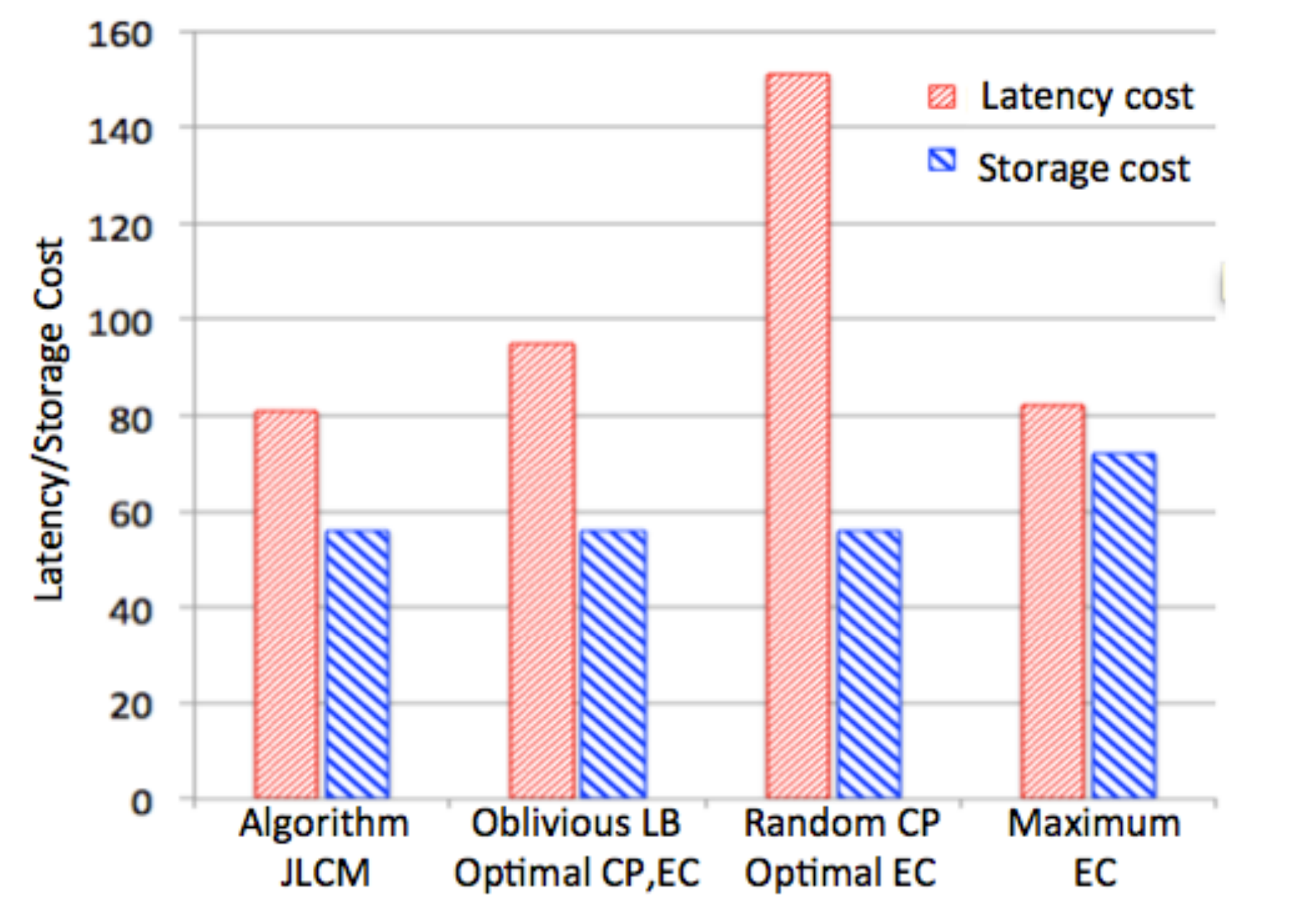}
\caption{Comparison of joint latency and cost minimization with some oblivious approaches. Algorithm JLCM minimizes latency-plus-cost over 3 dimensions: load-balancing (LB), chunk placement (CP), and erasure code (EC), while any optimizations over a subset of the dimensions is non-optimal.}\label{fig:optimality}
\end{center}
\end{figure}

\vspace{0.07in}
{\noindent \bf Remark 2}: Latency and storage cost tradeoff. The use of $(n_i,k_i)$ MDS erasure code allows the content to be reconstructed from any subset of $k_i$-out-of-$n_i$ chunks, while it also introduces a redundancy factor of $n_i/k_i$. To model storage cost, we assume that each storage node $j\in\mathcal{M}$ charges a constant cost $V_j$ per chunk. Since $k_i$ is determined by content size and the choice of chunk size, we need to choose an appropriate $n_i$ which not only introduces sufficient redundancy for improving chunk availability, but also achieves a cost-effective solution. We consider RTT plus expected queuing delay and transfer delay as a measure of latency. To find the optimal parameters for scheduling, we use the optimization using the latency upper bound in Theorem \ref{old_upper_bound}.  To demonstrate this tradeoff, we use the theoretical models to develop an algorithm for Joint Latency and Cost Minimization (JLCM) and compare its performance with three oblivious schemes, each of which minimize latency-plus-cost over only a subset of the 3 dimensions: load-balancing (LB), chunk placement (CP), and erasure code (EC). We run Algorithm JLCM for $r=3$ files of size $(150,150,100)MB$ on our testbed, with $V_j=\$1$ for every $25MB$ storage and a tradeoff factor of $\theta=2$ sec/dollar. The result is shown in Figure.~\ref{fig:optimality}. First, even with the optimal erasure code and chunk placement (which means the same storage cost as the optimal solution from Algorithm JLCM), higher latency is observed in {\em Oblivious LB}, which schedules chunk requests according to a load-balancing heuristic that selects storage nodes with probabilities proportional to their service rates. Second, we keep optimal erasure codes and employ a random chunk placement algorithm, referred to as {\em Random CP}, which adopts the best outcome of 10 random runs. Large latency increment resulted by Random CP highlights the importance of joint chunk placement and load balancing in reducing service latency. Finally, {\em Maximum EC} uses maximum possible erasure code $n=m$ and selects all nodes for chunk placement. Although its latency is comparable to the optimal solution from Algorithm JLCM, higher storage cost is observed. Minimum latency-plus-cost can only be achieved by jointly optimizing over all 3 dimensions.

\begin{figure}[!thbp]
\begin{center}
\scalebox{0.30}{\includegraphics[draft=false]{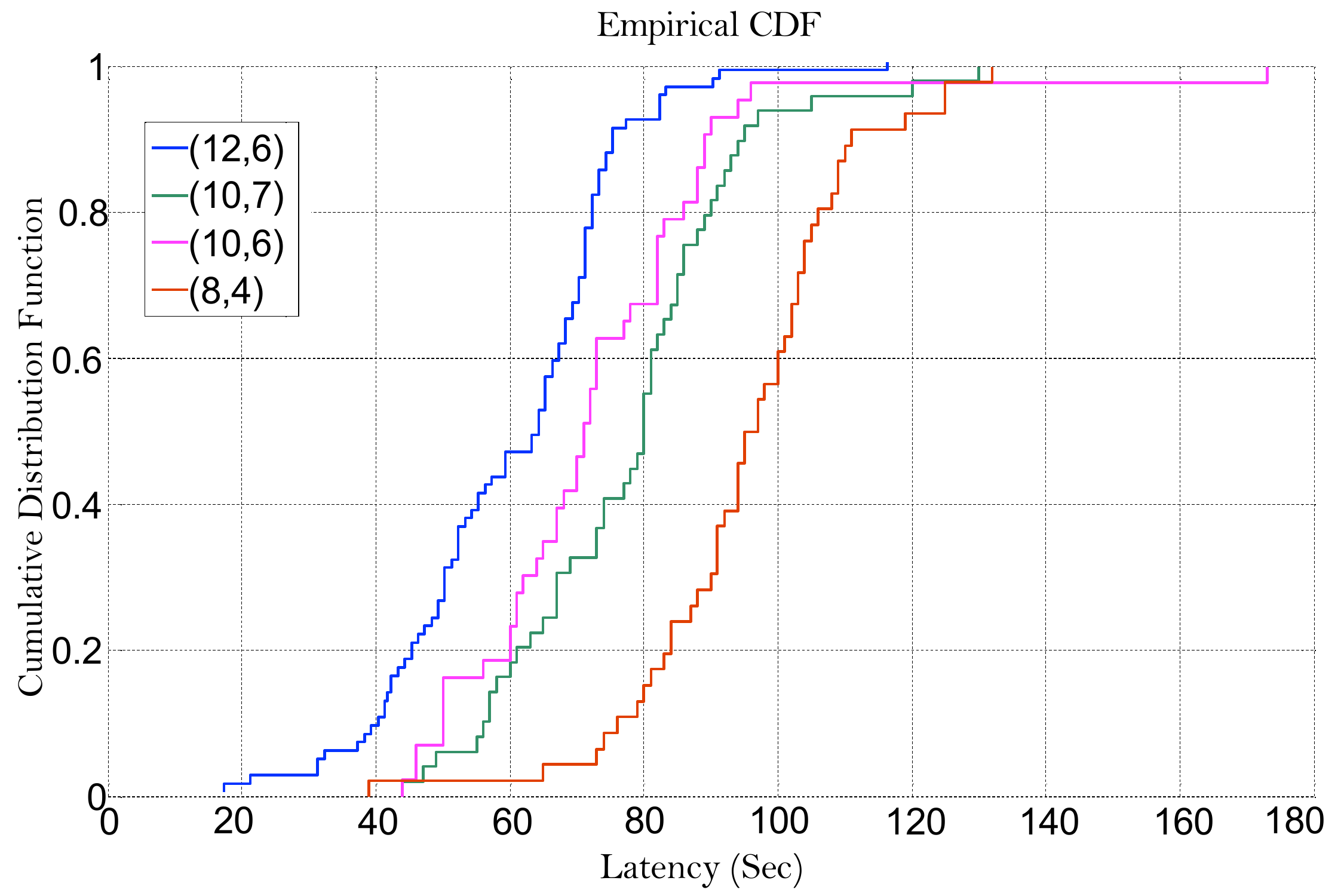}}
\caption{Actual service latency distribution for 1000 files of size 150 MB using erasure codes (12, 6), (10, 7), (10, 6), and (8, 4) for each quarter with aggregate request arrival rates set to $\lambda=0.118/s$. }
\label{fig:latency_distribution}
\end{center}
\end{figure}

\vspace{0.07in}
{\noindent \bf Remark 3}: Latency distribution and coding strategy tradeoff. To demonstrate the impact of coding strategies on latency distribution, we choose files of size 150 MB and the same storage cost and tradeoff factor as in the previous experiment. The files are divided into four classes (each class has 250 files) with different erasure code parameters, respectively (class-1 files using $(n,k)=(12, 6)$ , class-2 files using $(n,k)=(10, 7)$, class 3 using $(n,k)=(10, 6)$, and class 4 has $(n,k)=(8, 4)$). Aggregate request arrival rate for each file class are set to $\lambda_1=\lambda_4=0.0354/s$ and $\lambda_2=\lambda_3=0.0236/s$, which leads to an aggregate file request arrival rate of $\lambda_i=0.118/s$. We are choosing the values of erasure codes for a proper chunk size for our experiments so that the file sizes are widely used for today's data center storage users, and setting different request arrival rates for the two classes using the same value to see the latency distribution under different coding strategies. We retrieve the 1000 files at the designated request arrival rate and plot the CDF of download latency for each file in Fig. 10. We note that 95\% of download requests for files with erasure code (10, 7) complete within 100s, while the same percentage of requests for files using (12, 6) erasure code complete within 32s due to higher level of redundancy. In this experiment, erasure code (12, 6) outperforms (8, 4) in latency though they have the same level of redundancy because the latter has larger chunk size when file size are set to be the same.

\begin{figure}[!thbp]
\begin{center}
\scalebox{0.3}{\includegraphics[draft=false]{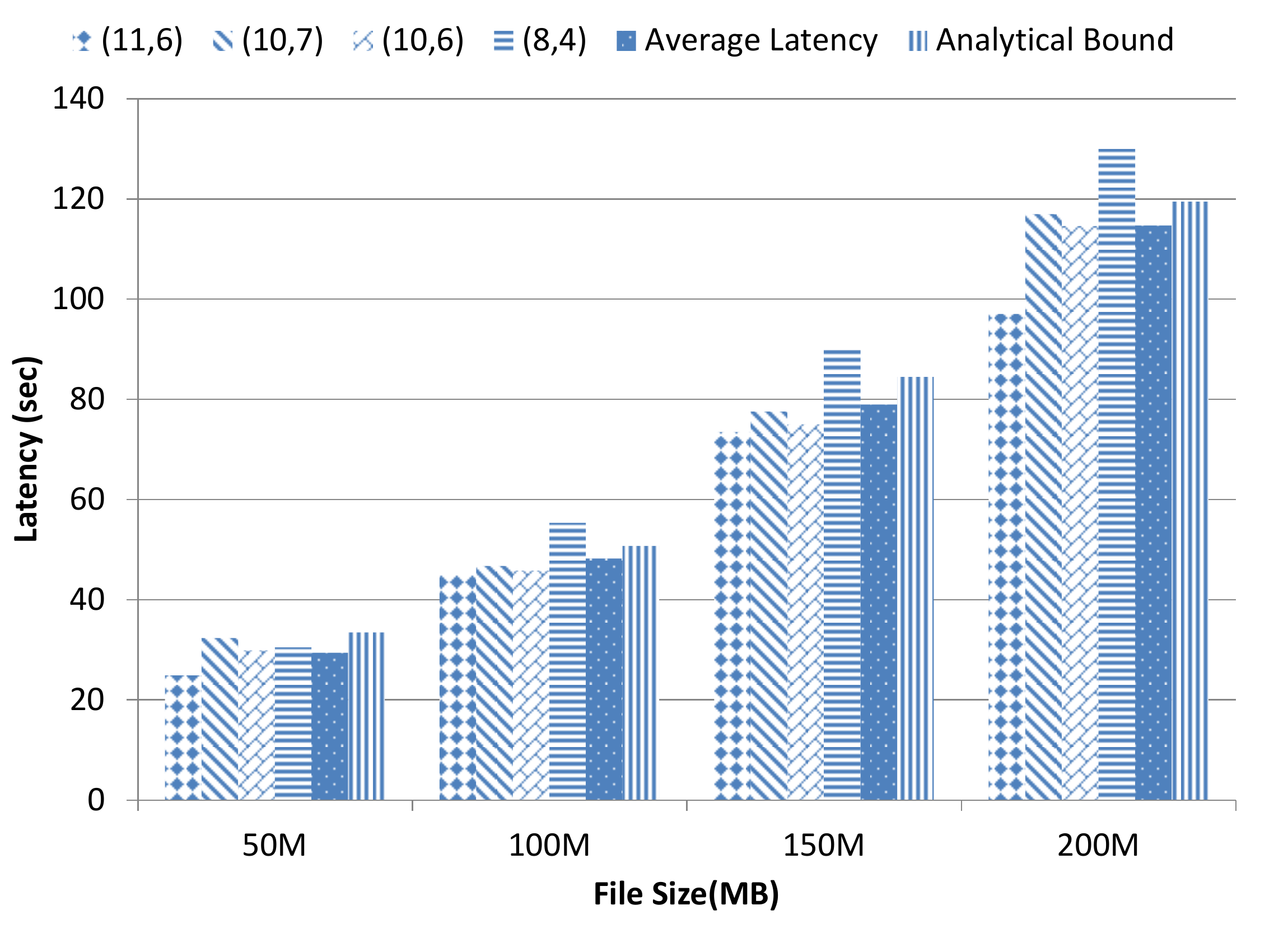}}
\caption{Evaluation of different chunk sizes. Latency increases super-linearly as file size grows due to queuing delay.}
\label{fig:varying_size}
\end{center}
\end{figure}

\vspace{0.07in}
{\noindent \bf Remark 4}: Latency and file size tradeoff. Increasing file size clearly generates high load on the storage system, thus resulting in higher latency. To illustrate this tradeoff, we vary file size in the experiment from (30, 20)MB to (150, 100)MB and plot download latency of individual files 1, 2, 3, average latency, and the analytical latency upper bound~\citep{Yu_TON} in Figure~\ref{fig:varying_size}. We see that latency increases super-linearly as file size grows, since it generates higher load on the storage system, causing larger queuing latency (which is super-linear according to our analysis). Further, smaller files always have lower latency because it is less costly to achieve higher redundancy for these files. We also observe that analytical latency bound in~\citep{Yu_TON} tightly follows actual service latency. In one case, service latency exceeds the analytical bound by 0.5 seconds. This is because theoretical bound quantifying network and queuing delay does not take into account Tahoe protocol overhead, which is indeed small compared to network and queuing delay.

\begin{figure}[!thbp]
\begin{center}
\scalebox{0.3}{\includegraphics[draft=false]{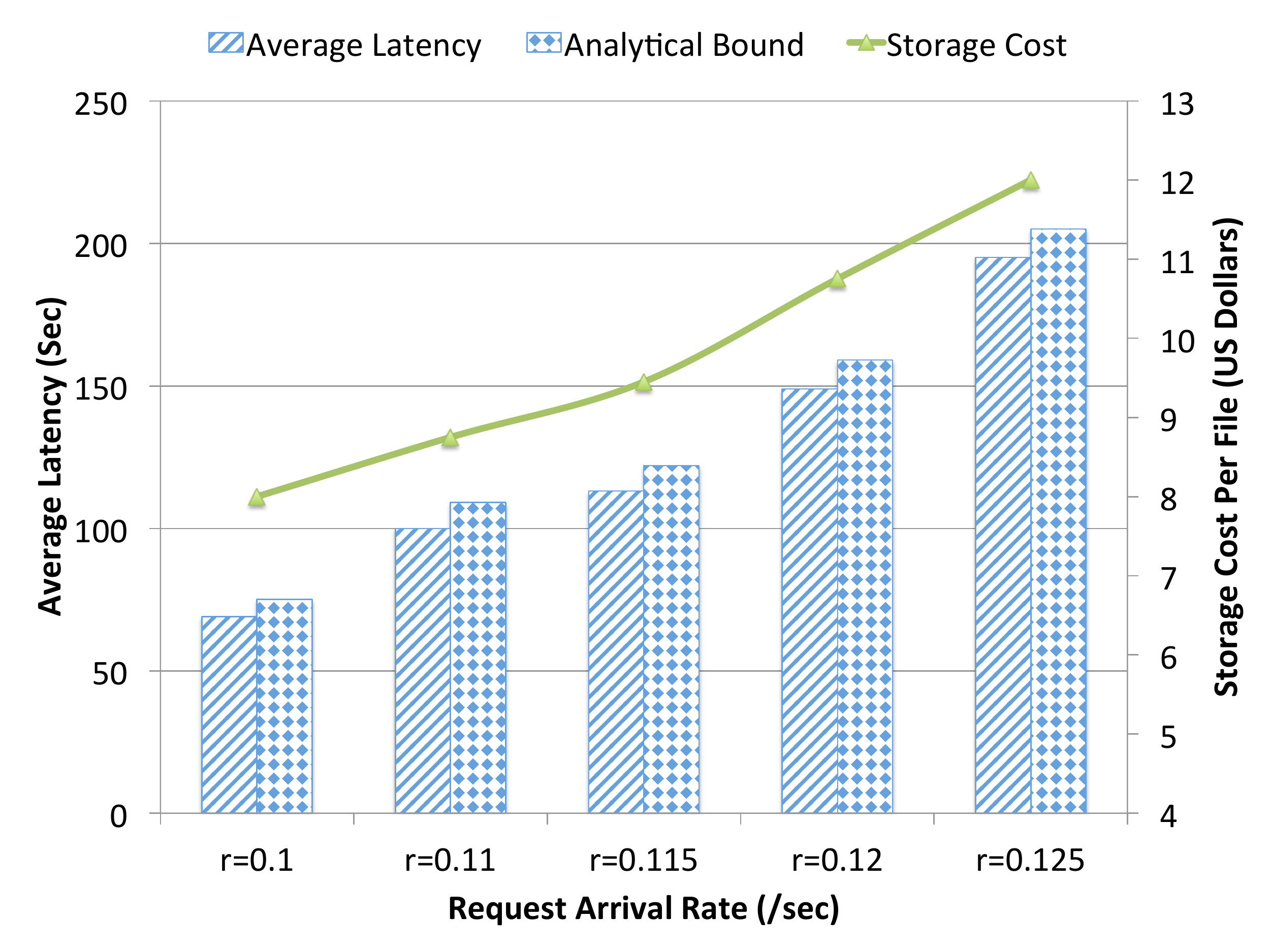}}
\caption{Evaluation of different request arrival rates. As arrival rates increase, latency increases and becomes more dominating in the latency-plus-cost objective than storage cost. }
\label{fig:workload}
\end{center}
\end{figure}

\vspace{0.07in}
{\noindent \bf Remark 5}: Latency and arrival rate tradeoff. We increase the maximum file request arrival rate from $\lambda_i=$1/(60sec) to $\lambda_i=$1/(30sec) (and other arrival rates also increase accordingly), while keeping file size at $(150,150,100)MB$. Actual service delay and the analytical bound~\citep{Yu_TON} for each scenario is shown by a bar plot in Figure~\ref{fig:workload} and associated storage cost by a curve plot. As arrival rates increase, latency increases and becomes more dominating in the latency-plus-cost objective than storage cost. Thus, the marginal benefit of adding more chunks (i.e., redundancy) eventually outweighs higher storage cost introduced at the same time. Figure~\ref{fig:workload} also shows that to achieve a minimization of the latency-plus-cost objective, an optimal solution from theoretical models allows higher storage cost for larger arrival rates, resulting in a nearly-linear growth of average latency as the request arrival rates increase. For instance, Algorithm JLCM chooses (10,6), (11,6), and (10,4) erasure codes at the largest arrival rates, while (8,6), (9,6), and (7,4) codes are selected at the smallest arrival rates in this experiment. We believe that this ability to autonomously manage latency and storage cost for latency-plus-cost minimization under different workload is crucial for practical distributed storage systems relying on erasure coding.

\begin{figure}[!thbp]
\begin{center}
\scalebox{0.57}{\includegraphics[draft=false]{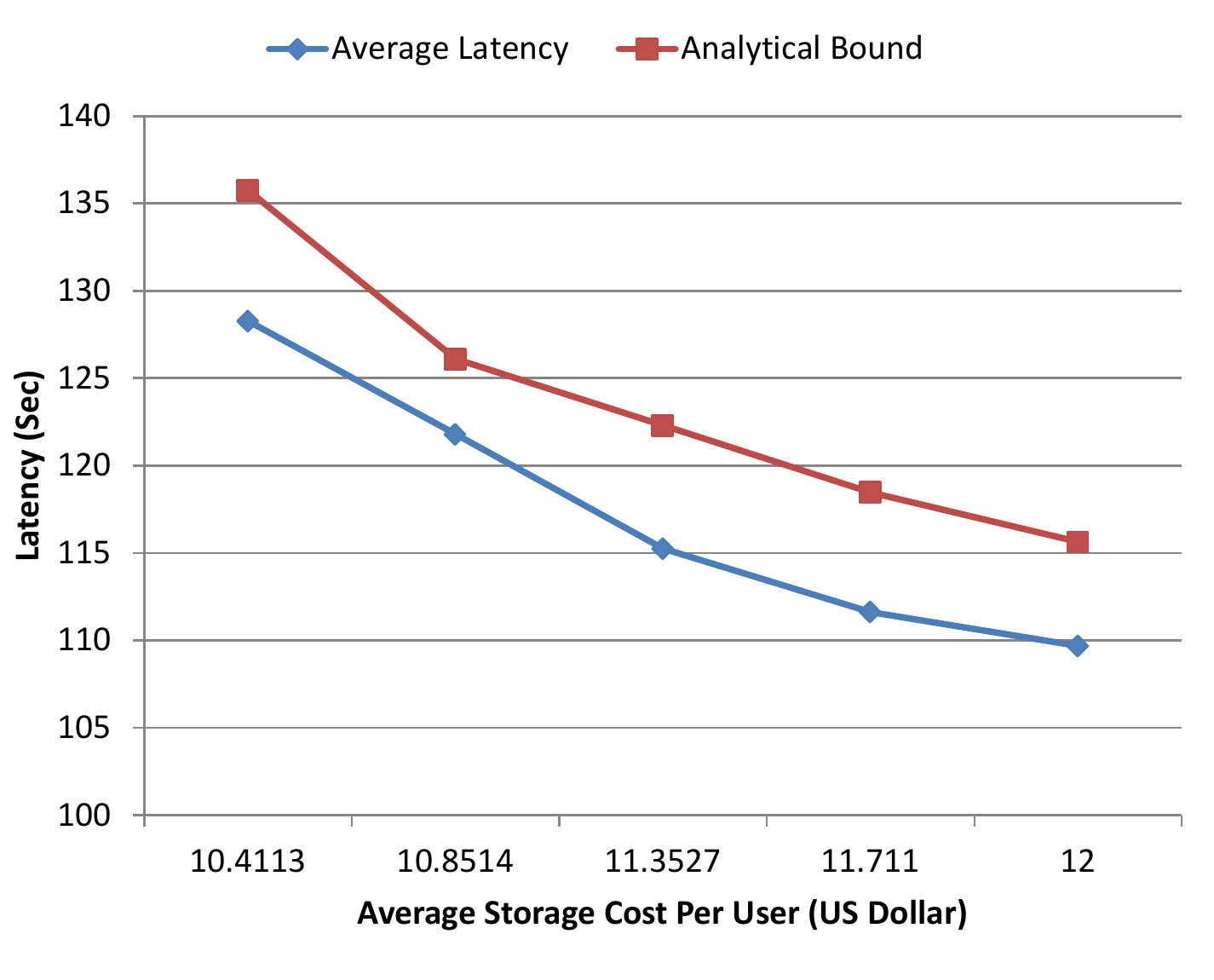}}
\caption{Visualization of latency and cost tradeoff for varying $\theta=0.5$ second/dollar to $\theta=100$ second/dollar. As $\theta$ increases, higher weight is placed on the storage cost component of the latency-plus-cost objective, leading to less file chunks and higher latency. }
\label{fig:tradeoff}
\end{center}
\end{figure}

\vspace{0.07in}
{\noindent \bf Remark 6}: Visualizing tradeoff curve. We demonstrate the tradeoff curve between latency and storage cost. Varying the tradeoff factor in Algorithm JLCM from $\theta=0.5$ sec/dollar to $\theta=100$ sec/dollar for fixed file size of $(150,150,100)MB$ and arrival rates $\lambda_i=$1/(30 sec), 1/(30 sec),  1/(40 sec), we obtain a sequence of solutions, minimizing different latency-plus-cost objectives. As $\theta$ increases, higher weight is placed on the storage cost component of the  latency-plus-cost objective, leading to less file chunks in the storage system and higher latency. This tradeoff is visualized in Figure~\ref{fig:tradeoff}. When $\theta=0.5$, the optimal solution chooses (12,6), (11,6), and (9,4) erasure codes, which is nearly the maximum erasure code length allowable in our experiment and leads to highest storage cost (i.e., 32 dollars), yet lowest latency (i.e., 47 sec). On the other hand, $\theta=0.5$ results in the choice of (6,6),	(7,6), and (4,4) erasure code, which is almost the minimum possible cost for storing the three file, with the highest latency of $65$ seconds. Further, the theoretical tradeoff calculated by analytical bound~\citep{Yu_TON} is very close to the actual measurement on our testbed. These results allow operators to exploit the latency and cost tradeoff in an erasure-coded storage system by selecting the best operating point.

\section{Applications in Caching and Content Distribution}

The application of erasure codes can go beyond data storage. In this section, we briefly introduce a novel caching framework leveraging erasure codes, known as functional caching \citep{Sprout,Yu-TON16}. Historically, caching is a key solution to relieve traffic burden on networks \citep{td_cache}. By storing large chunks of popular data at different locations closer to end-users, caching can greatly reduce congestion in the network and improve service delay for processing file requests. It is very common for 20\% of the video content to be accessed 80\% of the time, so caching popular content at proxies significantly reduces the overall latency on the client side.

However, caching with erasure codes has not been well studied. The current results for caching systems cannot automatically be carried over to caches in erasure coded storage systems. First, using an $(n,k)$ maximum-distance-separable (MDS) erasure code, a file is encoded into $n$ chunks and can be recovered from any subset of $k$ distinct chunks. Thus, file access latency in such a system is determined by the delay to access file chunks on hot storage nodes with slowest performance. Significant latency reduction can be achieved by caching a few hot chunks (and therefore alleviating system performance bottlenecks), whereas caching additional chunks only has diminishing benefits. Second, caching the most popular data chunks is often optimal because the cache-miss rate and the resulting network load are proportional to each other. However, this may not be true for an erasure-coded storage, where cached chunks need not be identical to the transferred chunks. More precisely, a function of the data chunks can be computed and cached, so that the constructed new chunks, along with the existing chunks, also satisfy the property of being an MDS code. There have been caching schemes that cache the entire file~\citep{Nadgowda2014,Chang:2008:BDS:1365815.1365816,Zhu:2004:PSP:1006209.1006221}, while we can cache partial file  for an eraure-coded system (practically proposed for replicated storage systems in \citep{naik2015read}) which gives extra flexibility and the evaluation results depict the advantage of caching partial files.

A new {\em functional caching} approach called \textit{Sprout} that can efficiently capitalize on existing file coding in erasure-coded storage systems has been proposed in~\citep{Sprout,Yu-TON16}. In contrast to exact caching that stores $d$ chunks identical to original copies, our functional caching approach forms $d$ new data chunks, which together with the existing $n$ chunks satisfy the property of being an $(n+d,k)$ MDS code. Thus, the file can now be recovered from any $k$ out of $n+d$ chunks (rather than $k$ out of $n$ under exact caching), effectively extending coding redundancy, as well as system diversity for scheduling file access requests. The proposed functional caching approach saves latency due to more flexibility to obtain $k-d$ chunks from the storage system at a very minimal additional computational cost of creating the coded cached chunks.

\begin{figure}[!thbp]
\vspace{-2mm}
\begin{center}
\scalebox{0.2}{\fbox{\includegraphics[draft=false]{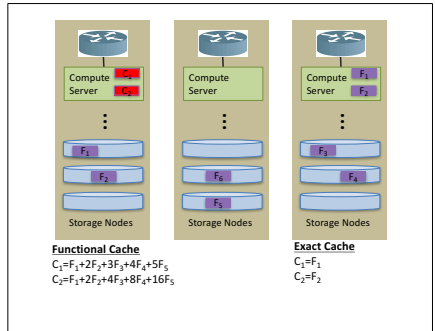}}}
\vspace{-3mm}
\caption{An illustration of functional caching and exact caching in an erasure-coded storage system with one file using a $(5,4)$ erasure code.}
\label{fig:conv}
\end{center}
\vspace{-.2in}
\end{figure}

\vspace{0.07in}
{\noindent \bf Example.} Consider a datacenter storing a single file using a $(5,4)$ MDS code. The file is split into $k_i=4$ chunks, denoted by $A_1,A_2,A_3,A_4$, and then linearly encoded to generate $n_i=5$ coded chunks $F_1=A_1$, $F_2=A_2$, $F_3=A_3$, $F_4=A_4$, and $F_5=A_1+A_2+A_3+A_4$ in a  finite field of order at-least 5. Two compute servers in the datacenter access this file and each is equipped with a cache of size $C=2$ chunks as depicted in Figure~\ref{fig:conv}. The compute server on the right employs an exact caching scheme and stores chunks $F_1,F_2$ in the cache memory. Thus, 2 out of 3 remaining chunks (i.e., $F_3$, $F_4$ or $F_5$) must be retrieved to access the file, whereas chunks $F_1,F_2$ and their host nodes will not be selected for scheduling requests. Under functional caching, the compute server on the left generates
$d_i=2$ new coded chunks, i.e., $C_1=A_1+2A_2+3A_3+4A_4$ and $C_2=4A_1+3A_2+2A_3+1A_4$, and saves them in its cache memory. It is easy to see that chunks $F_1,\ldots,F_5$ and $C_1,C_2$ now form a $(7,4)$ erasure code. Thus, the file can be retrieved by accessing $C_1,C_2$ in the cache together with any 2 out of 5 chunks from $F_1,\ldots,F_5$. This allows an optimal request scheduling mechanism to select the least busy chunks/nodes among all 5 possible candidates in the system, so that the service latency is determined by the best 2 storage node with minimum queuing delay. In contrast, service latency in exact caching is limited by the latency of accessing a smaller subset of chunks $F_3,F_4,$ and $F_5$. In order to have a $(n,k)$ coded file in the storage server, we can construct chunks by using an $(n+k,k)$ MDS code, where $n$ chunks are stored in the storage server. The remaining $k$ out of the $n+k$ coded chunks are assigned to be in part in cache based on the contents of the file in the cache. Thus, irrespective of the value of $d\le k$, we ascertain that $(n+d,k)$ code, formed with $n$ coded chunks in the storage server and $k$ coded chunks in the cache, will be MDS.

\begin{figure}[!thbp]
\vspace{-2mm}
\begin{center}
{\includegraphics[width=0.5\textwidth,draft=false]{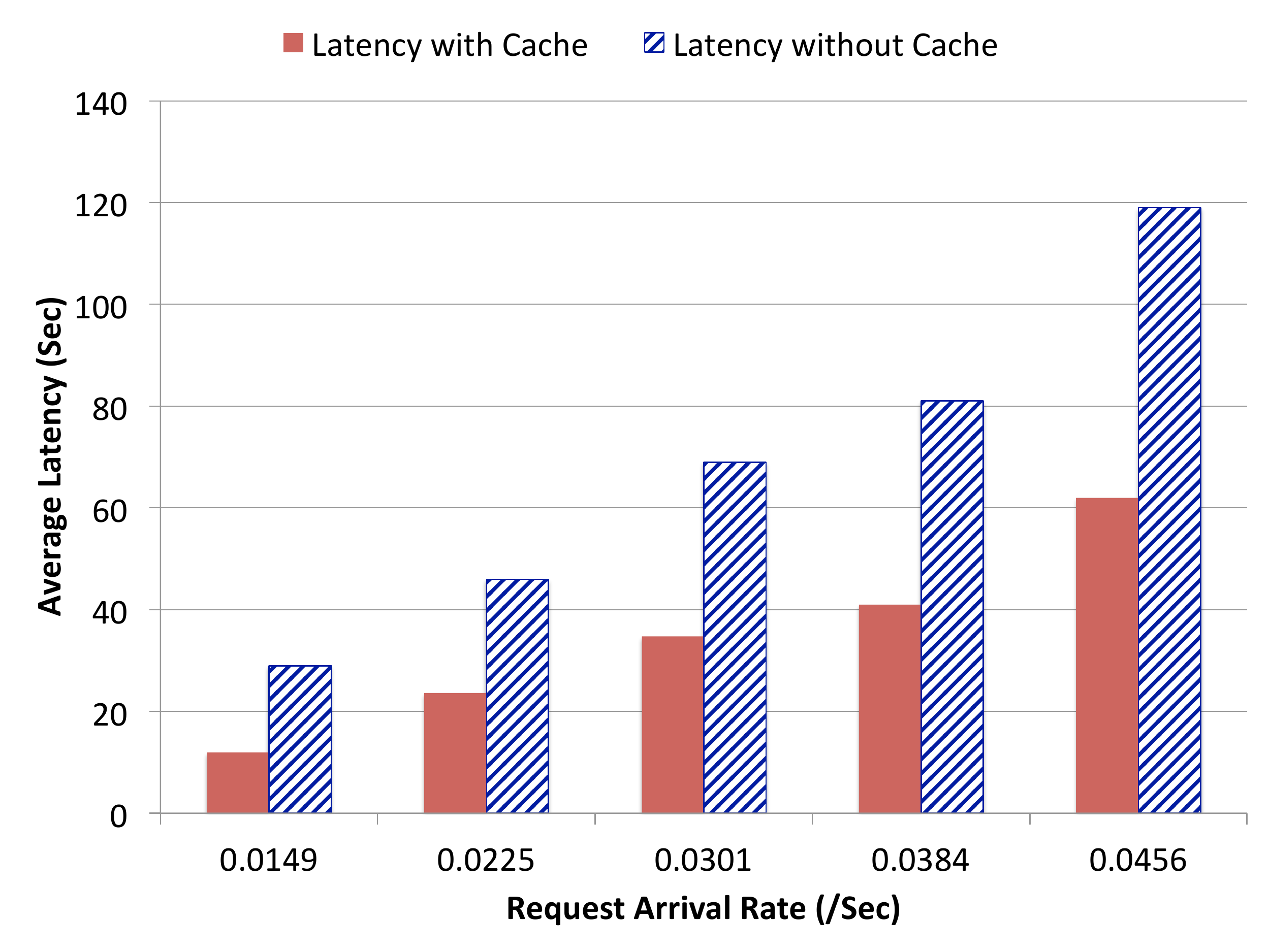}}%
\caption{Comparison of average latency of functional caching and Tahoe's native storage system without caching, with varying average arrival rates for $r=1000$ files of 200MB, where the cache size fixed at 2500.}
\label{fig:workload2}
\end{center}
\vspace{-.2in}
\end{figure}

Utilizing the theoretical models developed for quantifying service latency in erasure-coded storage systems such as \citep{Yu_TON}, we can obtain latency bound for functional caching and use it to formulate and solve a cache-content optimization problem. Due to space limitations, we skip the technical details and refer readers to~\citep{Yu-TON16}, where extension of Theorem \ref{old_upper_bound} is provided with functional caching. Implementing functional caching on the Tahoe testbed, we fix file size to be 200MB and vary the file request arrival rates with average request arrival rate of the $r=1000$ files in \{0.0149/sec, 0.0225/sec, 0.0301/sec, 0.0384/sec, 0.0456/sec\}. The  cache size is fixed, and set to be 2500. Actual average service latency of files for each request arrival rate is shown by a bar plot in Figure~\ref{fig:workload}. In this experiment we also compare the performance of functional caching with Tahoe’s built-in native storage system without caching. Fig \ref{fig:workload} shows that our algorithm with caching outperforms Tahoe native storage in terms of average latency for all the average arrival rates. Functional caching gives an average  49\% reduction in latency. Future work includes designing various cache replacement policies with respect to erasure codes, as well as developing a theoretical framework to quantify and optimize the resulting latency.

\backmatter 
\chapter*{Acknowledgements}

The authors would like to thank their collaborators for contributions to this line of work: Yih-Farn Robin Chen and Yu Xiang at AT\&T Labs-Research, Moo-Ryong Ra at Amazon, Vinay Vaishampayan at City University of NY, Ajay Badita and Parimal Parag at IISc Bangalore,  Abubakr Alabassi, Jingxian Fan, and Ciyuan Zhang at Purdue University, and Chao Tian at Texas A\&M University.

The authors would  like to thank Alexander  Barg at the University of Maryland for the many suggestions on the manuscript. The authors are also grateful to anonymous reviewers for valuable comments that have significantly improved the manuscript.
\printbibliography

\end{document}